\RequirePackage{fix-cm}
\documentclass[linesnumbered,ruled,vlined]{article}
\usepackage[utf8]{inputenc}
\usepackage{times}
\usepackage{amsmath}
\allowdisplaybreaks
\usepackage{amssymb} 
\usepackage{mathtools}
\usepackage{graphicx}
\graphicspath{ {images/} }
\usepackage[labelfont=bf]{caption}
\usepackage{subcaption}
\usepackage{setspace}
\usepackage[margin=1in]{geometry}
\usepackage{pdflscape}

\usepackage[unicode=true,
 bookmarks=false,
 breaklinks=true,pdfborder={0 0 0},pdfborderstyle={},backref=false,colorlinks=true,citecolor=blue, linkcolor=blue]
 {hyperref}

\usepackage{xcolor}
\usepackage{overpic}
\usepackage[all]{nowidow}
\usepackage{float}
\usepackage[ruled]{algorithm2e}

\usepackage{color}
\usepackage{verbatim}
\usepackage{amsthm}
\usepackage{undertilde}
\usepackage{wasysym}
\usepackage{esint}
\usepackage[authoryear]{natbib}

\theoremstyle{plain}
\newtheorem{assumption}{\protect\assumptionname}
\theoremstyle{plain}
\newtheorem{prop}{\protect\propositionname}

\theoremstyle{definition}

\theoremstyle{definition}

\theoremstyle{plain}
\newtheorem{lem}{\protect\lemmaname}

\theoremstyle{plain}
\newtheorem{thm}{\protect\theoremname}
\theoremstyle{remark}

\theoremstyle{plain}

\usepackage{verbatim}
\usepackage{array}
\usepackage{tabularx}
\usepackage{ragged2e}
\usepackage[T1]{fontenc}
\usepackage{xltabular}
\usepackage{booktabs}
\newcolumntype{P}[1]{>{\RaggedRight\arraybackslash}p{#1}}

\usepackage[labelfont={bf}, 
textfont={it},
margin=1cm,
font=small]{caption}
\usepackage[final]{microtype}

\usepackage{amsmath, amssymb}
\newcommand{\vertiii}[1]{{\left\vert\kern-0.05ex\left\vert\kern-0.05ex\left\vert #1 
		\right\vert\kern-0.05ex\right\vert\kern-0.05ex\right\vert}}

\DeclareMathOperator*{\argmin}{arg\,min}

\DeclareFontEncoding{FMS}{}{}
\DeclareFontSubstitution{FMS}{futm}{m}{n}
\DeclareFontEncoding{FMX}{}{}
\DeclareFontSubstitution{FMX}{futm}{m}{n}
\DeclareSymbolFont{fouriersymbols}{FMS}{futm}{m}{n}
\DeclareSymbolFont{fourierlargesymbols}{FMX}{futm}{m}{n}
\DeclareMathDelimiter{\VERT}{\mathord}{fouriersymbols}{152}{fourierlargesymbols}{147}

\usepackage{titling}
\usepackage{xpatch}

\usepackage{subfiles}

\LinesNumbered
\RestyleAlgo{algoruled} 

\providecommand{\corollaryname}{Corollary}
\providecommand{\definitionname}{Definition}
\providecommand{\examplename}{Example}
\providecommand{\lemmaname}{Lemma}
\providecommand{\remarkname}{Remark}
\providecommand{\theoremname}{Theorem}

\providecommand{\assumptionname}{Assumption}
\providecommand{\propositionname}{Proposition}

\definecolor{Red}{rgb}{1.0, 0.0, 0.0}   
\definecolor{Teal}{rgb}{0.0, 0.5, 0.5}  

\global\long\def\ddd{,\ldots,}%

\global\long\def\bB{\mathbf{B}}%

\global\long\def\bA{\mathbf{A}}%

\global\long\def\bS{\mathbf{S}}%

\global\long\def\bTheta{\boldsymbol{\Theta}}%

\global\long\def\bbeta{\boldsymbol{\beta}}%

\global\long\def\bdelta{\boldsymbol{\delta}}%

\global\long\def\hTheta{\widehat{\boldsymbol{\Theta}}}%

\global\long\def\bepsilon{\boldsymbol{\varepsilon}}%

\global\long\def\bSigma{\boldsymbol{\Sigma}}%

\global\long\def\bGamma{\boldsymbol{\Gamma}}%

\global\long\def\bDelta{\boldsymbol{\Delta}}%

\global\long\def\Tr{\mathrm{Tr}}%

\global\long\def\bY{\mathbf{Y}}%

\global\long\def\bM{\mathbf{M}}%

\global\long\def\bN{\mathbf{N}}%

\global\long\def\tr{\mathrm{Tr}}%

\global\long\def\bX{\mathbf{X}}%

\global\long\def\bW{\mathbf{W}}%

\global\long\def\bZ{\mathbf{Z}}%

\global\long\def\ba{\mathbf{a}}%

\global\long\def\bx{\mathbf{x}}%

\global\long\def\by{\mathbf{y}}%

\global\long\def\bz{\mathbf{z}}%

\global\long\def\bs{\mathbf{s}}%

\global\long\def\bm{\mathbf{m}}%

\global\long\def\R{\mathbb{R}}%

\global\long\def\E{\mathbb{E}}%

\global\long\def\Pr{\mathrm{Pr}}%

\global\long\def\sn{\sum_{i=1}^{N}}%

\global\long\def\ind{\mathbb{I}}%

\global\long\def\diag{\operatorname{diag}}%

\global\long\def\sgn{\operatorname{sgn}}%

\global\long\def\prox{\operatorname{prox}}%

\global\long\def\Proj{\operatorname{Proj}}%

\global\long\def\argmin{\operatorname*{arg\,min}}%

\global\long\def\card{\operatorname{card}}%

\global\long\def\vec{\operatorname{vec}}%

\begin{document}
\global\long\def\ddd{,\ldots,}%

\global\long\def\bB{\mathbf{B}}%

\global\long\def\bA{\mathbf{A}}%

\global\long\def\bS{\mathbf{S}}%

\global\long\def\hTheta{\widehat{\boldsymbol{\Theta}}}%

\global\long\def\bTheta{\boldsymbol{\Theta}}%

\global\long\def\bbeta{\boldsymbol{\beta}}%

\global\long\def\bdelta{\boldsymbol{\delta}}%

\global\long\def\bepsilon{\boldsymbol{\varepsilon}}%

\global\long\def\bSigma{\boldsymbol{\Sigma}}%

\global\long\def\bGamma{\boldsymbol{\Gamma}}%

\global\long\def\bDelta{\boldsymbol{\Delta}}%

\global\long\def\tr{\mathrm{Tr}}%

\global\long\def\bX{\mathbf{X}}%

\global\long\def\bY{\mathbf{Y}}%

\global\long\def\bW{\mathbf{W}}%

\global\long\def\bM{\mathbf{M}}%

\global\long\def\bN{\mathbf{N}}%

\global\long\def\bZ{\mathbf{Z}}%

\global\long\def\ba{\mathbf{a}}%

\global\long\def\bx{\mathbf{x}}%

\global\long\def\by{\mathbf{y}}%

\global\long\def\bz{\mathbf{z}}%

\global\long\def\bs{\mathbf{s}}%

\global\long\def\bm{\mathbf{m}}%

\global\long\def\R{\mathbb{R}}%

\global\long\def\E{\mathbb{E}}%

\global\long\def\Pr{\mathrm{Pr}}%

\global\long\def\sn{\sum_{i=1}^{N}}%

\global\long\def\ind{\mathbb{I}}%

\global\long\def\diag{\operatorname{diag}}%

\global\long\def\card{\operatorname{card}}%

\global\long\def\vec{\operatorname{vec}}%

\global\long\def\sgn{\operatorname{sgn}}%

\global\long\def\prox{\operatorname{prox}}%

\global\long\def\Proj{\operatorname{Proj}}%

\global\long\def\argmin{\operatorname*{arg\,min}}%

\title{Multivariate regression with missing response data for modelling regional DNA methylation QTLs}
\author{Shomoita Alam \thanks{Co-first author, Department of Mathematics and Statistics, McGill
University, Montreal, QC, Canada (shomoita.alam@mail.mcgill.ca)}, 
Yixiao Zeng \thanks{Co-first author, Lady Davis Institute for Medical Research, Jewish General Hospital, Montreal, QC, Canada (yixiao.zeng@mail.mcgill.ca)},
Sasha Bernatsky \thanks{Department of Medicine, McGill University, The Research Institute of the McGill University Health Centre, Montreal, Quebec, Canada (sasha.bernatsky@mcgill.ca)}, Marie Hudson \thanks{Lady Davis Institute for Medical Research, Jewish General Hospital, Department of Medicine, McGill University, Montreal, Quebec, Canada (marie.hudson@mcgill.ca)}, Inés Colmegna \thanks{Department of Medicine, McGill University, The Research Institute of the McGill University Health Centre, Montreal, Quebec, Canada (ines.colmegna@mcgill.ca)}, \\
David A. Stephens \thanks{Co-corresponding author, Department of Mathematics and Statistics, McGill
University, Montreal, QC, Canada (david.stephens@mcgill.ca)}, 
Celia M. T. Greenwood \thanks{Co-corresponding author, Lady Davis Institute for Medical Research, Jewish General Hospital, Montreal, QC, Canada (celia.greenwood@mcgill.ca)}, 
Archer Y. Yang \thanks{Co-corresponding author, Department of Mathematics and Statistics, McGill
University; Mila - Quebec AI institute, Montreal, QC, Canada (archer.yang@mcgill.ca)}}
\date{}
\maketitle

\begin{abstract}
Identifying genetic regulators of DNA methylation (mQTLs) with multivariate models enhances statistical power, but is challenged by missing data from bisulfite sequencing. Standard imputation-based methods can introduce bias, limiting reliable inference. We propose \texttt{missoNet}, a novel convex estimation framework that jointly estimates regression coefficients and the precision matrix from data with missing responses. By using unbiased surrogate estimators, our three-stage procedure avoids imputation while simultaneously performing variable selection and learning the conditional dependence structure among responses. We establish theoretical error bounds, and our simulations demonstrate that \texttt{missoNet} consistently outperforms existing methods in both prediction and sparsity recovery. In a real-world mQTL analysis of the CARTaGENE cohort, \texttt{missoNet} achieved superior predictive accuracy and false-discovery control on a held-out validation set, identifying known and credible novel genetic associations. The method offers a robust, efficient, and theoretically grounded tool for genomic analyses, and is available as an R package.
\end{abstract}

\section{Introduction}

DNA methylation in mammals involves adding a methyl group to a cytosine. When this occurs at  a CpG dinucleotide, methylation will occur on both DNA strands \citep{smith2013dna}. This epigenetic modification is crucial for regulating gene expression, either by inhibiting transcription factor binding or by recruiting chromatin-modifying proteins \citep{jones2012functions}. There is increasing evidence that perturbations in DNA methylation play a key role in how genetic and environmental factors contribute to complex diseases \citep{robertson2005dna}. Therefore, identifying the genetic determinants that influence DNA methylation -- referred to as methylation quantitative trait loci (mQTLs) -- is essential for understanding epigenetic regulation \citep{perzel2021genome}.

Most mQTL studies test each CpG site on its own, fitting a separate regression of methylation on genotype for every site \citep{hannon2018leveraging}. However, many CpG sites are clustered within specific genomic regions that may hold greater biological significance when examined collectively rather than in isolation \citep{jones2012functions}. For example, several CpGs in the alcohol-dehydrogenase gene (\textit{ADH4}) promoter share the same local genetic control \citep{zhang2014identification}, and dozens of correlated CpGs span the Major Histocompatibility Complex (MHC) on chromosome 6 under complex genetic regulation \citep{villicana2023genetic}. Treating a small neighbourhood of CpGs as a multivariate outcome, rather than running many separate tests, can boost statistical power and reveal coherent biological mechanisms \citep{gaunt2016systematic}. This article therefore adopts a region-based multivariate mQTL framework to take advantage of these spatial correlations.

\subsection{Motivation for our method}

We propose a multivariate regression model to analyze regional methylation data. In this analytical framework, methylation levels at multiple CpGs within a region are treated as multi-dimensional responses, and the corresponding cis-proximal SNPs as covariates. Typically, only a small subset of candidate SNPs in a regulatory region influences methylation. However, linkage disequilibrium—strong correlation among SNPs—can produce numerous spurious marginal associations. Our model adopts sparse penalization methods to select important SNPs that affect variations in methylation patterns across multiple correlated CpG sites, and simultaneously estimates their effects. Meanwhile, our method can model the partial correlation structure of the methylation levels at  multiple CpG sites, providing meaningful insights into the underlying epigenetic mechanisms. 

There has been extensive literature on sparse penalized multivariate regression methods to enable variable selection of covariates, as well as estimation of the conditional dependencies among the response variables. For example, \citet{yin2013adjusting} proposed a two-stage estimation procedure in a multivariate sub-Gaussian model to adjust for the covariate effects on the mean of the random vector, in order to obtain a more precise and interpretable estimate of the precision matrix. \citet{cai2013covariate} took a similar approach to estimate covariate-adjusted precision matrix, but without making a multivariate normal assumption on the error distribution. \citet{wang2015joint} proposed a method that decomposes the multivariate regression problem into a series of penalized conditional log-likelihood of each response conditional on the covariates and other responses.  

While these methods assume fully observed data, real-world DNA methylation datasets generated by sequencing techniques -- particularly bisulfite sequencing -- often contain missing values \citep{ singer2019practical}. Missing calls arise when reads have low depth or quality, when mapping is ambiguous, or when laboratories follow different library designs or protocols \citep{luo2023recall, ross2022batch}. For example, Bisulfite conversion -- the key step that distinguishes methylated from unmethylated cytosines -- can break DNA and lower coverage, especially in samples with limited or degraded input, leaving gaps at many CpGs \citep{olova2018comparison, simons2025comparative}. In addition, GC-rich promoters and repetitive elements are intrinsically hard to sequence; their complex structures cause systematic drop-outs that cluster within these loci \citep{pappalardo2021losing, ross2013characterizing}. 
Illumina BeadChip arrays report fewer missing values, but their probes are sparse and grouped in small clusters, so strong regional correlations are rare; for that reason, the present work focuses on sequencing-based datasets and does not analyse BeadChip data further \citep{illumina_methylation_guide}.

To address the missingness issue, ad hoc imputation methods have been proposed to fill in missing values prior to analysis \citep{little2019statistical, schafer1997analysis}, and likelihood-based methods using the EM algorithm can also be used in some situations \citep{stadler2012missing}. However, when the number of CpG sites greatly exceeds the number of samples, ad‑hoc imputation and standard EM‑based likelihood methods tend to become unstable and yield less reliable estimates \citep{wang2015high, brini2024missing}.

Much of the research on high-dimensional regression with missing data has focused on univariate outcomes, with particular attention to estimating the graphical structure of correlated predictors. For instance, \citet{stadler2012missing} introduced their Expectation-Maximization (EM)-based approach for estimating sparse inverse covariance matrices when covariates are missing and extended this method to sparse linear regression with missing predictors in the univariate outcome setting. In the errors-in-variables framework, missing data are modeled as multiplicative errors, and this approach has been extensively studied \citep{hwang1986multiplicative, Multiplicative_error}. \citet{loh2012} addressed high-dimensional univariate sparse regression with corrupted covariates by proposing unbiased surrogate estimators and solving the problem using a non-convex projected gradient descent algorithm, which converges to a near-global minimizer under certain conditions. However, this method relies on restrictive side constraints and assumes prior knowledge of an unknown constant, which can limit practical applicability. To address these limitations, \citet{datta2017} proposed the convex conditioned Lasso (CoCoLasso). Their method defines a nearest positive semi-definite matrix projection operator, ensuring the underlying optimization problem remains convex and results in superior performance compared to the non-convex approach. However, there is a notable gap in the literature regarding missing responses in multivariate regression settings, particularly when the responses are correlated.

\subsection{Summary of the proposed method}

In this paper, we propose a multivariate regression framework to jointly estimate the regression coefficients and precision matrix in the presence of missing responses. Unlike methods designed for complete data, our approach inherently handles missing data without requiring imputation, offering a distinct advantage. Our approach follows a three-stage estimation procedure, where each stage solves a convex optimization problem. In the first stage, we obtain an initial estimate of the regression coefficients under the assumption that the outcomes are independent. We leverage the unbiased surrogate estimator proposed by \citet{loh2012} to correct for the impact of missing data in the responses. Assuming independent outcomes, we solve the problem for each response individually using any suitable penalized least squares regression method.  In the second stage, we use the initial estimate of the regression coefficients as a plug-in to estimate the precision matrix. A key challenge here is that the plug-in covariance estimate may not be positive semi-definite when responses contain missing values, making the objective function non-convex or unbounded from below. To overcome this, we adopt the projected unbiased surrogate estimate from \citet{datta2017}, which ensures that the covariance matrix is positive semi-definite, transforming the problem into a convex one. This allows us to efficiently solve it using any existing graphical Lasso algorithm. In the third and final stage, we refine the estimate of the regression coefficients by incorporating the precision matrix obtained from the second stage. To efficiently perform this refinement, we developed a proximal gradient descent algorithm for estimating the regression coefficients.

On the theoretical side, we provide non-asymptotic bounds for the statistical errors at each stage of the algorithm. These bounds quantify how the statistical error at each stage is influenced by the degree of missingness, the error variance, the dimensions of the responses and covariates, and the sample size. To validate these theoretical guarantees, we conducted simulations that closely align with the predicted bounds. On the computational side, we developed an efficient proximal gradient descent algorithm and implemented it as a user-friendly \texttt{R} package, \texttt{missoNet}. We benchmarked its performance against several established methods, including multivariate regression with covariance estimation (\textit{MRCE}) \citep{rothman2010sparse}, separate Lasso with imputation, and conditional graphical Lasso with missing values (\textit{cglasso}) \citep{augugliaro2023cglasso}. Our method demonstrated strong performance across various simulation scenarios, consistently outperforming the competing approaches. \texttt{missoNet} also exhibited significant computational efficiency, achieving faster runtimes than other methods in our comparisons.

Most of our work uses targeted bisulfite sequencing reads from a subset of participants in  \href{www.cartagene.qc.ca}{the CARTaGENE Population Biobank}. That cohort has been widely described and gives access to rich health and genetic data \citep{awadalla2013cohort}. The original project asked whether methylation at differentially methylated cytosines (DMCs) was tied to anti-citrullinated-protein antibodies (ACPA), a key marker for rheumatoid arthritis \citep{shao2019rheumatoid, zeng2021thousands}. Here, we broaden the question: do nearby single nucleotide polymorphisms (SNPs) also shape those methylation patterns? 
To investigate this, we split the genome into small blocks of neighbouring CpG sites and tested each block against SNPs that fell inside the same window, focusing on \textit{cis} effects \citep{do2017genetic}. 
We ran our methodology \texttt{missoNet} and competing tools including single-site tests and existing multivariate mappers such as \texttt{MRCE} and \texttt{cglasso}. 
Simulations show that when we mixed true causal SNPs with synthetic variants, \texttt{missoNet} was better at telling the difference, showing higher power and lower false-positive rates than the other methods over a wide range of settings.
Moreover, most SNP–CpG links flagged by \texttt{missoNet} are already listed in large mQTL catalogues \citep{gaunt2016systematic}, confirming their credibility. In an independent validation set, SNPs selected by \texttt{missoNet} led to higher methylation prediction accuracy compared to alternative approaches.


The paper is organized as follows. In Section 2, we formulate the problem of jointly estimating the regression coefficients and the precision matrix for multivariate regression with missing data. Section 3 presents the three-stage estimation procedure, providing a detailed discussion of each step and addressing the associated challenges. In Section 4, we develop theoretical bounds for the statistical errors at each stage of the estimation process. Section 5 describes the simulation scenarios used to validate the theoretical results and reports comparative performance across different methods. Section 6 demonstrates the effectiveness of our method for modeling regional mQTLs. Finally, Section 7 summarizes the findings, highlights the advantages of our approach, and outlines directions for future work.

\subsection*{Notation:}

In this paper, we denote a random variable with an uppercase letter
and a data matrix with an uppercase bold letter, for example, $A$ is
a random variable, and $\mathbf{A}$ is a data matrix. For a matrix $\mathbf{A}$, we denote by $\mathbf{A}\succeq0$ when
$\mathbf{A}$ is positive semi-definite. Let $\VERT\mathbf{A}\VERT_{1}$
be the operator norm induced by $\ell_{1}$ norm for vectors, which
can be computed by $\VERT\mathbf{A}\VERT_{1}=\max_{j}\sum_{i}\mid a_{ij}\mid$,
i.e., the maximum absolute column sum of the matrix. Denoted by $\VERT\mathbf{A}\VERT_{2}$
the operator norm that can be computed as the greatest singular value
of $\mathbf{A}$, i.e., $\VERT\mathbf{A}\VERT_{2}=\max_{j}\sigma_{j}(\mathbf{A})$.
Let $\VERT\mathbf{A}\VERT_{\infty}$ be the operator norm induced
by $\ell_{\infty}$ norm, which can be computed by $\VERT\mathbf{A}\VERT_{\infty}=\max_{i}\sum_{j}\mid a_{ij}\mid$,
i.e., the maximum absolute row sum of the matrix. Let $\VERT\mathbf{A}\VERT_{1,1}=\sum_{i,j}\mid a_{ij}\mid$
be the elementwise $\ell_{1}$-norm, $\VERT\mathbf{A}\VERT_{F}=\sqrt{\sum_{i,j}\mid a_{ij}\mid^{2}}$
be the Frobenius norm, and $\VERT\mathbf{A}\VERT_{\max}=\max_{i,j}\mid a_{ij}\mid$
be the elementwise maximum norm. Let $\Lambda_{\min}(\mathbf{A})$
and $\Lambda_{\max}(\mathbf{A})$ denote the smallest and largest
eigenvalues of $\mathbf{A}$. 

For two matrices $\mathbf{A}=(a_{ij})$ and $\mathbf{B}=(b_{ij})$,
 define $\mathbf{A}\odot\mathbf{B}=(a_{ij}b_{ij})$ as their elementwise
product, and $\mathbf{A}\oslash\mathbf{B}=(a_{ij}/b_{ij})$ as their
elementwise division. We denote the $(k,l)$th element of a matrix $\bM$ by
$(\bM)_{kl}$, its $k$th row by $(\bM)_{k\bullet}$ and $l$th column
by $(\bM)_{\bullet l}$.  For any integer $a$,
we define $[a]\coloneqq \left\{ 1,\ldots,a\right\} $ as a set of variable indices. For
instance, for vector $\mathbf{a}\in\mathbb{R}^{p}$ and index set $I\subset[p]$, we use $\mathbf{a}_{I}$
to denote the same vector as $\mathbf{a}$ but with elements $[p]\backslash I$
set to zero. For matrix $\mathbf{A}\in \mathbb{R}^{p\times q}$ and  index set $I\subset[p]\times[q]$, we define $\mathbf{A}_{I}$ similarly, i.e., $(\mathbf{A}_I)_{ij}=(\mathbf{A})_{ij}$ for $(i,j)\in I$ and $(\mathbf{A}_I)_{ij}=0$ for $(i,j)\in ([p]\times[q]) \backslash I$.

Let $\mathbb{Q}$ be a generic Euclidean space. Let $\bM$ and $\bN$
be any conformable matrices or vectors in $\mathbb{Q}$. Let us define
the inner product $\langle\cdot,\cdot\rangle$ as $\langle\bM,\bN\rangle=\tr(\bM^{\top}\bN)$.
For a norm $\mathcal{R}$ defined on $\mathbb{Q}$, the dual norm
$\mathcal{R}^{*}$ can be defined by 
\[
\mathcal{R}^{*}(\bM)\equiv\sup_{\bN\in\mathbb{Q}\backslash\left\{ 0\right\} }\frac{\langle\bM,\bN\rangle}{\mathcal{R}(\bN)}.
\]
Let 
$S_{l}(\mathbf{B}^{*})\coloneqq\{k:\beta_{kl}^{*}\neq0\ \text{for }k\in(1,\ldots,p)\}$ be the support of the $l$th column of matrix $\mathbf{B}^{*}$
for $l=1\ddd q$. The corresponding cardinality is $s_{l}\coloneqq|S_{l}|$. 

\section{Methods}

Consider a $q$-dimensional random vector $Y \coloneqq (Y_{1}, \dots, Y_{q})^{\top} \in \mathbb{R}^{q}$ and a $p$-dimensional covariate vector $X \coloneqq (X_{1}, \dots, X_{p})^{\top} \in \mathbb{R}^{p}$. Assuming that the underlying random vector $Y$ follows a multivariate Gaussian distribution, we can model a linear relationship between $Y$ and $X$ as
\begin{equation}
    Y = \bB^{* \top} X + \varepsilon \label{eq:model1}
\end{equation}
where $\mathbf{B}^{*} \in \mathbb{R}^{p \times q}$ is the matrix of true regression coefficients, and $\varepsilon \in \mathbb{R}^{q}$ is the random error vector following a multivariate Gaussian distribution $N(\mathbf{0}, \bSigma_{\varepsilon\varepsilon}^{*})$, with $\bSigma_{\varepsilon\varepsilon}^{*} \in \mathbb{R}^{q \times q}$. Assuming $Y$ and $X$ are centered, the intercept term is omitted.

The model in \eqref{eq:model1} implies that the regression function has the form $\mathbb{E}(Y|X) = \bB^{* \top} X$ and that the conditional covariance $\mathrm{Cov}(Y|X) = \bSigma_{\varepsilon\varepsilon}^{*}$. Typically, the goal is to estimate the coefficient matrix $\bB^{*}$ and the precision matrix $\bTheta_{\varepsilon\varepsilon}^{*} \coloneqq (\bSigma^{*}_{\varepsilon\varepsilon})^{-1}$ using $n$ i.i.d. observations of $(Y, X)$.

Under the multivariate normal assumption, the precision matrix $\bTheta_{\varepsilon\varepsilon}^{*}$ represents a conditional Gaussian graphical model \citep{lauritzen1996graphical, friedman2008sparse, yuan2007model}, where a zero off-diagonal element indicates conditional independence among the responses. Since $\bTheta_{\varepsilon\varepsilon}^{*}$ captures the conditional dependencies among the response variables, it provides a highly interpretable network structure.

When covariates are high-dimensional, i.e., $p\gg n$, the estimation problem becomes challenging. For the case where all the errors $\varepsilon$'s in  \eqref{eq:model1} are uncorrelated, i.e., all off-diagonal elements of $\bSigma_{\varepsilon\varepsilon}^{*}$ are zeros, \citet{yuan2007dimension} proposed to impose a constraint that encourages sparsity on the singular values of $\widehat{\bB}$ and hence reducing the rank of $\widehat{\bB}$. There are also some other existing methods for the uncorrelated error case using certain group-wise regularization approaches \citep{turlach2005simultaneous, peng2010regularized, obozinski2011support}.

\citet{rothman2010sparse} shows that the estimation accuracy of $\bB^{*}$ can be improved by considering the dependence among outcomes. Under the correlated outcomes case, one needs to estimate both $\bB^{*}$ and $\bTheta_{\varepsilon\varepsilon}^{*}$. They proposed the sparse multivariate regression with covariance estimation (MRCE) to jointly estimate $\mathbf{B}^{*}$ and $\bTheta_{\varepsilon\varepsilon}^{*}$ by minimizing the negative log-likelihood with $\ell_{1}$ penalization as follows: 
\begin{equation}
(\widehat{\bTheta},\widehat{\bB})= \argmin_{\bTheta\succeq0,\bB} \tr\left[\frac{1}{2n}(\bY-\bX\bB)^{\top}(\bY-\bX\bB)\bTheta\right]-\frac{1}{2}\log\det(\bTheta)+\lambda_{\bTheta}\VERT\bTheta\VERT_{1,\mathrm{off}}+\lambda_{\bB}\VERT\bB\VERT_{1,1} \label{eq:mrcesol}
\end{equation}
with $\VERT\bTheta\VERT_{1,\mathrm{off}}=\sum_{j\neq j'}|\Theta_{jj'}|$,
$\VERT\mathbf{B}\VERT_{1,1}=\sum_{j,k}|B_{jk}|$, and $\lambda_{\bTheta},\lambda_{\bB}\geq0$
are tuning parameters controlling the sparsity in $\widehat{\bTheta}$
and $\widehat{\bB}$, respectively.

Suppose we have $n$ independent and identically distributed observations
from some joint distribution $(Y,X)$.
In matrix notation, we can rewrite \eqref{eq:model1} as a model of
$n$ stacked observations
\begin{equation}
	\bY=\bX\mathbf{B}^{*}+\boldsymbol{\varepsilon},\label{eq:main_model}
\end{equation}
where $\bX=(\bx_{1},\ldots,\bx_{n})^{\top}\in\mathbb{R}^{n\times p}$ with $\bx_{i}=(x_{i1},\ldots,x_{ip})^{\top}\in\mathbb{R}^{p}$ and $\bY=(\by_{1},\ldots,\by_{n})^{\top}\in\mathbb{R}^{n\times q}$ with $\mathbf{y}_{i}=(y_{i1},\ldots,y_{iq})^{\top}\in\mathbb{R}^{q}$ denote the data matrices for the predictors and responses, respectively, and $\boldsymbol{\varepsilon}=(\bepsilon_{1},\ldots,\bepsilon_{n})^{\top}\in\mathbb{R}^{n\times q}$ with $\bepsilon_{i}=(\varepsilon_{i1},\ldots,\varepsilon_{iq})^{\top}\in\mathbb{R}^{q}$ denotes the matrix of random noises. We assume $\mathbb{E}(\boldsymbol{\varepsilon}|\bX)=\mathbf{0}$ and also $\mathrm{Cov}(\boldsymbol{\varepsilon}|\bX)=\bSigma_{\varepsilon\varepsilon}^{*}$. The true regression coefficient is $\mathbf{B}^{*}\in\mathbb{R}^{p\times q}$ with the element in the $k$th row and the $j$th column being $\beta_{kj}^{*}$. We assume that the design matrix $\bX$ has normalized columns, that is, $(1/n)\sum_{i=1}^{n}x_{ij}^{2}=1$ for every $j=1,\ldots,p$.

The goal is to estimate the coefficient matrix $\bB^{*}$ and the precision matrix $\bTheta_{\varepsilon\varepsilon}^{*}\coloneqq(\bSigma_{\varepsilon\varepsilon}^{*})^{-1}$ using observations $(\bY,\bX)$.  When there are missing values in the response data $\bY$, one cannot directly estimate $\bTheta_{\varepsilon\varepsilon}^{*}$ and $\mathbf{B}^{*}$ using the MRCE method as in \eqref{eq:mrcesol}. To see this, we rewrite the objective function in \eqref{eq:mrcesol} as 
\begin{align}
F(\bTheta,\bB) \coloneqq \Tr\left[\frac{1}{2}(\mathbf{S}_{yy}-2{\bS}_{xy}^{\top}\bB+\mathbf{B}^{\top}\mathbf{S}_{xx}\mathbf{B})\bTheta\right]-\frac{1}{2}\log\det(\bTheta)+\lambda_{\bTheta}\VERT\bTheta\VERT_{1,\mathrm{off}}+\lambda_{B}\VERT\mathbf{B}\VERT_{1,1},\label{eq:mrceexpan}
\end{align}
where $\mathbf{S}_{yy}=\frac{1}{n}\bY^{\top}\bY$, $\mathbf{S}_{xy}=\frac{1}{n}\bX^{\top}\bY$
and $\mathbf{S}_{xx}=\frac{1}{n}\bX^{\top}\bX$ are empirical covariance
matrices, which are unbiased estimators of $\bSigma_{yy}^{*}=\mathbb{E}[YY^{\top}]$,
$\bSigma_{xy}^{*}=\mathbb{E}[XY^{\top}]$ and $\bSigma_{xx}^{*}=\mathbb{E}[XX^{\top}]$,
respectively. Since the data matrix $\bY$ is not fully observed, $\mathbf{S}_{yy}$
and $\mathbf{S}_{xy}$ based on the non-missing observations might be biased in \eqref{eq:mrceexpan}. In a fully observed and noiseless setting, the sample covariance is at least guaranteed to be positive semi-definite. However, with missing data, the sample covariance matrix may not be positive semi-definite. In fact, in a high-dimensional setting ($p > n$), the sample covariance is guaranteed to have negative eigenvalues and therefore, be indefinite. Moreover, the objective function is unbounded from below when the sample covariance has a negative eigenvalue, so an optimum may not necessarily exist.

Let $\bY=(y_{ij})\in\mathbb{R}^{n\times q}$ be the underlying unobserved true response matrix, and $\bZ=(z_{ij})\in\mathbb{R}^{n\times q}$ be the  observed response matrix with missingness. We consider a missing completely at random case, where we observe $z_{ij}=y_{ij}$ with probability $1-\rho_{j}$ and zero
otherwise, i.e., each element of $\bY$ in the $j$th column has probability $\rho_{j}$ of being missing. This can be modeled by assuming $\bZ=\bY\odot\bW$, where matrix $\bW\in\mathbb{R}^{n\times q}$
introduces missingness to $\bY$. Let $w_{ij}$ be the $(i,j)$th element of $\bW$ which is set to be a Bernoulli $(1-\rho_{j})$ random
variable
\begin{align}
w_{ij}=\left\{ \begin{array}{ll}
1 & \text{ with probability }1-\rho_{j},\\
0 & \text{ otherwise. }
\end{array}\right.
\end{align}
for $j=1,...,q$. Therefore $w_{i\bullet}$, the $i$th row  of $\bW$, is drawn independently
and identically from the multivariate distribution 
 of random vector $W$ with the first and the second-order expectations $\boldsymbol{\mu}_{W}=\mathbb{E}[W]\in\mathbb{R}^{q}$
and $\mathbb{E}[W{W}^{\top}]\in\mathbb{R}^{q\times q}$, respectively. Also assume
that $\bW$ is independent of the random
error matrix $\boldsymbol{\varepsilon}$ in \eqref{eq:main_model}, which have the following specific forms:
\begin{align}
\mathbb{E}[W]=[(1-\rho_{1}),\cdots,(1-\rho_{q})]^{\top},\ \ \mathbb{E}[W{W}^{\top}]_{ij}=\left\{ \begin{array}{ll}
(1-\rho_{i})(1-\rho_{j}) & \text{ if }i\neq j,\\
(1-\rho_{i}) & \text{ if }i=j.
\end{array}\right.
\end{align}
The problem reduces to the standard MRCE model \citep{rothman2010sparse}
for fully observed data when $\boldsymbol{\rho}=(\rho_{1},\ldots,\rho_{q})^{\top}=\mathbf{0}$.
However, in practice, $\boldsymbol{\rho}$ may not be known and must
be estimated empirically from the data. We can estimate $\rho_{j}$ using $\widehat{\rho}_{j}$, where $\widehat{\rho}_{j}$
is the empirical missing probability of the $j$th column. 

Due to missingness in responses, the empirical covariance matrices $\mathbf{S}_{zz}=\frac{1}{n}\bZ^{\top}\bZ$, $\mathbf{S}_{xz}=\frac{1}{n}\bX^{\top}\bZ$ cannot be directly used for replacing $\mathbf{S}_{yy}$ and $\mathbf{S}_{xy}$ in $\eqref{eq:mrceexpan}$, as they are biased for estimating $\bSigma_{yy}^{*}=\mathbb{E}[YY^{\top}]$ 
and $\bSigma_{xy}^{*}=\mathbb{E}[XY^{\top}]$, respectively. 
To address this problem, one can use the surrogate estimators  \citep{loh2012}
\begin{align}
\widehat{\bS}_{yy}\coloneqq\frac{1}{n}{\bZ}^{\top}\bZ\oslash\mathbb{E}[W{W}^{\top}],\ \ \widehat{\bS}_{xy}\coloneqq\frac{1}{n}\bX^{\top}\bZ\oslash[\mathbb{E}[W],...,\mathbb{E}[W]]^{\top}.\label{eq:surrog}
\end{align}
These are unbiased estimators for $\bSigma_{yy}^{*}=\mathbb{E}[YY^{\top}]$ and $\bSigma_{xy}^{*}=\mathbb{E}[XY^{\top}]$ under conditions of response missingness. In this paper, we propose to estimate $\bTheta_{\varepsilon\varepsilon}^{*}$ and $\mathbf{B}^{*}$ by minimizing the following objective function with respect to $\bTheta\succeq0$ and $\bB$.
\begin{equation}
Q(\bTheta,\bB)\coloneqq\tr\Big[\frac{1}{2}(\widehat{\bS}_{yy}-2\widehat{\bS}_{xy}^{\top}\bB+\mathbf{B}^{\top}\mathbf{S}_{xx}\mathbf{B})\bTheta\Big]-\frac{1}{2}\log\det(\bTheta)+\lambda_{\bTheta}\VERT\bTheta\VERT_{1,\mathrm{off}}+\lambda_{\bB}\VERT\mathbf{B}\VERT_{1,1}.\label{eq:surrogate_obj}
\end{equation}
Here, compared to \eqref{eq:mrceexpan}, we replace the unobserved covariance matrices $\mathbf{S}_{yy}$
and $\mathbf{S}_{xy}$ with $\widehat{\bS}_{yy}$ and $\widehat{\bS}_{xy}$ in \eqref{eq:surrogate_obj}. The matrix $\mathbf{S}_{xx}$ remains unchanged as the covariate matrix $\bX$ is fully observed.

\section{Estimation } \label{P2:estimation}

We develop an iterative three-stage procedure for solving \eqref{eq:surrogate_obj}. Specifically, in the first stage, we optimize the objective function $Q$ with respect to $\bB$  while setting $\bTheta=\mathbf{I}$; In the second stage, we optimize $Q$ with respect to $\bTheta$ given $\bB=\widehat{\bB}^{(1)}$ from the first stage estimation; In the final stage, we obtain a refined estimate $\widehat{\bB}^{(2)}$ by incorporating the estimated precision matrix $\widehat{\bTheta}_{\varepsilon\varepsilon}$ from the second stage into the optimization of $Q$ with respect to $\bB$.

\begin{align*}
\text{Stage I: }&\widehat{\bB}^{(1)}=  \argmin_{\bB}
 Q(\mathbf{I},\bB) \\
\text{Stage II: }& \widehat{\bTheta}_{\varepsilon\varepsilon}=  \argmin_{\bTheta\succeq0}
 Q(\bTheta,\widehat{\bB}^{(1)})\\
\text{Stage III: }& \widehat{\bB}^{(2)}=   \argmin_{\bB}Q(\widehat{\bTheta}_{\varepsilon\varepsilon},\bB) 
\label{eq:obj}
\end{align*}


The details of the procedure are described in the remainder of this section.
\subsection{Stage I} 

In the first stage, we set $\bTheta=\mathbf{I}$ in \eqref{eq:surrogate_obj} and minimize the objective function $Q$ with respect to $\bB$
\begin{equation}	\widehat{\bB}^{(1)}=\argmin_{\bB}\ \mathrm{\tr}(\mathbf{B}^{\top}\mathbf{S}_{xx}\mathbf{B}/2-\widehat{\bS}_{xy}^{\top}\bB)+\lambda_{\bB}\VERT\mathbf{B}\VERT_{1,1},\label{1stB}
\end{equation}
where $\widehat{\bS}_{xy}$ is an unbiased estimator
of $\mathbb{E}[XY^{\top}]$ defined in \eqref{eq:surrog}. 
This loss function is convex because $\bX$ is fully observable (or uncontaminated) and therefore $\bS_{xx}$ is positive semi-definite. Define $\widehat{\bB}^{(1)}=(\widehat{\bbeta}_{1}\ddd\widehat{\bbeta}_{q}) \in \mathbb{R}^{p \times q}$
with its $l$th column as $\widehat{\bbeta}_{l}=(\widehat{\beta}_{l1}\ddd\widehat{\beta}_{lp})^{\top}\in\R^{p}$.
Each column of $\widehat{\bbeta}_{l}$ can be computed separately, by solving  penalized least squares problems with univariate response in a column-by-column fashion, i.e., for each $l=1\ddd q$,
\begin{equation}
	\widehat{\bbeta}_{l}=\argmin_{\bbeta\in\mathbb{R}^{p}}\mathcal{L}_{l}(\bbeta)+\lambda_{l}\|\bbeta\|_1,\qquad \text{where }\mathcal{L}_{l}(\bbeta)=\bbeta^{\top}\bS_{xx}\bbeta/2-\left\{ (\widehat{\bS}_{xy})_{\bullet l}\right\} ^{\top}\bbeta
 \label{eq:Beta_obj}
\end{equation}
with $(\widehat{\bS}_{xy})_{\bullet l}$ being $l$th column of
$\widehat{\bS}_{xy}$.
Hence each sub-problem can be solved very efficiently by using
any standard Lasso algorithm  \citep{tibshirani1996regression}, e.g. coordinate descent \citep{friedman2010regularization} or proximal gradient descent algorithm \citep{parikh2014proximal}. 

\subsection{Stage II}

Given $\widehat{\bB}^{(1)}$ from the first stage estimation, it is
tempting to estimate $\bTheta_{\varepsilon\varepsilon}^{*}$ by the
solution to the graphical Lasso problem 
\begin{equation}
\min_{\bTheta\succeq\mathbf{0}}\left\{ \tr(\bTheta\widehat{\bS}_{\varepsilon\varepsilon})-\log\det(\bTheta)+\lambda_{\bTheta}\VERT\bTheta\VERT_{1,\mathrm{off}}\right\} ,\label{eq:glasso}
\end{equation}
where $\widehat{\bS}_{\varepsilon\varepsilon}\coloneqq\widehat{\bS}_{yy}-\widehat{\bB}^{(1)\top}\bS_{xx}\widehat{\bB}^{(1)}$
is a plug-in estimate of the error covariance  $\bSigma_{\varepsilon\varepsilon}^{*}$ (since $\bSigma_{yx}^{*} = \bB^{*\top}\bSigma_{xx}^{*}$), in which $\widehat{\bS}_{yy}$
is an unbiased surrogate estimator of $\mathbb{E}[YY^{\top}]$. The details of the derivation is provided in \ref{deriv_See}. This approach is expected to improve the regression coefficient estimates in the third stage by incorporating the precision matrix's dependence structure.

However, the objective function \eqref{eq:glasso} is no longer convex
since the input estimator of the covariance matrix $\widehat{\bS}_{\varepsilon\varepsilon}$
may not be positive semi-definite in the presence of missing data. We can illustrate this phenomenon with a specific
example. Suppose that all columns of responses have the same probability
$\rho\in[0,1]$ of being missing, hence $\widehat{\bS}_{yy}=\bZ^{\top}\bZ/n-\rho\text{diag}(\bZ^{\top}\bZ/n)$.
When $n\ll q$, $\bZ^{\top}\bZ/n$ has the
rank at most $n$. Subtracting the diagonal matrix $\rho\text{diag}(\bZ^{\top}\bZ/n)$
from $\bZ^{\top}\bZ/n$ causes $\widehat{\bS}_{yy}$
to have negative eigenvalues even if under moderate missingness. This
leads to an indefinite $\widehat{\bS}_{\varepsilon\varepsilon}$
that renders the objective function in \eqref{eq:glasso} non-convex
and unbounded from below.

To solve this problem, we can approximate $\widehat{\bS}_{\varepsilon\varepsilon}$
using a positive semi-definite matrix $\widetilde{\bS}_{\varepsilon\varepsilon}$
\citep{datta2017}. Specifically, we find the nearest positive semi-definite
matrix projection for $\widehat{\bS}_{\varepsilon\varepsilon}$ as
\begin{align}
\widetilde{\bS}_{\varepsilon\varepsilon}\coloneqq\underset{\mathbf{K}\succeq0}{\mathrm{argmin}}\ \VERT\widehat{\bS}_{\varepsilon\varepsilon}-\mathbf{K}\VERT_{\text{max}},\label{eq:Projected_Syy}
\end{align}
where $\mathbf{K}$ is a positive semi-definite matrix and $\VERT\cdot\VERT_{\text{max}}$
is the elementwise maximum norm. Using $\widetilde{\bS}_{\varepsilon\varepsilon}$ instead, we can estimate $\widehat{\bTheta}_{\varepsilon\varepsilon}$ by minimizing
the following objective function 
\begin{align}
\widehat{\bTheta}_{\varepsilon\varepsilon}=\argmin_{\bTheta\succeq\mathbf{0}}\left\{ \tr(\bTheta\widetilde{\bS}_{\varepsilon\varepsilon})-\log\det(\bTheta)+\lambda_{\bTheta}\VERT\bTheta\VERT_{1,\mathrm{off}}\right\} .\label{theta2}
\end{align}
The problem in \eqref{theta2} can be solved using any existing graphical Lasso
algorithm, such as GLASSO \citep{friedman2008sparse} or QUIC \citep{hsieh2014quic}.

 \subsection{Stage III}

In the final stage of estimation, we use the estimated precision matrix, $\widehat{\bTheta}_{\varepsilon\varepsilon}$, from the previous stage to obtain a refined estimate, $\widehat{\bB}^{(2)}$. This involves minimizing the following objective function with $\widehat{\bTheta}_{\varepsilon\varepsilon}$ as a fixed plug-in estimate:

\begin{equation}
	\widehat{\bB}^{(2)}=\argmin_{\bB}\ \mathrm{\tr}[(\mathbf{B}^{\top}\mathbf{S}_{xx}\mathbf{B}/2-\widehat{\bS}_{xy}^{\top}\bB)\widehat{\bTheta}_{\varepsilon\varepsilon}]+\lambda_{\bB}\VERT\mathbf{B}\VERT_{1,1}.\label{eq:Bobj}
\end{equation}

The problem \eqref{eq:Bobj} can be solved using a proximal gradient descent algorithm \citep{parikh2014proximal}. Let $\widehat{\bB}^{(2)}_{k}$ represent the $k$-th iteration update of $\widehat{\bB}^{(2)}$. For $k = 1, 2, \dots, K$, we apply the fast iterative shrinkage-thresholding algorithm (FISTA, \cite{beck2009fast}) for updates:
\begin{align}
\mathbf{V}_{k} & = \widehat{\mathbf{B}}^{(2)}_{k-1} + \frac{k-2}{k+1}(\widehat{\mathbf{B}}^{(2)}_{k-1} - \widehat{\mathbf{B}}^{(2)}_{k-2}),\nonumber\\
\widehat{\mathbf{B}}^{(2)}_{k} & = \text{prox}_{t_{k}}(\mathbf{V}_{k} - 2t_{k}(\mathbf{S}_{xx}\mathbf{V}_{k}-\widehat{\mathbf{S}}_{xy}) \widehat{\mathbf{\Theta}}_{\varepsilon\varepsilon}),
\label{FISTA}
\end{align}
starting with $\widehat{\mathbf{B}}^{(2)}_{0} = \widehat{\mathbf{B}}^{(2)}_{1} = \widehat{\mathbf{B}}^{(1)}$. Here, $t_k$ is the step size at iteration $k$, and $\text{prox}_{t_k}(\cdot) :\mathbb{R}^{p \times q} \rightarrow \mathbb{R}^{p \times q}$ represents the proximal operator, which simplifies to enforce $\ell_{1}$ sparsity by soft-thresholding each matrix element towards zero \citep{parikh2014proximal}:
\begin{align}
\text{prox}_{t_k}(v_{ij}) = \text{sign}(v_{ij})\cdot(|v_{ij}| - \lambda_\mathbf{B}t_k)_{+}, \nonumber
\end{align}
where $(v)_{+} = \text{max}(v,0)$ denotes the positive part of $v$. The FISTA method incorporates a momentum term similar to Nesterov's acceleration \citep{nesterov2013gradient}, which speeds up convergence from $\mathcal{O}(1/k)$ to $\mathcal{O}(1/k^2)$ compared to the standard ISTA \citep{beck2009fast}.

Choosing an appropriate step size $t_k$ is critical for effective convergence. Adaptive techniques like backtracking line search can be employed to determine $t_k$ at each iteration \citep{civicioglu2013backtracking}. The backtracking line search begins with an initial step size $t_k > 0$, reducing it by a factor $\eta$ (where $0 < \eta < 1$) until the Armijo condition \citep{armijo1966minimization} is met:
\begin{align}
    f(\widehat{\mathbf{B}}_{k-1}- t_{k} G_{t_{k}}(\widehat{\mathbf{B}}_{k-1})) \leq f(\widehat{\mathbf{B}}_{k-1})
    -t_{k} \langle \nabla f(\widehat{\mathbf{B}}_{k-1}), G_{t_{k}}(\widehat{\mathbf{B}}_{k-1}) \rangle
    +\frac{t_{k}}{2}\VERT G_{t_{k}}(\widehat{\mathbf{B}}_{k-1})\VERT^2_F,
\label{Armijo}
\end{align}
where $G_t(\mathbf{B}) = \frac{\mathbf{B}-\text{prox}_{t}(\mathbf{B} - t \nabla f(\mathbf{B}))}{t}$ represents the effective gradient adjusted for the proximal mapping. The term $\langle \mathbf{A}, \mathbf{C} \rangle := \mathrm{Tr}(\mathbf{C}^{\top} \mathbf{A})$ is the trace inner product, and $\VERT \cdot \VERT_F^2$ represents the squared Frobenius norm. Details of the accelerated proximal gradient method with backtracking line search are provided in Algorithm \ref{alg_FISTA} in \ref{Comp_details}. The computational details of implementing the \texttt{missoNet} \texttt{R}-package are outlined in  \ref{missonet_comp}.

\subsection{Tuning parameter selection}

The tuning parameters $\lambda_{\bB^{(1)}}$, $\lambda_{\bTheta}$ and $\lambda_{\bB^{(2)}}$ in the three-stage estimation procedure need to be tuned.  We use the Bayesian Information Criterion (BIC) to choose these parameters, favoring this criterion for its computational efficiency compared to cross-validation and its tendency to select more parsimonious models. The joint BIC \citep{yin2011sparse} can be defined as
\[\mathrm{BIC}(\widehat{\bB}, \widehat{\bTheta}_{\varepsilon\varepsilon} )=2 n  \tr  [\widehat{\bTheta}_{\varepsilon\varepsilon} (\widehat{\bB}^{\top}\bS_{xx}\widehat{\bB}/2-\widehat{\bS}_{xy}^{\top}\widehat{\bB})] -n \log\det(\widehat{\bTheta}_{\varepsilon\varepsilon}) + \log n \{q + \sum_{i<j} \mathbb{I}((\widehat{\bTheta}_{\varepsilon\varepsilon})_{i,j}\neq 0) + \sum_{i,j} \mathbb{I}(\widehat{\bB}_{i,j}\neq 0) \}.\]

Specifically, for the first stage, we fix $\bTheta=\mathbf{I}$ and use the following BIC formula to select $\lambda_{\bB^{(1)}}$:
\begin{align*}
\lambda_{\bB^{(1)}} &= \argmin_{\lambda}\mathrm{BIC}(\widehat{\bB}^{(1)}[\lambda], \mathbf{I})\\
&= \argmin_{\lambda} 2n \tr (\widehat{\bB}^{(1)\top}[\lambda]\bS_{xx}\widehat{\bB}^{(1)}[\lambda]/2-\widehat{\bS}_{xy}^{\top}\widehat{\bB}^{(1)}[\lambda]) + \log n \{q + \sum_{i,j} \mathbb{I}((\widehat{\bB}^{(1)}[\lambda])_{i,j}\neq 0) \}.
\end{align*}
The optimal $\lambda_{\bB^{(1)}}$ is determined as the one that minimizes the corresponding $\mathrm{BIC}(\widehat{\bB}^{(1)}[\lambda], \mathbf{I})$. Then we select $\widehat{\bB}^{(1)}$ corresponding to the optimal $\lambda_{\bB^{(1)}}$ and denote it by $\widehat{\bB}^{(1)}[\lambda_{\bB^{(1)}}]$. 
For the second stage, we plug-in $\widehat{\bB}^{(1)}[\lambda_{\bB^{(1)}}]$ obtained from the first stage and use the following BIC formula to select $\lambda_{\bTheta}$:
\begin{align*}
    \lambda_{\bTheta} & = \argmin_{\lambda}\mathrm{BIC}(\widehat{\bTheta}_{\varepsilon\varepsilon}[\lambda] , \widehat{\bB}^{(1)}[\lambda_{\bB^{(1)}}])\\
    &= \argmin_{\lambda} 2n \tr [\widehat{\bTheta}_{\varepsilon\varepsilon}[\lambda](\widehat{\bB}^{(1)\top}[\lambda_{\bB^{(1)}}]\bS_{xx}\widehat{\bB}^{(1)}[\lambda_{\bB^{(1)}}]/2-\widehat{\bS}_{xy}^{\top}\widehat{\bB}^{(1)}[\lambda_{\bB^{(1)}}])] -n \log\det(\widehat{\bTheta}_{\varepsilon\varepsilon}[\lambda])  \\
    &\hspace{3in}+ \log n \{q + \sum_{i<j} \mathbb{I}((\widehat{\bTheta}_{\varepsilon\varepsilon}[\lambda])_{i,j}\neq 0)\}. 
    \end{align*}
The optimal $\lambda_{\bTheta}$ is selected as the one that minimizes the corresponding $\mathrm{BIC}(\widehat{\bTheta}_{\varepsilon\varepsilon}[\lambda] , \widehat{\bB}^{(1)}[\lambda_{\bB^{(1)}}])$. We choose $\widehat{\bTheta}_{\varepsilon\varepsilon}$ with the optimal $\lambda_{\bTheta}$ and denote it by $\widehat{\bTheta}_{\varepsilon\varepsilon}[\lambda_{\bTheta}]$. Finally, for the third stage of estimation, we plug-in $\widehat{\bTheta}_{\varepsilon\varepsilon}[\lambda_{\bTheta}]$ obtained from the second stage and use the following BIC formula to estimate $\lambda_{\bB^{(2)}}$:
\begin{align*}
    \lambda_{\bB^{(2)}} &= \argmin_{\lambda}\mathrm{BIC}(\widehat{\bB}^{(2)}[\lambda], \widehat{\bTheta}_{\varepsilon\varepsilon}[\lambda_{\bTheta}])\\
    &= \argmin_{\lambda} 2n  \tr  [\widehat{\bTheta}_{\varepsilon\varepsilon}[\lambda_{\bTheta}] (\widehat{\bB}^{(2)\top}[\lambda]\bS_{xx}\widehat{\bB}^{(2)}[\lambda]/2-\widehat{\bS}_{xy}^{\top}\widehat{\bB}^{(2)}[\lambda])] + \log n \{q + \sum_{i,j} \mathbb{I}((\widehat{\bB}^{(2)}[\lambda])_{i,j}\neq 0) \}.
\end{align*}
The optimal $\lambda_{\bB^{(2)}}$ is selected as the one that minimizes the corresponding $\mathrm{BIC}(\widehat{\bB}^{(2)}[\lambda], \widehat{\bTheta}_{\varepsilon\varepsilon}[\lambda_{\bTheta}])$. We choose the estimate  $\widehat{\bB}^{(2)}[\lambda_{\bB^{(2)}}]$ corresponding to the optimal $\lambda_{\bB^{(2)}}$.

\section{Theoretical properties}

In this section, we present our main results, establishing estimation error bounds for $\widehat{\bB}^{(1)}$, $\widehat{\bTheta}_{\varepsilon\varepsilon}$, and $\widehat{\bB}^{(2)}$. We begin in Section \ref{stage1} with the necessary assumptions and lemmas to derive the $\ell_1$ and $\ell_2$ estimation error bounds for each column of $\widehat{\bB}^{(1)}$, denoted by $\widehat{\bbeta}_{l}$ for $l = 1, \ldots, q$, as stated in Proposition \ref{prop:1}. In Section \ref{stage2}, we introduce lemmas and the irrepresentability assumption, enabling us to derive the estimation error bound for the precision matrix $\widehat{\bTheta}_{\varepsilon\varepsilon}$ with respect to the elementwise maximum norm and operator norm, as presented in Proposition \ref{prop:2}. Finally, Section \ref{stage3} discusses the conditions and assumptions needed to establish Theorem \ref{thm:1}, which provides the estimation error bound for $\widehat{\bB}^{(2)}$ in terms of the Frobenius norm and elementwise $\ell_1$-norm. Across all theoretical bounds, we characterize the impact of missingness probability and error variance on estimation error.

\subsection{Estimation error for $\widehat{\protect\bB}^{(1)}$\label{stage1}}

The following assumption is imposed on the population covariance matrix
to mildly control the error of the Lasso solution. 
\begin{assumption}
[Restricted eigenvalue condition]\label{assu:1} Define the cone
set 
\[
\mathbb{C}(S_{l})=\left\{ \bdelta\in\mathbb{R}^{p}:\|(\bdelta)_{S_{l}^{c}}\|_{1}\leq3\|(\bdelta)_{S_{l}}\|_{1}\right\} .
\]
We assume the following restricted eigenvalue condition \citep[Page 208,][]{wainwright2019high}
for $\bS_{xx}$ over $\mathbb{C}(S_{l})$,
\[
0<\kappa_{l}=\min_{\bdelta\neq0,\ \bdelta\in\mathbb{C}(S_{l})}\frac{\bdelta^{\top}\bS_{xx}\bdelta}{\|\bdelta\|_{2}^{2}}.
\]
\end{assumption}
This assumption assumes restricted curvature of the loss function in
certain directions, specifically, along the defined cone set $\mathbb{C}(S_{l})$. We would also require the following lemmas to derive the estimation error bound for $\widehat{\protect\bB}^{(1)}$, as will be discussed in Proposition \ref{prop:1}. In \ref{Pr_Lem1} and \ref{Pr_Lem2} we provide proofs of these lemmas and in \ref{Pr_Prop1} we provide the proof of Proposition \ref{prop:1}. 
\begin{lem}
\label{lem:1} Assume that each row of the error matrix $\boldsymbol{\varepsilon}\in\mathbb{R}^{n\times q}$
and each row of multiplicative error matrix $\mathbf{W}\in\mathbb{R}^{n\times q}$
are two sub-Gaussian random vectors where each elements of the vectors
follow the sub-Gaussian distribution with parameters $\sigma_{\varepsilon}^{2}$
and $\sigma_{W}^{2}$ respectively. Under this assumption, the elementwise
max norm of the deviation between $\widehat{\bS}_{xy}$ and $\mathbf{S}_{xy}$
satisfies the following probability bound for any $t\leq t_{0}^{(1)}$
with $t_{0}^{(1)}=\sigma_{\varepsilon}\sigma_{W}X_{\max}/\mu_{\min}$,
\[
\mathrm{Pr}[\VERT\widehat{\bS}_{xy}-\mathbf{S}_{xy}\VERT_{\max}\geq t]\leq pqC\exp\left(-\frac{cn\mu_{\min}^{2}t^{2}}{\sigma_{W}^{2}X_{\max}^{2}\max(s_{\max}^{2}X_{\max}^{2}B_{\max}^{2},\sigma_{\varepsilon}^{2})}\right),
\]
where $\mu_{\min}=\min_{j}(1-\rho_{j})>0$, $X_{\max}=\max_{i,k}|X_{ik}|<\infty$,
$B_{\max}=\max_{k,j}|\beta_{kj}^{*}|$, $s_{\max}=\max_{j}s_{j}$,
for $j\in\{1,\ldots,q\}$. Let $(\widehat{\bS}_{xy})_{\bullet l}$
and $(\bS_{xy})_{\bullet l}$ be the $l$th columns of $\widehat{\bS}_{xy}$
and $\bS_{xy}$, respectively, then we have 
\[
\mathrm{Pr}[\|(\widehat{\bS}_{xy})_{\bullet l}-(\mathbf{S}_{xy})_{\bullet l}\|_{\infty}\geq t]\leq pC\exp\left(-\frac{cn\mu_{\min}^{2}t^{2}}{\sigma_{W}^{2}X_{\max}^{2}\max(s_{\max}^{2}X_{\max}^{2}B_{\max}^{2},\sigma_{\varepsilon}^{2})}\right),
\]
\end{lem}

\begin{lem}
\label{lem:2} Assume that each row of the error matrix $\boldsymbol{\varepsilon}\in\mathbb{R}^{n\times q}$
is a sub-Gaussian random vector where each element of the vector follows
the sub-Gaussian distribution with parameter $\sigma_{\varepsilon}^{2}$.
Under this assumption, the elementwise max norm of the deviation between
$\mathbf{S}_{xy}$ and $\bS_{xx}\bB^{*}$ satisfies the following
probability bound 
\[
\mathrm{Pr}(\VERT\bS_{xy}-\bS_{xx}\bB^{*}\VERT_{\max}\geq t)\leq pqC\exp\left[-\frac{cnt^{2}}{\sigma_{\varepsilon}^{2}X_{\max}^{2}}\right]
\]
and
\[
\mathrm{Pr}(\|(\bS_{xy})_{\bullet l}-\bS_{xx}\bbeta_{l}^{*}\|_{\infty}\geq t)\leq pC\exp\left[-\frac{cnt^{2}}{\sigma_{\varepsilon}^{2}X_{\max}^{2}}\right]
\]
where $(\bS_{xy})_{\bullet l}$ and $\bbeta_{l}^{*}$ are the $l$th
columns of $\bS_{xy}$ and $\bB^{*}$, respectively, and $X_{\max}=\max_{i,k}|X_{ik}|<\infty$.
\end{lem}

With Lemmas \ref{lem:1} and \ref{lem:2} established, we can now derive the estimation error bound for $\widehat{\bbeta}_{l}$ in \eqref{eq:Beta_obj} for the initial estimation of $\bB^{*}$ in \eqref{1stB}.
Without loss of generality, we can write the $l$th column of the
true regression coefficient matrix as $\boldsymbol{\beta}_{l}^{*}=(\boldsymbol{\beta}_{S_{l}}^{*\top},\mathbf{0}^{\top})^{\top}$
and the corresponding matrix $\mathbf{X}=((\mathbf{X})_{\bullet S_{l}},(\mathbf{X})_{\bullet S_{l}^{c}})$.
Hence, the true model for the $l$th column of $\mathbf{Y}$ can be
written as $(\mathbf{Y})_{\bullet l}=(\mathbf{X})_{\bullet S_{l}}\boldsymbol{\beta}_{S_{l}}^{*}+(\boldsymbol{\varepsilon})_{\bullet l}$.

\begin{prop}
\label{prop:1} We assume that the $l$th column of the true coefficient
matrix from \eqref{eq:model1}, $\bbeta_{l}^{*}$ has support $S_{l}\subseteq\{1,\ldots,p\}$
with cardinality $s_{l}\coloneqq|S_{l}|$, meaning that $\beta_{lj}^{*}=0$
for all $j\in S_{l}^{c}$, where $S_{l}^{c}$ denotes the complement
of $S_{l}.$ Let us consider that Assumption \ref{assu:1} on the
parameter $\kappa_{l}>0$ hold. We assume that the tuning parameter
$\lambda_{l}$ in \eqref{eq:Beta_obj} satisfies 
\begin{equation}
\lambda_{l}\geq C(\sigma_{W},\sigma_{\varepsilon})\sqrt{\frac{\log p}{n}},\label{eq:lambda_ell}
\end{equation}
where $C(\sigma_{W},\sigma_{\varepsilon})=2X_{\max}\max\left[\sigma_{W}s_{\max}X_{\max}B_{\max}/\mu_{\min},\sigma_{W}\sigma_{\varepsilon}/\mu_{\min},\sigma_{\varepsilon}\right].$
Then any estimate $\widehat{\bbeta}_{l}$ from \eqref{eq:Beta_obj}
satisfy the following estimation error bounds 
\[
\|\widehat{\bbeta}_{l}-\bbeta_{l}^{*}\|_{2}\leq3\sqrt{s_{l}}\lambda_{l}/\kappa_{l},\qquad\text{and}\qquad\|\widehat{\bbeta}_{l}-\bbeta_{l}^{*}\|_{1}\leq4\sqrt{s_{l}}\|\widehat{\bbeta}_{l}-\bbeta_{l}^{*}\|_{2}=12s_{l}\lambda_{l}/\kappa_{l}
\]
with a probability at least $1-C\exp(-c\log p)$. 
\end{prop}

\subsection{Estimation error for ${\widehat{\mathbf{\Theta}}}_{\varepsilon\varepsilon}$}

\label{stage2}

This section establishes the estimation bound for $\widehat{\bTheta}_{\varepsilon\varepsilon}$ in Proposition \ref{prop:2}. We would require the following lemmas and the assumption (irrepresentability condition) to prove the proposition.

\begin{lem}
\label{lem3} Assume that each row of random error matrix $\boldsymbol{\varepsilon}\in\mathbb{R}^{n\times q}$
and each row of multiplicative error matrix $\mathbf{W}\in\mathbb{R}^{n\times q}$
follow the sub-Gaussian distribution with parameters $\sigma_{\varepsilon}^{2}$
and $\sigma_{W}^{2}$ respectively, then the elementwise max norm
of the deviation between $\widehat{\bS}_{yy}$ and $\mathbf{S}_{yy}$
satisfies the probability bound for any $t\leq t_{0}^{(2)}$ with
\[
t_{0}^{(2)}\coloneqq\min(X_{\max}^{2}B_{\max}^{2}s_{\max}^{2}\sigma_{W}^{2}/m_{\min},X_{\max}B_{\max}s_{\max}\sigma_{W}^{2}\sigma_{\varepsilon}^{2}/m_{\min}n^{1/3},\sigma_{W}^{2}\sigma_{\varepsilon}^{2}/m_{\min}n^{1/3})
\]
\[
\mathrm{Pr}(\VERT\widehat{\bS}_{yy}-\mathbf{S}_{yy}\VERT_{\max}\geq t)\leq4q^{2}C\Bigg\{\exp\left(-\frac{cnt^{2}m_{\min}^{2}}{\sigma_{W}^{4}\max\left\{ X_{\max}^{4}B_{\max}^{4}s_{\max}^{4},X_{\max}^{2}B_{\max}^{2}s_{\max}^{2}\sigma_{\varepsilon}^{4},\sigma_{\varepsilon}^{4}\right\} }\right)\Bigg\}.
\]
where $m_{\min}=\min_{j,k}\mathbb{E}(W_{j}W_{k}^{\top})>0$, $X_{\max}=\max_{i,k}|X_{ik}|<\infty$,
$B_{\max}=\max_{k,j}|\beta_{kj}^{*}|$, $s_{\max}=\max_{j}s_{j}$,
for $j\in\{1,\ldots,q\}$. 
\end{lem}

\begin{lem}
\label{lem4} Assume that each row of the error matrix $\boldsymbol{\varepsilon}\in\mathbb{R}^{n\times q}$
and each row of multiplicative error matrix $\mathbf{W}\in\mathbb{R}^{n\times q}$
are two sub-Gaussian random vectors where each elements of the vectors
follow the sub-Gaussian distribution with parameters $\sigma_{\varepsilon}^{2}$
and $\sigma_{W}^{2}$ respectively. Under this assumption, the elementwise
max norm of the deviation between $\widehat{\bS}_{\varepsilon\varepsilon}$
and $\bSigma_{\varepsilon\varepsilon}^{*}$ satisfies the following
condition
\[
\mathrm{Pr}(\VERT\widehat{\bS}_{\varepsilon\varepsilon}-\bSigma_{\varepsilon\varepsilon}^{*}\VERT_{\max}\leq\Delta)\geq1-C\exp(-c\log q^{2})-C\exp(-c\log(pq)),
\]
where $\Delta$ is set to be defined as follows:
\begin{align}
\Delta & =C_1(\sigma_W,\sigma_\varepsilon) \sqrt{\frac{\log(q^{2})}{n}}+X_{\max}^{2}\left(\max_{l\in[q]}12s_{l}\lambda_{l}/\kappa_{l}\right)^{2}+X_{\max}^{2}s_{\max}B_{\max}\max_{l\in[q]}12s_{l}\lambda_{l}/\kappa_{l}\label{eq:Delta-1}\\
 & \qquad +C_2(\sigma_\varepsilon)\sqrt{\frac{\log(q^{2})}{n}}+\sigma_{\varepsilon}^{2}\sqrt{\frac{\log(q^{2})}{n}},\nonumber 
\end{align}
where $C_1(\sigma_W,\sigma_\varepsilon)=\sigma_{W}^{2}\max\left\{ X_{\max}^{2}B_{\max}^{2}s_{\max}^{2}/m_{\min},X_{\max}B_{\max}s_{\max}\sigma_{\varepsilon}^{2}/m_{\min},\sigma_{\varepsilon}^{2}/m_{\min}\right\}$, $C_2(\sigma_\varepsilon)=\sigma_{\varepsilon}X_{\max}s_{\max}B_{\max}$ and $C_3(\sigma_\varepsilon) = \sigma_{\varepsilon}^{2}.$

\end{lem}

We can bound the deviation between the projected estimate $\widetilde{\bS}_{\varepsilon\varepsilon}$
and the true $\mathbf{\bSigma_{\varepsilon\varepsilon}^{*}}$ using the deviation between the surrogate estimate $\widehat{\bS}_{\varepsilon\varepsilon}$ and  $\mathbf{\bSigma_{\varepsilon\varepsilon}^{*}}$
\begin{equation}
\VERT\widetilde{\bS}_{\varepsilon\varepsilon}-\bSigma_{\varepsilon\varepsilon}^{*}\VERT_{\max}\leq\VERT\widetilde{\bS}_{\varepsilon\varepsilon}-\widehat{\bS}_{\varepsilon\varepsilon}\VERT_{\max}+\VERT\widehat{\bS}_{\varepsilon\varepsilon}-\bSigma_{\varepsilon\varepsilon}^{*}\VERT_{\max}\leq2\VERT\widehat{\bS}_{\varepsilon\varepsilon}-\bSigma_{\varepsilon\varepsilon}^{*}\VERT_{\max}.\label{eq:Proj_conv_rate}
\end{equation}
This is because, by definition of $\widetilde{\bS}_{\varepsilon\varepsilon}$ in \eqref{eq:Projected_Syy},
we have $\VERT\widetilde{\bS}_{\varepsilon\varepsilon}-\widehat{\bS}_{\varepsilon\varepsilon}\VERT_{\max}\leq\VERT\widehat{\bS}_{\varepsilon\varepsilon}-\bSigma_{\varepsilon\varepsilon}^{*}\VERT_{\max}$  
(as $\text{\ensuremath{\bSigma_{\varepsilon\varepsilon}^{*}}}$
is also positive semi-definite). Combining it with the triangle
inequality leads to the result.

Next, to establish the following irrepresentability condition assumption and Proposition \ref{prop:2} let us define some terminologies. Following \citet{ravikumar2011high}, we define the maximum degree
or row cardinality of $\bTheta_{\varepsilon\varepsilon}^{*}$ as 
\[
d_{q}=\max_{l\in[q]}\mathrm{card}\left\{ l^{\prime}\in[q]\backslash\left\{ l\right\} :(\bTheta_{\varepsilon\varepsilon})_{ll^{\prime}}\neq0\right\} ,
\]
and let $\kappa_{\bSigma_{\varepsilon\varepsilon}^{*}}=\VERT\bSigma_{\varepsilon\varepsilon}^{*}\VERT_{\infty}$.
We further let $S=\left\{ (l,l^{\prime})\in[q]\times[q]:(\bTheta_{\varepsilon\varepsilon})_{ll^{\prime}}\neq0\right\} $,
$S^{c}=[q]\times[q]\backslash S$, and $\bGamma=\bSigma_{\varepsilon\varepsilon}^{*}\otimes\bSigma_{\varepsilon\varepsilon}^{*}\in\mathbb{R}^{q^{2}}\times\mathbb{R}^{q^{2}}$.
For any two subsets $T$ and $T^{\prime}$ of $[q^{2}]$, let $(\bGamma)_{TT^{\prime}}$
denote the $\card(T)\times\card(T^{\prime})$ matrix with rows and
columns of $\bGamma$ indexed by $T$ and $T^{\prime}$, respectively.
Then we set $\kappa_{\bGamma}=\VERT(\bGamma)_{SS}^{-1}\VERT_{\infty}$.
To
derive the bound for the estimator $\widehat{\bTheta}_{\varepsilon\varepsilon}$,
let us introduce the irrepresentability condition introduced in Assumption
1 in \citet{ravikumar2011high} for graphical Lasso without any corrupted
data.

\begin{assumption}
[Irrepresentability condition]\label{assu:2} There exists $\alpha\in(0,1]$
such that 
\[
\max_{e\in S^{c}}\VERT(\bGamma)_{\left\{ e\right\} S}(\bGamma)_{SS}^{-1}\VERT_{1,1}\leq1-\alpha
\]
where $\bGamma=\bSigma_{\varepsilon\varepsilon}^{*}\otimes\bSigma_{\varepsilon\varepsilon}^{*}\in\mathbb{R}^{q^{2}}\times\mathbb{R}^{q^{2}}$.
\end{assumption}
\begin{prop}
\label{prop:2} Suppose that, for all $l\in[q],$ $\kappa_{l}>0$,
$\lambda_{l}$ in \eqref{eq:Beta_obj} satisfies \eqref{eq:lambda_ell}
and $n$ is sufficiently large to ensure that Proposition \ref{prop:1}
applies. Further, assume that Assumption \ref{assu:2} is satisfied
and that $n$ is sufficiently large to ensure that 
\begin{equation}
{6\left(1+\frac{8}{\alpha}\right){}^{2}\max\left(\kappa_{\bSigma_{\varepsilon\varepsilon}^{*}},\kappa_{\bGamma},\kappa_{\bSigma_{\varepsilon\varepsilon}^{*}}^{3},\kappa_{\bGamma}^{2}\right)d_{q}\times\Delta\leq1.}\label{eq:ss_restriction}
\end{equation}
Finally, suppose that the tuning parameter $\lambda_{\bTheta}$ in
\eqref{eq:glasso} satisfies 
\begin{equation}
\lambda_{\bTheta}=\frac{8\Delta}{\alpha}.\label{eq:tuning_theta}
\end{equation}
Then with probability at least 
\begin{align}
 & 1-C\exp(-c\log q^{2})-C\exp(-c\log(pq))\label{eq:Prob_P2}
\end{align}
the estimator $\widehat{\bTheta}_{\varepsilon\varepsilon}$ satisfies
\begin{equation}
\VERT\widehat{\bTheta}_{\varepsilon\varepsilon}-\bTheta_{\varepsilon\varepsilon}^{*}\VERT_{\max}\leq\left\{ 2\kappa_{\bGamma}(1+\frac{8}{\alpha})\right\} \Delta\coloneqq\Delta_{\infty}(\bTheta_{\varepsilon\varepsilon}^{*})\label{eq:Theta_max_bound}
\end{equation}
and 
\begin{equation}
\VERT\widehat{\bTheta}_{\varepsilon\varepsilon}-\bTheta_{\varepsilon\varepsilon}^{*}\VERT_{\mathrm{2}}\leq d_{q}\Delta_{\infty}(\bTheta_{\varepsilon\varepsilon}^{*})\coloneqq\Delta_{1}(\bTheta_{\varepsilon\varepsilon}^{*}).\label{eq:Theta_op_bound}
\end{equation}
\end{prop}
The proofs of Lemma \ref{lem3}, Lemma \ref{lem4} and Proposition \ref{prop:2} are provided in \ref{Pr_Lem3}, \ref{Pr_Lem4} and \ref{Pr_Prop2}, respectively. 

\subsection{Estimation error for $\widehat{\protect\bB}^{(2)}$}

\label{stage3}

Let $\mathcal{L}(\cdot;\bS_{xx},\bS_{xy},\bTheta):\mathbb{R}^{p\times q}\rightarrow\mathbb{R}$
be the loss functions to estimate \textbf{$\bB^{*}$} which depend
on matrices $\bS_{xx}\in\mathbb{R}^{p\times p}$, $\bS_{xy}\in\mathbb{R}^{p\times q}$
and $\bTheta\in\mathbb{R}^{q\times q}$. Let $\bB\in\mathbb{\mathbb{R}}^{p\times q}$
be any arbitrary matrix and set 
\begin{align}
\mathcal{L}(\bB;\bS_{xx},\bS_{xy},\bTheta) & =\tr[(\bB^{\top}\bS_{xx}\bB/2-\bS_{xy}^{\top}\mathbf{B})\bTheta]\nonumber \\
 & =\vec(\bB)^{\top}(\bTheta\otimes\bS_{xx})\vec(\bB)/2-\tr(\bTheta\bS_{xy}^{\top}\bB).\label{eq:2ndBeta_loss}
\end{align}
The quantity $\bS_{xx}$ in \eqref{eq:2ndBeta_loss} will be replaced
by the estimate of $\bSigma_{xx}^{*}$ from the data since there is
no missingness in $\bX$. The quantities $\bS_{xy}$ and $\bTheta$
will be replaced by the surrogate and the estimates of $\bSigma_{xy}^{*}$
and $\bTheta_{\varepsilon\varepsilon}^{*}$, respectively. For brevity,
we write, $\mathcal{L}(\cdot;\bS_{xx},\bS_{xy},\bTheta)$ as $\mathcal{L}$
in the following derivations.

We consider the elementwise sparsity of $\bB^{*}$. A matrix $\bB^{*}$
is called elementwise sparse if its support set $S\subset\mathbb{R}^{p\times q}$
is such that $s=\card(S)\ll pq$. In order to obtain an elementwise
sparse estimator of \textbf{$\bB^{*}$}, it is natural to regularize
the least squares program with the $\VERT\cdot\VERT_{1,1}$ penalty
of $\bB$, 
\begin{equation}
\widehat{\bB}^{(2)}=\argmin_{\bB\in\mathbb{R}^{p\times q}}\left\{ \mathcal{L}(\bB;\bS_{xx},\widehat{\bS}_{xy},\widehat{\bTheta}_{\varepsilon\varepsilon})+\lambda_{\bB}\VERT\bB\VERT_{1,1}\right\} ,\label{eq:Beta_obj_mat}
\end{equation}
where $\lambda_{\bB}>0$ is a tuning parameter.

Next, we establish the restricted eigenvalue (RE) condition for the
loss function. Following \citet{negahban2012unified}, let us define,
for all $\bDelta=\bB-\bB^{*}\in\mathbb{Q}$, 
\begin{align*}
\mathcal{E}\mathcal{L}(\bDelta,\bB^{*}) & =\mathcal{E}\mathcal{L}(\bDelta,\bB^{*};\bS_{xx},\bS_{xy},\bTheta_{\varepsilon\varepsilon}^{*})\\
 & =\mathcal{L}(\bB^{*}+\bDelta;\bS_{xx},\bS_{xy},\bTheta_{\varepsilon\varepsilon}^{*})-\mathcal{L}(\bB^{*};\bS_{xx},\bS_{xy},\bTheta_{\varepsilon\varepsilon}^{*})-\langle\nabla_{\bB}\mathcal{L}(\bB^{*};\bS_{xx},\bS_{xy},\bTheta_{\varepsilon\varepsilon}^{*}),\bDelta\rangle\\
 & =\langle\bDelta^{\top}\bS_{xx}\bDelta,\bTheta_{\varepsilon\varepsilon}^{*}\rangle/2\\
 & =\vec(\bDelta)^{\top}(\bTheta_{\varepsilon\varepsilon}^{*}\otimes\bS_{xx})\vec(\bDelta)/2,
\end{align*}
where 
\[
\nabla_{\bB}\mathcal{L}(\bB;\bS_{xx},\bS_{xy},\bTheta_{\varepsilon\varepsilon}^{*})=\bS_{xx}\bB\bTheta_{\varepsilon\varepsilon}^{*}-\bS_{xy}\bTheta_{\varepsilon\varepsilon}^{*}
\]
To derive the bound for the estimator $\widehat{\bB}^{(2)}$
we need the following assumption on the RE condition for the loss
function. 
\begin{assumption}
\label{assu:3}The loss function $\mathcal{L}(\cdot;\bS_{xx},\bS_{xy},\bTheta_{\varepsilon\varepsilon}^{*})$
satisfies the RE condition 
\begin{equation}
\mathcal{E}\mathcal{L}(\bDelta,\bB^{*})\geq\kappa\VERT\bDelta\VERT_{F}^{2}\qquad\forall\bDelta\in\mathbb{C}(S).\label{eq:RE_refined_B}
\end{equation}
with constant $\kappa>0$ over the cone set $\mathbb{C}(S)=\left\{ \bM\in\mathbb{R}^{p\times q}:\VERT(\bM)_{S^{c}}\VERT_{1,1}\leq3\VERT(\bM)_{S}\VERT_{1,1}\right\} $. 
\end{assumption}
\begin{thm}
\label{thm:1}Suppose that Assumption \ref{assu:3} and the assumptions
of Proposition \ref{prop:2} hold. Further suppose that for $s=\card(S)$,
where $S\subset\mathbb{R}^{p\times q}$ is the support set 
\end{thm}
\begin{enumerate}
\item $n$ is sufficiently large to ensure that $\kappa^{\prime}\geq\kappa>0$,
where $\kappa^{\prime}$ is defined as the empirical counterpart to
$\kappa$ in Assumption \ref{assu:3} as 
\begin{equation}
\kappa^{\prime}=\kappa-\VERT\bS_{xx}\VERT_{2}\Delta_{1}(\bTheta_{\varepsilon\varepsilon}^{*})/2\label{eq:kappa_prime}
\end{equation}
\item The tuning parameter $\lambda_{\bB}$ in \eqref{eq:Beta_obj_mat}
satisfies 
\begin{equation}
\lambda_{\bB}/2\geq(\lambda_{0}/2)(\VERT\widehat{\bTheta}_{\varepsilon\varepsilon}-\bTheta_{\varepsilon\varepsilon}^{*}\VERT_{1}+\VERT\bTheta_{\varepsilon\varepsilon}^{*}\VERT_{1}),\label{eq:lambda_overall}
\end{equation}
with 
\[
\lambda_{0}\coloneqq  C(\sigma_{W},\sigma_{\varepsilon}) \sqrt{\frac{\log p}{n}},
\]
where $C(\sigma_{W},\sigma_{\varepsilon})=2X_{\max}\max\left[\sigma_{W}s_{\max}X_{\max}B_{\max}/\mu_{\min},\sigma_{W}\sigma_{\varepsilon}/\mu_{\min},\sigma_{\varepsilon}\right],$ then the estimator $\widehat{\bB}^{(2)}$
satisfies 
\begin{equation}
\VERT\widehat{\bB}^{(2)}-\bB^{*}\VERT_{F}\leq3\sqrt{s}\lambda_{\bB}/\kappa^{\prime},\qquad\VERT\widehat{\bB}^{(2)}-\bB^{*}\VERT_{1,1}\leq12s\lambda_{\bB}/\kappa^{\prime}\label{eq:B_final_est}
\end{equation}
with probability $1-qC\exp(-c\log p)$.
\end{enumerate}
\ref{Pr_Thm1} provides the proof of Theorem \ref{thm:1}. 
    
\section{Simulation}

\subsection{Simulations to verify theoretical results\label{sim_theory_verify}}
In this section, we present simulation results to study the scaling behavior of estimation errors across the three stages of estimation, as discussed by our theoretical analysis.

We generate an $n\times p$ matrix $\bX$ with rows
drawn independently from $\mathcal{N}_{p}(0,\bSigma^{*}_{xx})$ where $[\bSigma^{*}_{xx}]_{i,j}=0.7^{|i-j|}$ follows an AR(1) covariance. Response matrix  $\bY \in \mathbb{R}^{n \times q}$ is generated from a multivariate Gaussian distribution
with $\mathcal{N}_{q}(\bX\bB^{*},\bSigma^{*}_{\varepsilon\varepsilon})$. We consider the AR(1) error covariance $[\bSigma^{*}_{\varepsilon\varepsilon}]_{i,j}=\rho_{\varepsilon}^{|i-j|}$, which leads to a tri-diagonal, sparse precision matrix $\bTheta^{*}_{\varepsilon\varepsilon}=(\bSigma^{*}_{\varepsilon\varepsilon})^{-1}$. We generate $\bW$ from a Bernoulli distribution with varying
probability of missingness for each column $\rho_{W}$. To obtain an elementwise sparse model on $\bB^{*}$,
generated a $p\times q$ matrix $\bB^{*}$ so that in each
of its columns, only 5 entries of $p$ (chosen at random within
each column) are drawn from the uniform distribution on $[-1, 1]$ and the remaining entries are equal to zero. Finally, we derive the contaminated responses as $\bZ=\bY\odot\bW$.

We aim to examine how the estimation error \( \|\widehat{\bbeta}_{l}-\bbeta_{l}^{*}\|_{2} \), presented in Proposition \ref{prop:1}, with its upper bound dependent on the tuning parameter $\lambda_l$ defined in \eqref{eq:lambda_ell}, behaves under an AR(1) covariance structure for the error. We vary two parameters: the error correlation \( \rho_{\varepsilon} \), where $[\bSigma^{*}_{\varepsilon\varepsilon}]_{i,j}=\rho_{\varepsilon}^{|i-j|}$, and the minimum signal strength \( \mu_{\min} = \min_{j}\mu_{j} = \min_{j}(1 - \rho_{W,j}) \), which is related to the probability of missingness in the $j$th column of the outcome, \( \rho_{W,j} \). For simplicity, we assume the same probability of missingness across columns. We analyze the impact of varying \( \rho_{\varepsilon} \) and \( \rho_{W,j} \) separately to evaluate the performance of the estimation error. In this stage, we inspect two scenarios: (S1A) varying \( n \) with fixed \( p \), and (S1B) varying \( p \) with fixed \( n \). The plots present the average estimation error, with shaded regions indicating standard errors, based on 100 replicated samples.

\textbf{Scenario S1A:} Let $n = (200, 400, 800, 1600, 3200, 6400, 12800)$, $p = 100$, $q = (10, 20)$, and $s_{\max} = 5$. We evaluate the following two scenarios by plotting $\|\widehat{\bbeta}_{l} - \bbeta_{l}^{*}\|_{2}$ against $\sqrt{n}$, as shown in Figure \ref{fig:ERS1P1}:
(a)   $\rho_{\varepsilon} = (0, 0.3, 0.7, 0.9)$ with $\rho_{W,j} = 0.05$; and (b)   $\rho_{W,j} = (0.005, 0.1, 0.2, 0.3)$ with $\rho_{\varepsilon} = 0.7$.

\textbf{Scenario S1B:} We set $p = (50, 100, 200, 400, 800)$, $n = 400$, $q = (10, 20)$, and $s_{\max} = 5$. We test the following two scenarios by plotting $\|\widehat{\bbeta_{l}} - \bbeta_{l}^{*}\|_{2}$ versus $\sqrt{\log(p)}$, as shown in Figure \ref{fig:ERS1P2}:
(a)  $\rho_{\varepsilon} = (0, 0.3, 0.7, 0.9)$ with $\rho_{W,j} = 0.05$;
    and (b)  $\rho_{W,j} = (0.005, 0.1, 0.2, 0.3)$ with $\rho_{\varepsilon} = 0.7$.
    
     \begin{figure}[!h]
	\centering
	\includegraphics[scale=0.26]{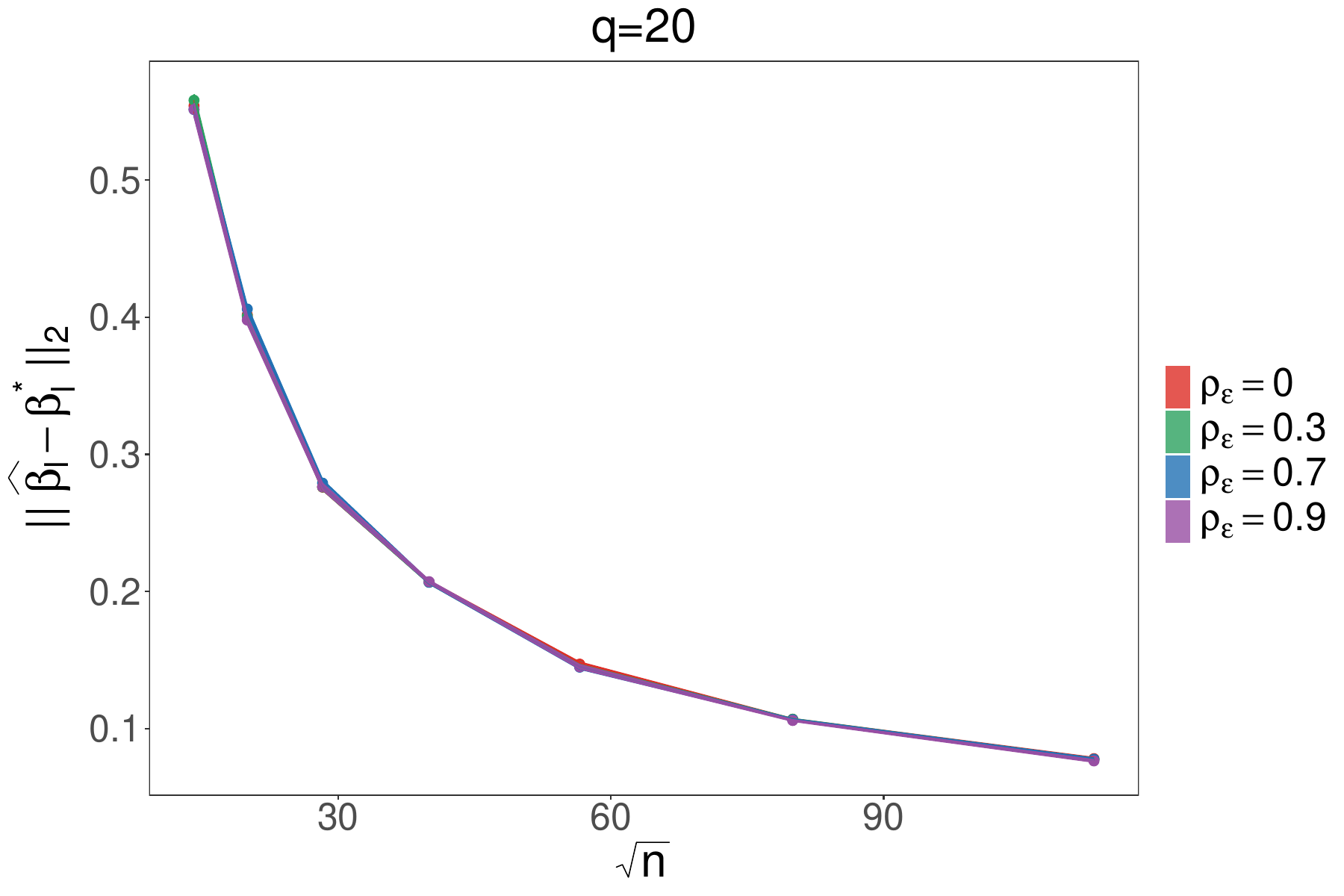}
        \includegraphics[scale=0.26]{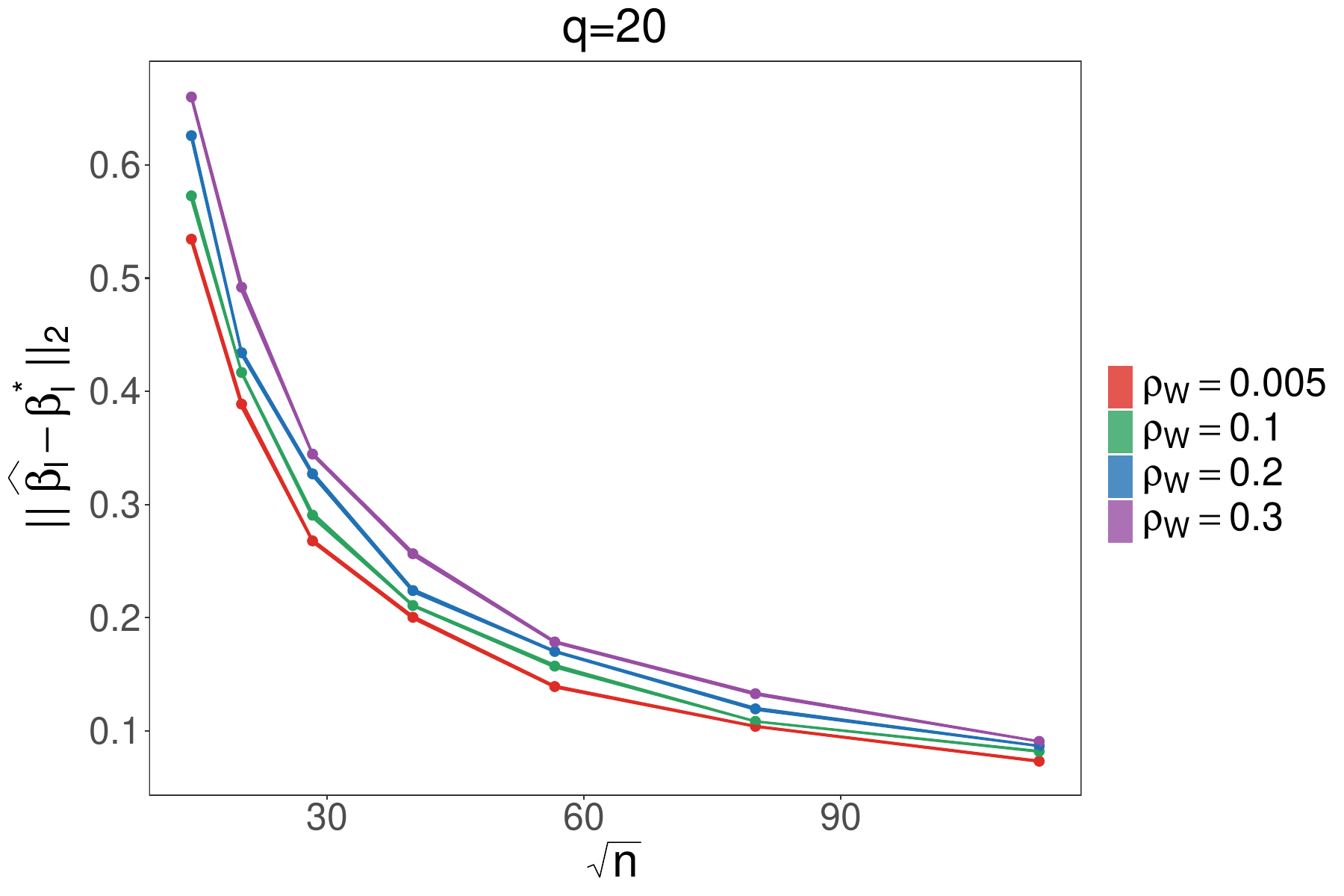}\\
	\caption{Scenario S1A. Plots of the average estimation error, $\|\widehat{{\bbeta}_{l}}-{\bbeta}_{l}^{*}\|_{2}$  against $\sqrt{n}$, $n=200,400,800,1600,3200,6400,12800$, for the outcome dimension,  $q=20$, with varying correlation between the errors, $\rho_{\varepsilon}$ (left panel) and the probability of being missing in the $j$th column for the
outcome $\rho_{W,j}$ (right panel).  Each point represents an average of 100 trials and the shaded regions indicate standard error for each method.} \label{fig:ERS1P1}
\end{figure}     
\begin{figure}[!h]
	\centering
	\includegraphics[scale=0.26]{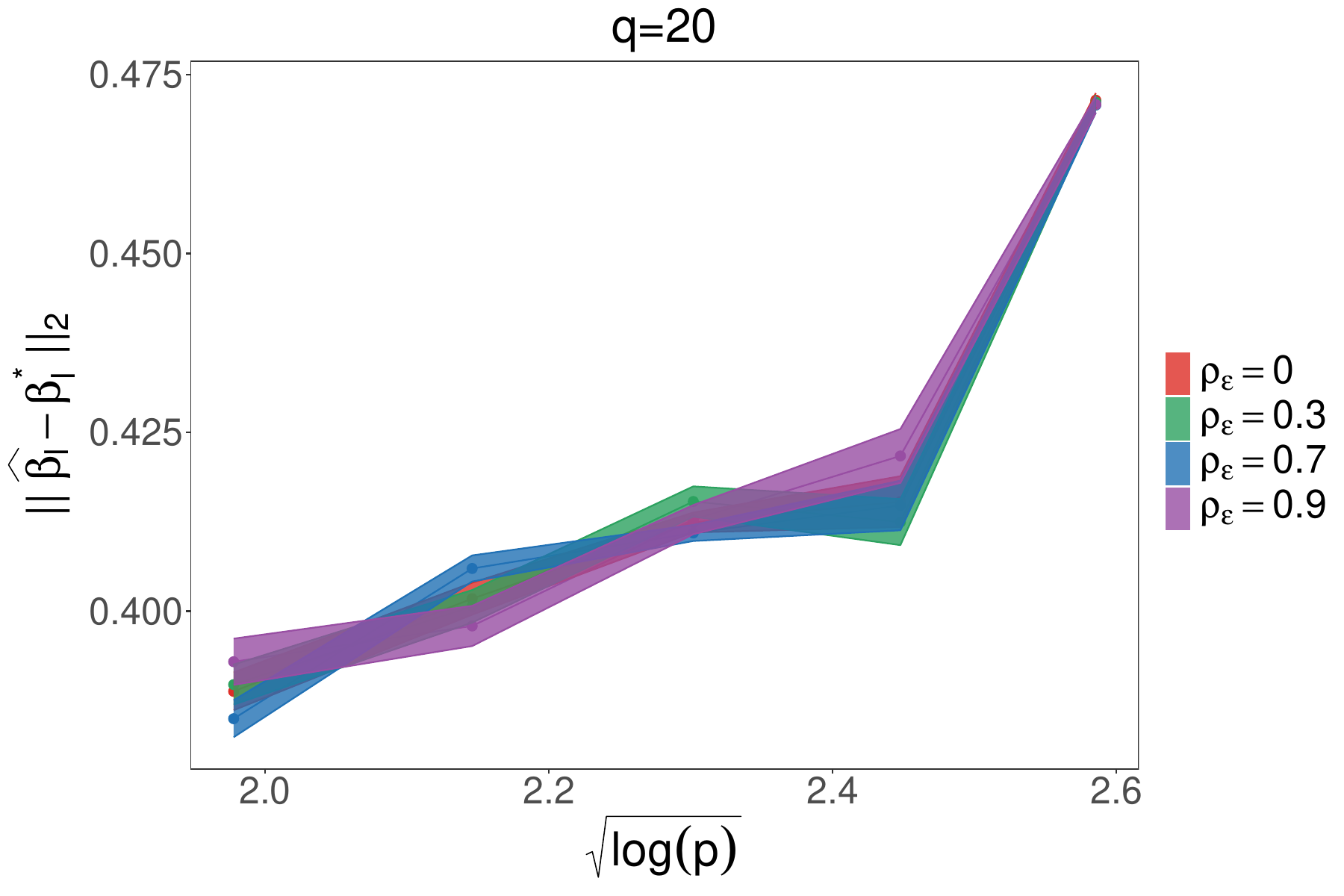}
        \includegraphics[scale=0.26]{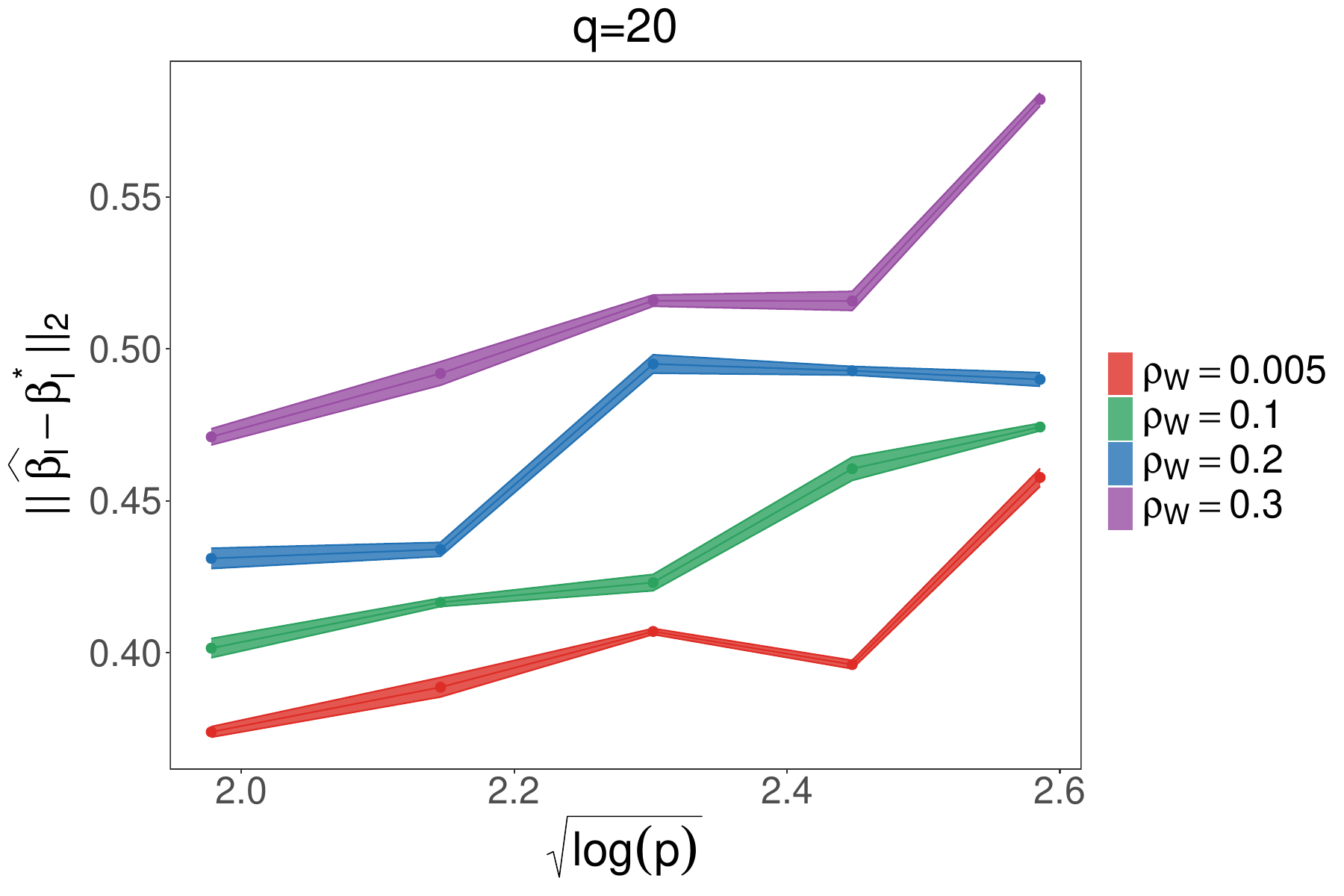}\\
	\caption{Scenario S1B. Plots of the average estimation error, $\|\widehat{{\bbeta}_{l}}-{\bbeta}_{l}^{*}\|_{2}$ against $\sqrt{\log(p)}$, $p=50,100,200,400$ and $800$, for the outcome dimension, $q=20$, with varying correlation between the errors, $\rho_{\varepsilon}$ (left panel) and the probability of being missing in each column of the outcomes, $\rho_{W,j}$ (right panel). Each point represents an average of 100 trials and the shaded regions indicate standard error for each method.} \label{fig:ERS1P2}
\end{figure}    

Figure \ref{fig:ERS1P1} illustrates Scenario S1A and shows several key trends for the upper bound of $\|\widehat{\bbeta}_{l}-\bbeta_{l}^{*}\|_{2}$. The $\ell_2$-norm of the estimation error for $\bbeta_{l}^{*}$ converges to zero at a rate of $1/\sqrt{n}$ as the sample size $n$ increases, aligning with the theoretical bound in both the setups with varying error correlation and degree of missingness in the outcomes.  Variations in the degree of correlation between the errors (left panel) $\rho_{\varepsilon}$, do not seem to impact the estimation errors in this stage. Additionally, assuming uniform missingness across all columns of the outcome, the right panel of Figure \ref{fig:ERS1P1} shows that the estimation error decreases when the proportion of missingness decreases. Figure \ref{fig:ERS1P1b} in \ref{q10Sim} provides $q=10$ scenario.

Figure \ref{fig:ERS1P2} illustrates Scenario S1B and shows several general trends along with a few anomalies. Theoretically, the estimation errors increase at a rate of $\sqrt{\log(p)}$ as the number of covariates $p$ increases, keeping other parameters constant. Consistent with the observations in Figure \ref{fig:ERS1P1}, varying $\rho_{\varepsilon}$ does not result in significant differences in the estimation errors (left panel). Additionally, when varying the probability of having missing data in the outcome (right panel), $\rho_{W}$, it is observed that as the proportion of missingness decreases, the estimation error decreases, as expected. Figure \ref{fig:ERS1P2b} in \ref{q10Sim} provides $q=10$ scenario.

\noindent \textbf{Stage 2:} We aim to verify the elementwise maximum norm deviation bound presented in \eqref{eq:Theta_max_bound} of Proposition \ref{prop:2}, where the upper bound depends on $\Delta$, as defined in \eqref{eq:Delta-1} of Lemma \ref{lem4}. We  vary $[\bSigma^{*}_{\varepsilon\varepsilon}]_{i,j}=\rho_{\varepsilon}^{|i-j|}$ and $m_{\min}=\min_{j,k}|\mathbb{E}(W_{j}W_{k})|=(1-\rho_{W,i})(1-\rho_{W,j})\text{ if }i\neq j,\text{ and }(1-\rho_{W,i})\text{ if }i=j$ separately to check the performance of $\VERT\widehat{\bTheta}_{\varepsilon\varepsilon}-\bTheta_{\varepsilon\varepsilon}^{*}\VERT_{\max}$. Since $[\bSigma^{*}_{\varepsilon\varepsilon}]_{i,j}$ and $m_{\min}$ are functions of $\rho_{\varepsilon}$ and $\rho_{W,j}$, respectively, we analyze the impact of varying \( \rho_{\varepsilon} \) and \( \rho_{W,j} \) separately on the estimation error. In this stage, we inspect two types of plots, varying $n$ while $q$ is fixed (Scenario S2A), and then varying $q$ while $n$ is fixed (Scenario S2B). The plots show the average estimation and standard errors (in shaded regions) over 100 replicated samples.  

\textbf{Scenario S2A: } We set $n=(200, 400, 800, 1600, 3200, 6400$, $12800)$, $p=100$, $q=10$ and $s_{\max}=5$. We test the following two scenarios by plotting $\VERT\widehat{\bTheta}_{\varepsilon\varepsilon}-\bTheta_{\varepsilon\varepsilon}^{*}\VERT_{\max}$ versus $\sqrt{n}$ shown in Figure \ref{fig:ERS2P3}:
        (a)   $\rho_{\varepsilon}=(0, 0.3, 0.7, 0.9)$ with $\rho_{W,j}=0.05$;
        and (b)   $\rho_{W,j}=(0.005,0.1,0.2,0.3)$ with $\rho_{\varepsilon}=0.7$. 

\textbf{Scenario S2B: } We set $q=(10, 20, 30, 40, 50)$, $n=400$, $p=100$ and $s_{\max}=5$. We test the following two scenarios by plotting $\VERT\widehat{\bTheta}_{\varepsilon\varepsilon}-\bTheta_{\varepsilon\varepsilon}^{*}\VERT_{\max}$
versus $\sqrt{\log(q^{2})}$ shown in Figure \ref{fig:ERS2P4}:
        (a)   $\rho_{\varepsilon}=(0, 0.3, 0.7, 0.9)$ with $\rho_{W,j}=0.05$;
        and (b)   $\rho_{W,j}=(0.005,0.1,0.2, 0.3)$ with $\rho_\varepsilon=0.7$.

     \begin{figure}[!h]
	\centering
        \includegraphics[scale=0.26]{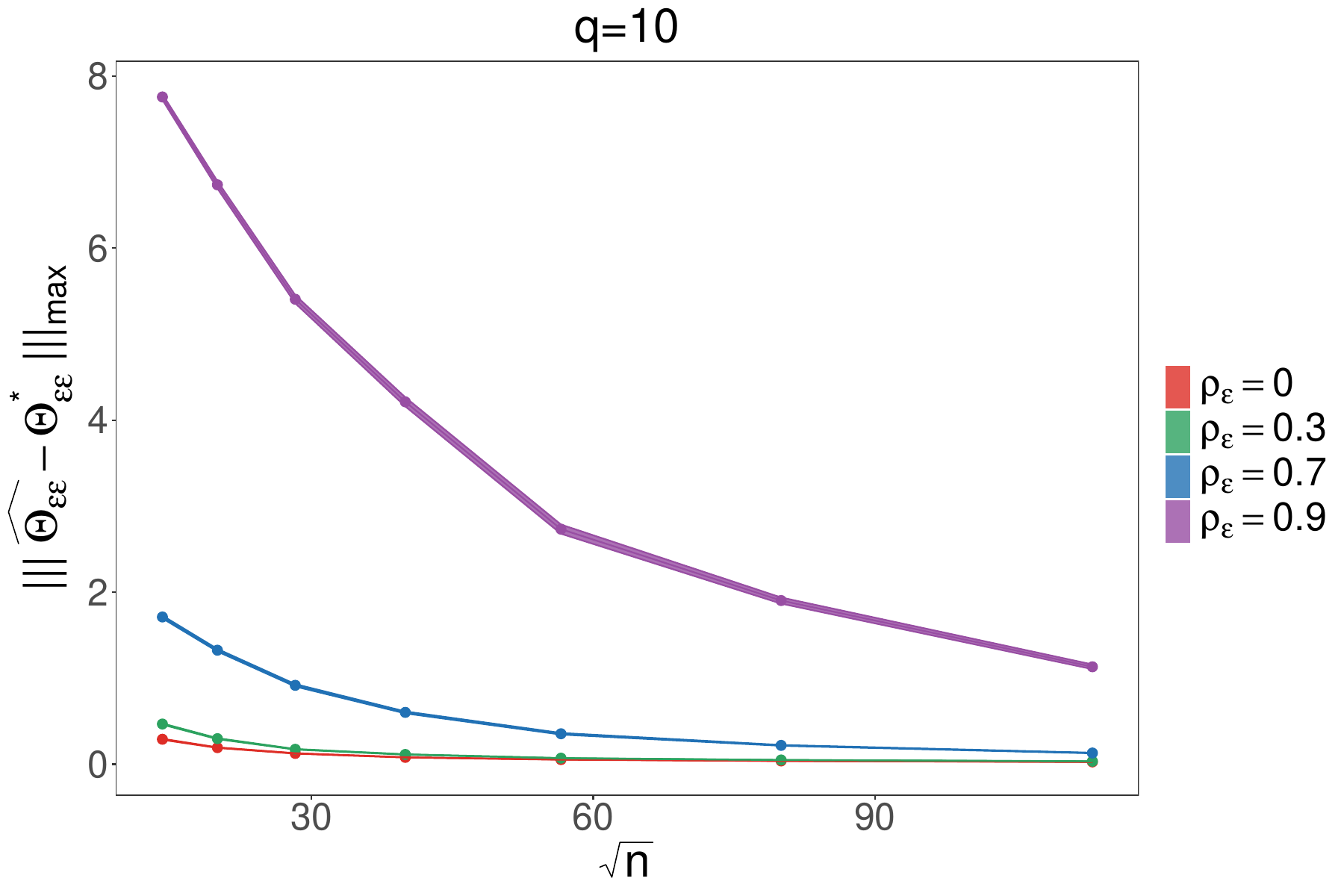}
        \includegraphics[scale=0.26]{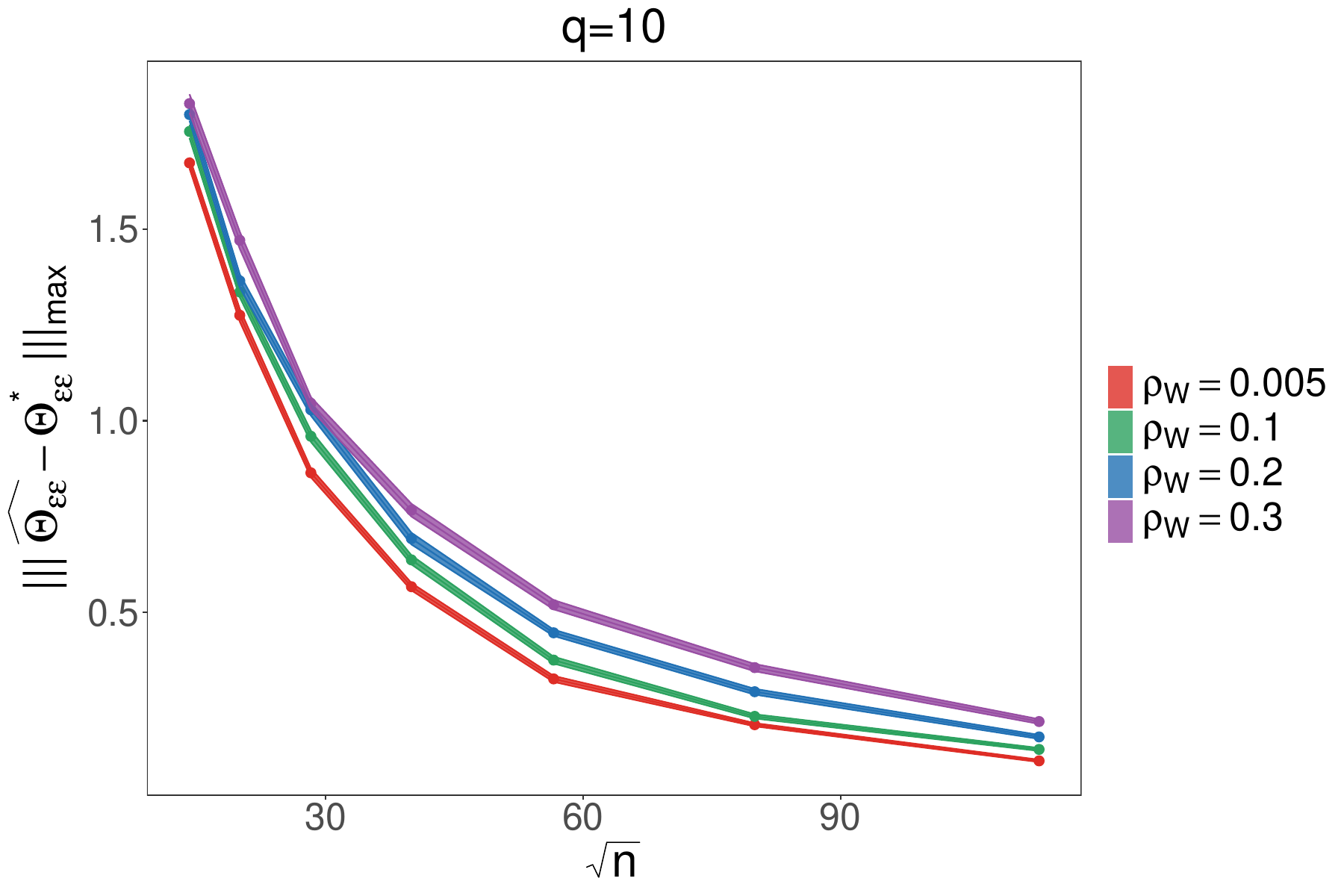}
	\caption{Scenario S2A. Plots of the average estimation error for $\bTheta_{\varepsilon\varepsilon}^{*}$, $\VERT\widehat{{\bTheta}}_{\varepsilon\varepsilon}-{\bTheta}^{*}_{\varepsilon\varepsilon}\VERT_{\mathrm{max}}$ against $\sqrt{n}$, $n=200, 400, 800, 1600, 3200, 6400$, and $12800$, for the outcome dimension, $q=10$, with for varying correlation between the errors, $\rho_{\varepsilon}$ (left panel) and the probability of being missing in each column of the outcomes, $\rho_{W,j}$ (right panel). Each point represents an average of 100 trials and the shaded regions indicate standard error for each method.} \label{fig:ERS2P3}
\end{figure}

     \begin{figure}[!h]
	\centering
        \includegraphics[scale=0.26]{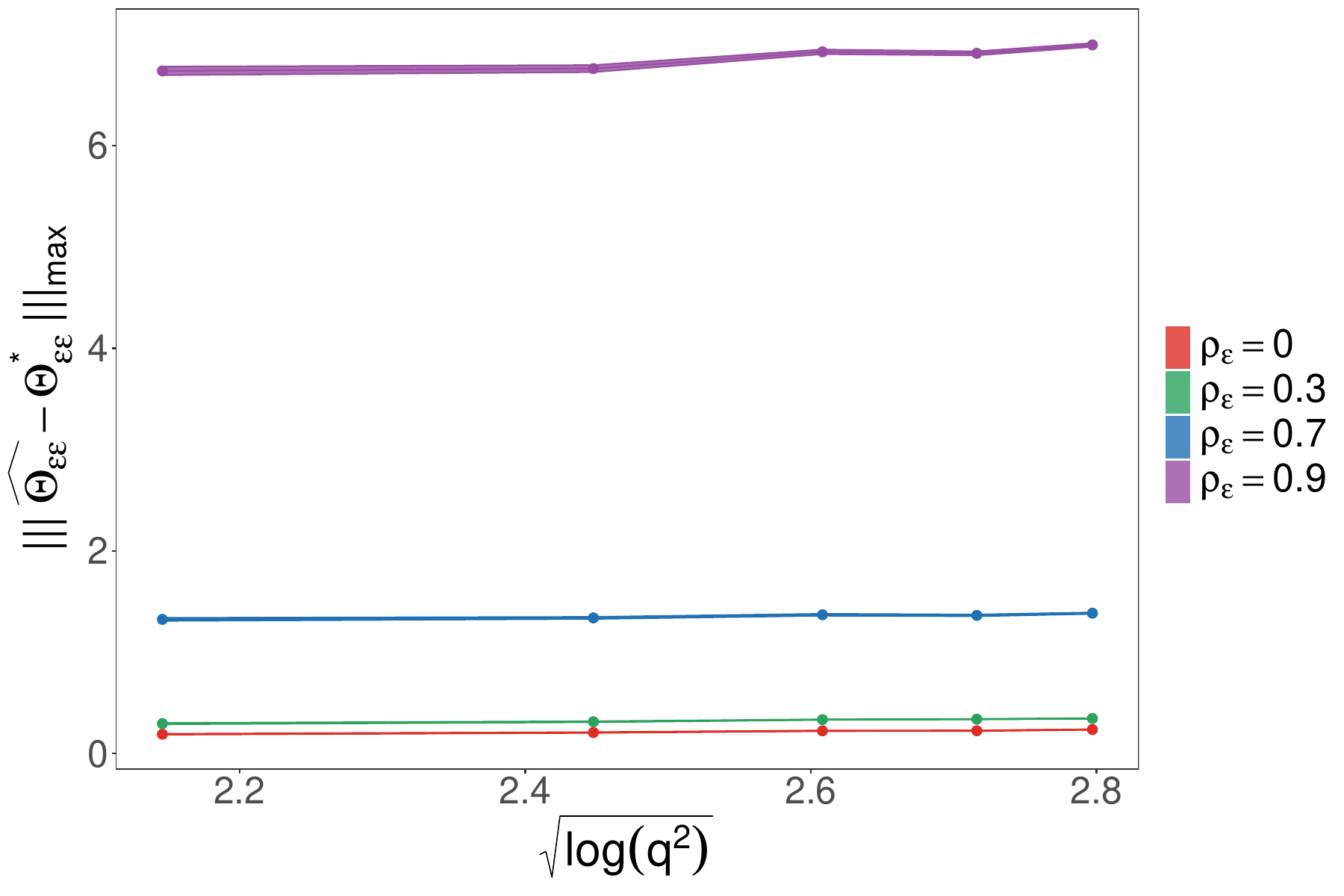}
        \includegraphics[scale=0.26]{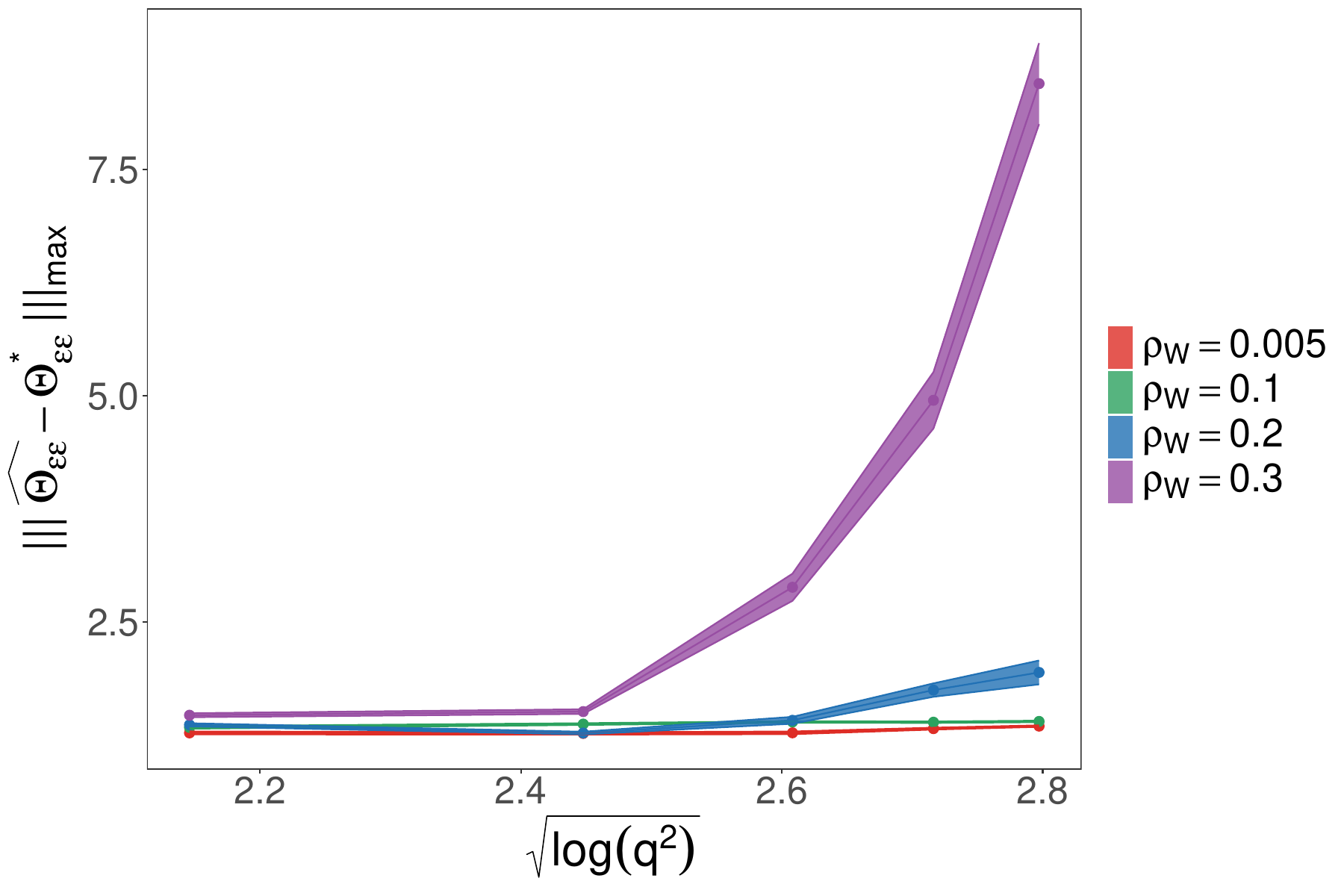}
	\caption{Scenario S2B. Plots of the average estimation error for $\bTheta_{\varepsilon\varepsilon}^{*}$, $\VERT\widehat{{\bTheta}}_{\varepsilon\varepsilon}-{\bTheta}^{*}_{\varepsilon\varepsilon}\VERT_{\mathrm{max}}$ against $\sqrt{log(q^2)}$, $q=10, 20, 30, 40$ and $50$, for sample size $n=400$, with varying correlation between the errors, $\rho_{\varepsilon}$ (left panel) and the probability of being missing in each column of the outcomes, $\rho_{W,j}$ (right panel). Each point represents an average of 100 trials and the shaded regions indicate standard error for each method.} \label{fig:ERS2P4}
\end{figure}

In Figure \ref{fig:ERS2P3}, we plot the estimation errors in $\bTheta^{*}_{\varepsilon\varepsilon}$ in terms of the elementwise max norm versus $\sqrt{n}$ depicting Scenario S2A. In Figure \ref{fig:ERS2P3}, there is a general decreasing trend in the curves of the estimation errors at the rate of $1/\sqrt{n}$ as $n$ increases, following the theoretical results by varying $\rho_{\varepsilon}$ (left panel) and $\rho_{W}$ (right panel), respectively. The correlation structure in the error, ranging from no correlation to strong correlation, has a noticeable impact on the estimation error. The estimation errors approach zero more quickly as $n$ becomes larger, especially when the correlation structure between the errors is weaker. In Figure \ref{fig:ERS2P3}, varying the probability of missingness in the outcomes shows that the estimation errors generally decrease as $n$ increases. As the degree of missingness increases, the estimation errors also increase, as expected.

Figure \ref{fig:ERS2P4} demonstrates Scenario S2B, the elementwise max norm estimation error for $\bTheta^{*}_{\varepsilon\varepsilon}$ versus $\sqrt{\log(q^2)}$, where $q$ is the dimension of the outcome. The estimation error curves exhibit a flat pattern as $q$ varies, showing no significant trend. However, it is harder to accurately estimate $\bTheta^{*}_{\varepsilon\varepsilon}$ when error correlations are stronger, characterized by varying $\rho_{\varepsilon}$ (left panel). In the plot depicting varying probabilities of missingness in the outcome (right panel), it is evident that as both the dimension of the outcome and the rate of missing data increase simultaneously, the estimation errors show a notable rise. Specifically, the estimation error accelerates notably when $q$ exceeds 20 and there is 30\% missing data in each outcome column.

\noindent  \textbf{Stage 3: } We aim to verify the Frobenius norm estimation error for the regression coefficients given in \eqref{eq:B_final_est} in Theorem \ref{thm:1}, where the upper bound depends on the tuning parameter $\lambda_{\bB}$ defined in \eqref{eq:lambda_overall}. We vary $[\bSigma^{*}_{\varepsilon\varepsilon}]_{i,j}=\rho_{\varepsilon}^{|i-j|}$, and \( \mu_{\min} = \min_{j}\mu_{j} = \min_{j}(1 - \rho_{W,j}) \) separately, and analyze the performance of the estimation error $\VERT\widehat{\bB}^{(2)}-\bB^{*}\VERT_{F}$ in terms of varying \( \rho_{\varepsilon} \) and \( \rho_{W,j} \). In this stage, we inspect two scenarios: (S3A) varying \( n \) with fixed \( p \), and (S3B) varying \( p \) with fixed \( n \). The plots present the average estimation error, with shaded regions indicating standard errors, based on 100 replicated samples.

\textbf{Scenario S3A:} We set $n=(200,400,800,1600,3200,6400,12800)$, $p=100$, $q=(10,20)$ and $s_{\max}=5$. We test the following two scenarios by plotting $\VERT\widehat{\bB}^{(2)}-\bB^{*}\VERT_{F}$ 
versus $\sqrt{n}$ shown in Figure \ref{fig:ERS1P5}:
        (a) we vary
$\rho_{\varepsilon}=(0,0.3,0.7,0.9)$ with $\rho_{W,j}=0.05$; 
       and (b) we vary $\rho_{W,j}=(0.005,0.1,0.2,0.3)$ with $\rho_{\varepsilon}=0.7$.

\textbf{Scenario S3B:} We set $p=(50,100,200,400,800)$, $n=400$, $q=(10,20)$ and $s_{\max}=5$. We test the following two scenarios by plotting $\VERT\widehat{\bB}^{(2)}-\bB^{*}\VERT_{F}$  versus $\sqrt{\log(p)}$ shown in Figure \ref{fig:ERS1P6}:
        (a) we vary $\rho_{\varepsilon}=(0, 0.3, 0.7, 0.9)$ with $\rho_{W,j}=0.05$;
        and (b) we vary $\rho_{W,j}=(0.005,0.1,0.2, 0.3)$ with $\rho_{\varepsilon}=0.7$.

     \begin{figure}[!h]
	\centering
	\includegraphics[scale=0.26]{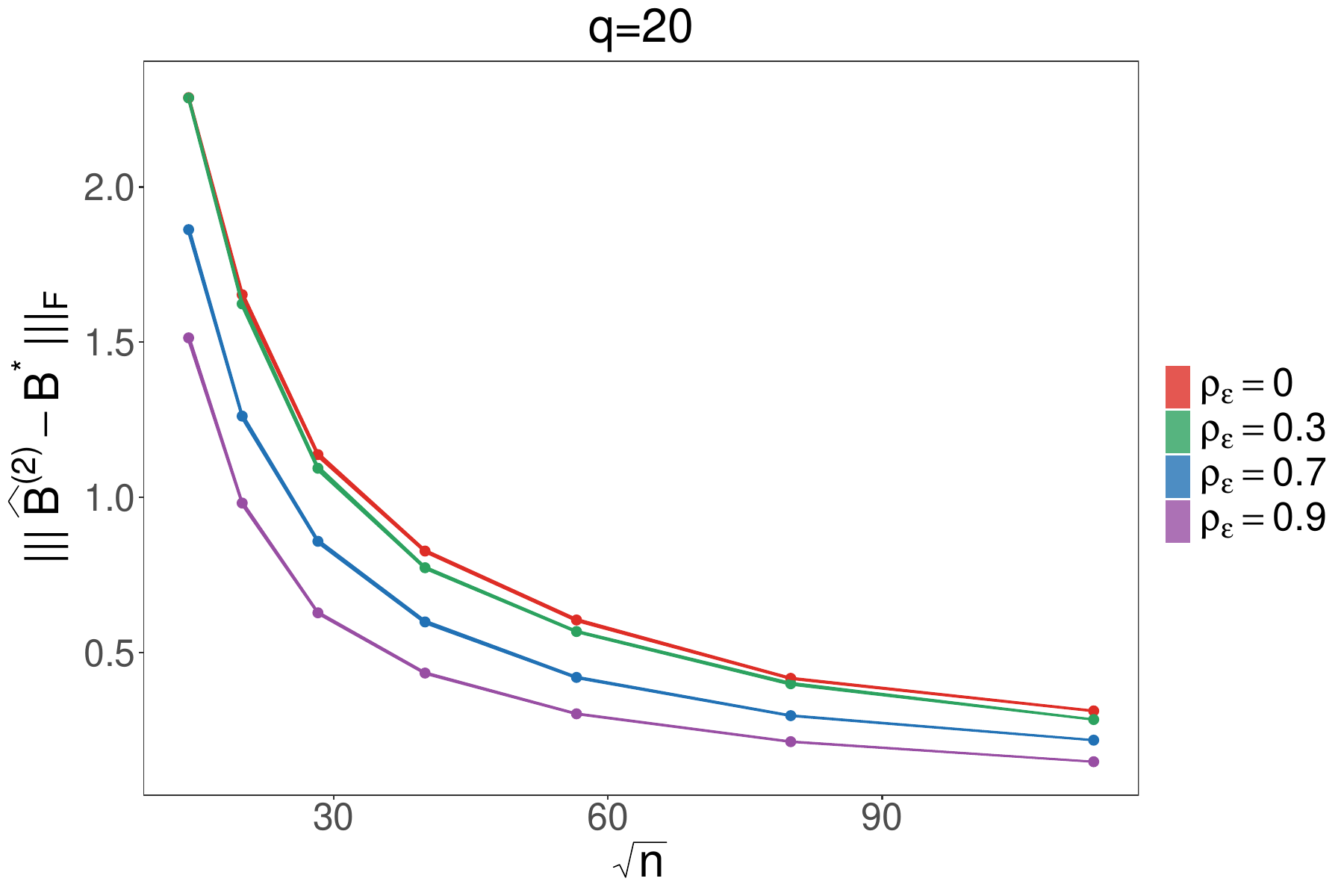}
        \includegraphics[scale=0.26]{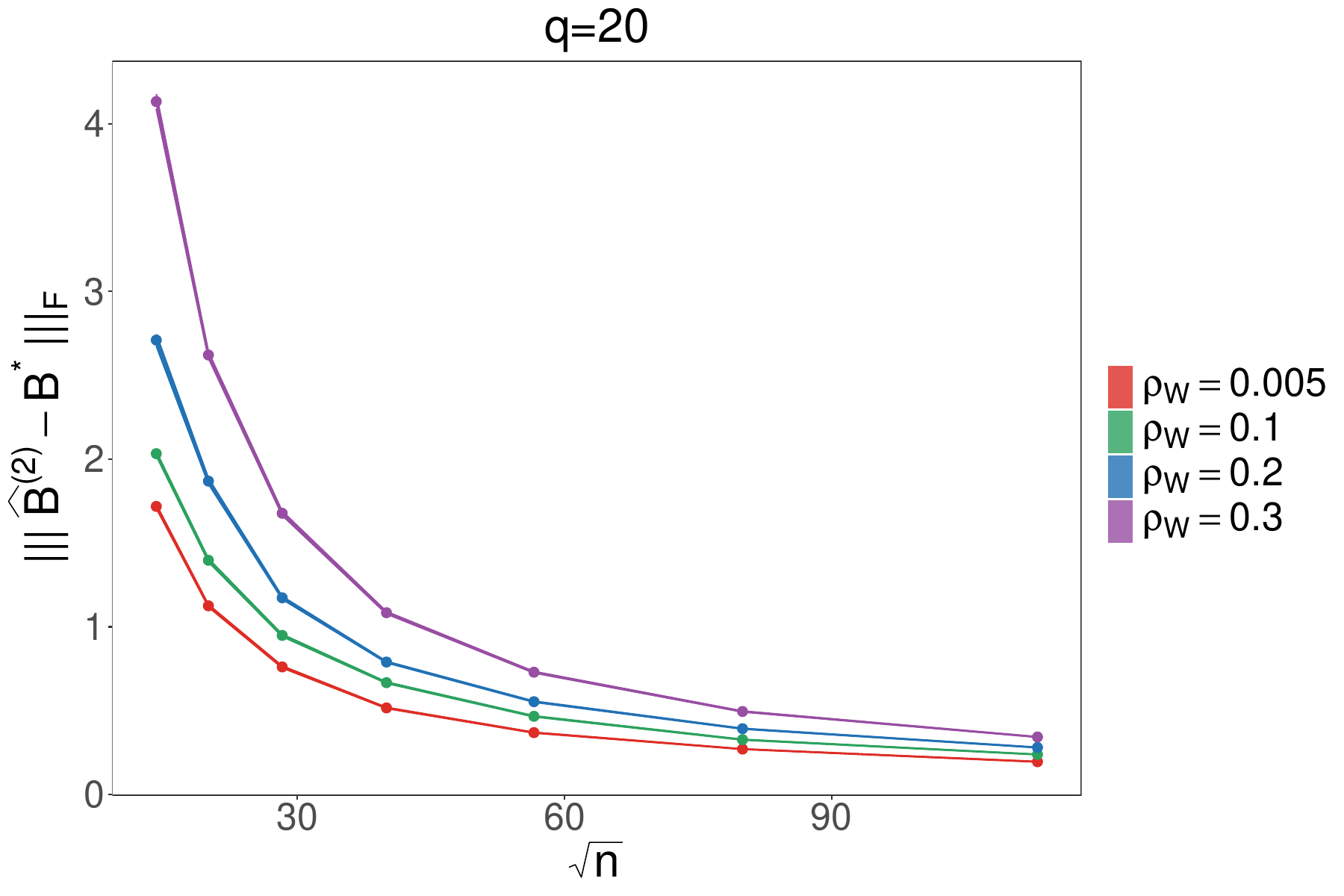}
	\caption{Scenario S3A. Plots of the average estimation error for $\bB^{*}$, $\VERT\widehat{\bB}^{(2)}-\bB^{*}\VERT_{F}$  against $\sqrt{n}$, $n=200,400,800,1600,3200,6400,12800$, for outcome dimension, $q=20$, with varying correlation between the errors, $\rho_{\varepsilon}$ (left panel) and the probability of being missing in the $j$th column for the
outcome $\rho_{W,j}$ (right panel).  Each point represents an average of 100 trials and the shaded regions indicate standard error for each method.} \label{fig:ERS1P5}
\end{figure} 
\begin{figure}[!h]
	\centering
	\includegraphics[scale=0.26]{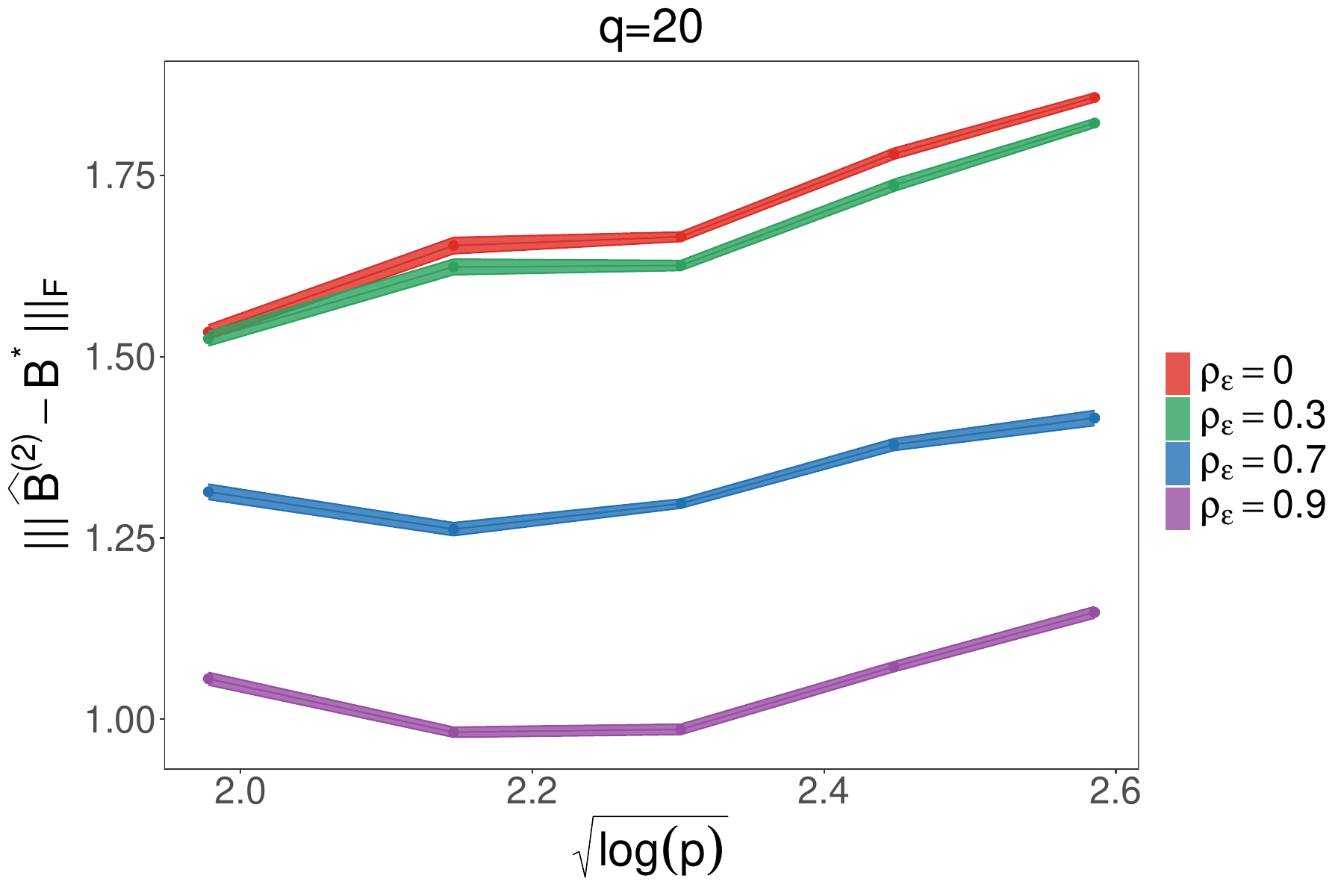}
        \includegraphics[scale=0.26]{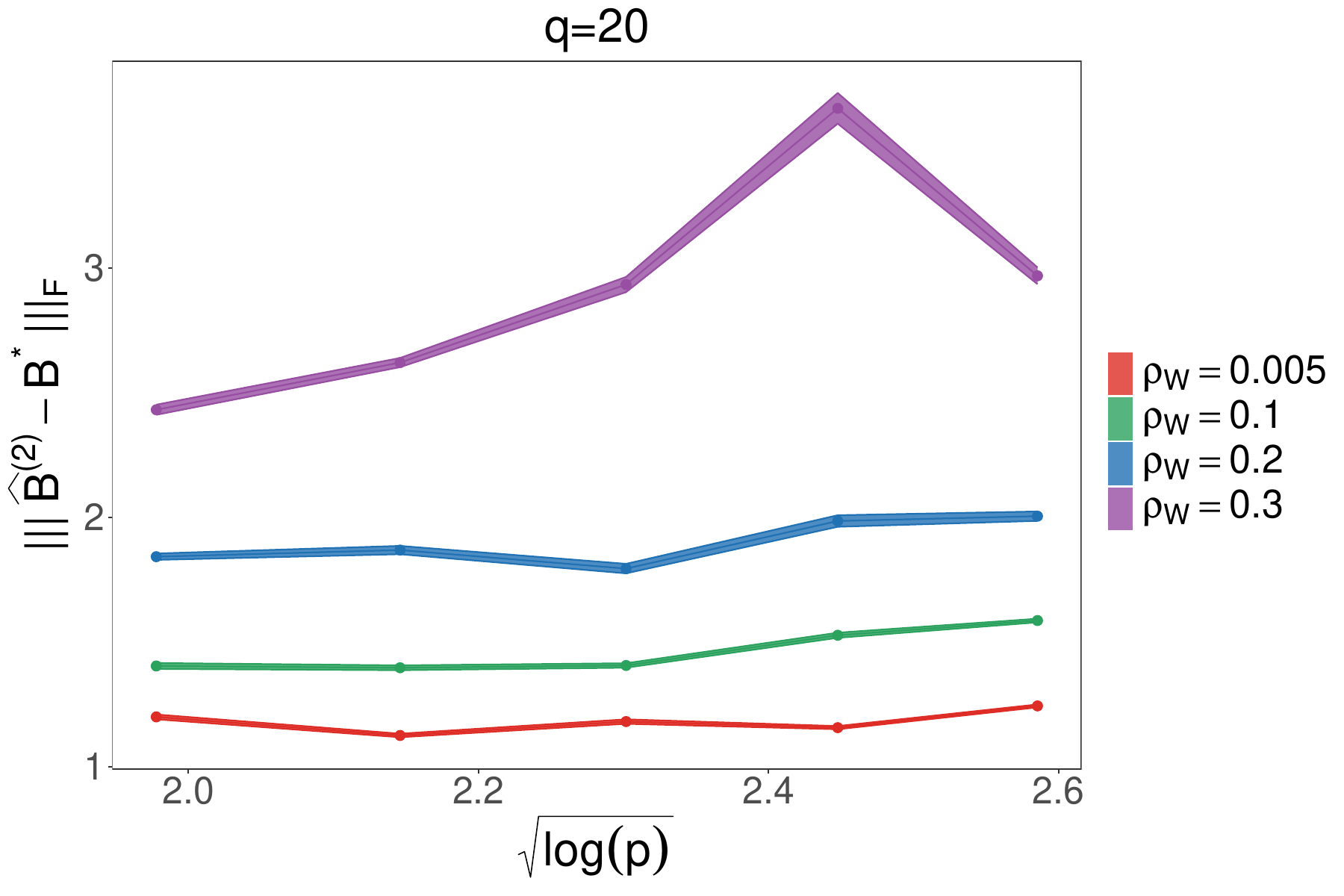}\\
	\caption{Scenario S3B. Plots of the average estimation error for $\bB^{*}$, $\VERT\widehat{\bB}^{(2)}-\bB^{*}\VERT_{F}$ against $\sqrt{\log(p)}$, $p=50,100,200,400$ and $800$, for outcome dimension, $q = 20$, with varying correlation between the errors, $\rho_{\varepsilon}$ (left panel) and the probability of being missing in each column of the outcomes, $\rho_{W,j}$ (right panel). Each point represents an average of 100 trials and the shaded regions indicate standard error for each method.} \label{fig:ERS1P6}
\end{figure}    

Figure \ref{fig:ERS1P5} illustrates Scenario S3A, the final stage Frobenius norm estimation error of $\bB^{*}$ against $\sqrt{n}$. The estimation error bound aligns with the theoretical result, which is bounded by $1/\sqrt{n}$, keeping other parameters constant. The estimation error $\VERT\widehat{\bB}^{(2)}-\bB^{*}\VERT_{F}$ approaches zero as $n$ becomes larger for varying degrees of error correlations. Unlike the Stage 1 estimation plots shown in Figure \ref{fig:ERS1P1} (left panel), we observe a visible separation in the curves for different degrees of error correlation. As the correlation strengthens, the estimation errors tend to decrease with increasing $n$. We also observe that the estimation error curves decrease at the rate of $1/\sqrt{n}$ as we vary the probability of missingness in the outcome (right panel), resulting in lower estimation error for low degrees of missingness, as expected. Figure \ref{fig:ERS1P5b} in \ref{q10Sim} provides $q=10$ scenario.

The final stage Frobenius norm estimation error of $\bB^{*}$ is plotted against $\sqrt{\log(p)}$ in Figure \ref{fig:ERS1P6} depicting Scenario S3B. Theoretically, the Frobenius norm error is upper bounded by $\sqrt{\log(p)}$, keeping other parameters constant. In the left plot, we observe a clear distinction in estimation error with varying error correlation, unlike in Figure \ref{fig:ERS1P2}. As the correlation becomes stronger, the estimation errors tend to decrease. In the right plot, we vary the probability of missingness in the outcomes and observe that with higher probability of missingness, the estimation error spikes significantly when the number of covariates is larger. Figure \ref{fig:ERS1P6b} in \ref{q10Sim} provides $q=10$ scenario.

\begin{table}[ht]
    \centering
    \begin{minipage}{0.48\linewidth}
        \centering
        \renewcommand{\arraystretch}{1.0}
        \caption{Frobenius norm error for setups with varying $n$ and $\rho_{\varepsilon}$ with $p=100$, $q=20$, $\rho_{W,j} = 0.05$, and $s_{\max}=5$. Averages and standard errors in parentheses are based on 100 replications.}
        \begin{tabular}{cccc}
            \hline
            $n$ & $\rho_{\varepsilon}$ & $\VERT \widehat{\bB}^{(1)} - \bB^{*} \VERT_{F}$ & $\VERT \widehat{\bB}^{(2)} - \bB^{*} \VERT_{F}$ \\ \hline
            200   & 0   & 2.558 (0.018) & 2.288 (0.011) \\
                  & 0.3 & 2.577 (0.020) & 2.287 (0.012) \\
                  & 0.7 & 2.546 (0.018) & 1.863 (0.014) \\
                  & 0.9 & 2.546 (0.021) & 1.514 (0.013) \\ \hline
            400   & 0   & 1.858 (0.011) & 1.653 (0.011) \\
                  & 0.3 & 1.853 (0.011) & 1.624 (0.011) \\
                  & 0.7 & 1.877 (0.009) & 1.262 (0.009) \\
                  & 0.9 & 1.841 (0.013) & 0.982 (0.007) \\ \hline
            800   & 0   & 1.274 (0.009) & 1.138 (0.007) \\
                  & 0.3 & 1.275 (0.009) & 1.094 (0.008) \\
                  & 0.7 & 1.288 (0.010) & 0.859 (0.006) \\
                  & 0.9 & 1.277 (0.010) & 0.629 (0.004) \\ \hline
            1600  & 0   & 0.952 (0.004) & 0.828 (0.004) \\
                  & 0.3 & 0.952 (0.005) & 0.774 (0.004) \\
                  & 0.7 & 0.952 (0.004) & 0.599 (0.005) \\
                  & 0.9 & 0.955 (0.005) & 0.435 (0.003) \\ \hline
            3200  & 0   & 0.675 (0.006) & 0.606 (0.004) \\
                  & 0.3 & 0.668 (0.006) & 0.569 (0.003) \\
                  & 0.7 & 0.665 (0.006) & 0.421 (0.003) \\
                  & 0.9 & 0.667 (0.006) & 0.304 (0.002) \\ \hline
            6400  & 0   & 0.491 (0.002) & 0.418 (0.002) \\
                  & 0.3 & 0.491 (0.002) & 0.400 (0.004) \\
                  & 0.7 & 0.489 (0.002) & 0.298 (0.002) \\
                  & 0.9 & 0.487 (0.002) & 0.214 (0.002) \\ \hline
            12800 & 0   & 0.359 (0.003) & 0.313 (0.002) \\
                  & 0.3 & 0.356 (0.003) & 0.285 (0.001) \\
                  & 0.7 & 0.357 (0.003) & 0.219 (0.001) \\
                  & 0.9 & 0.351 (0.004) & 0.149 (0.001) \\ \hline
        \end{tabular}
        \label{tab:q20_rho_eps}
    \end{minipage}%
    \hfill
    \begin{minipage}{0.48\linewidth}
        \centering
        \renewcommand{\arraystretch}{1.0}
        \caption{Frobenius norm error for setups with varying $n$ and $\rho_{W,j}$ with $p=100$, $q=20$, $\rho_{\varepsilon} = 0.7$, and $s_{\max}=5$. Averages and standard errors in parentheses are based on 100 replications.}
        \begin{tabular}{cccc}
            \hline
            $n$ & $\rho_{W,j}$ & $\VERT \widehat{\bB}^{(1)} - \bB^{*}\VERT_{F}$ & $\VERT \widehat{\bB}^{(2)} - \bB^{*}\VERT_{F}$ \\ \hline
            200   & 0.005 & 2.465 (0.015) & 1.718 (0.011) \\
                  & 0.1   & 2.644 (0.020) & 2.032 (0.013) \\
                  & 0.2   & 2.886 (0.017) & 2.710 (0.035) \\
                  & 0.3   & 3.048 (0.011) & 4.132 (0.044) \\ \hline
            400   & 0.005 & 1.794 (0.015) & 1.125 (0.007) \\
                  & 0.1   & 1.928 (0.007) & 1.397 (0.010) \\
                  & 0.2   & 2.012 (0.011) & 1.869 (0.017) \\
                  & 0.3   & 2.280 (0.018) & 2.621 (0.018) \\ \hline
            800   & 0.005 & 1.236 (0.007) & 0.761 (0.006) \\
                  & 0.1   & 1.343 (0.011) & 0.949 (0.007) \\
                  & 0.2   & 1.516 (0.009) & 1.173 (0.008) \\
                  & 0.3   & 1.599 (0.006) & 1.677 (0.012) \\ \hline
            1600  & 0.005 & 0.922 (0.006) & 0.516 (0.003) \\
                  & 0.1   & 0.973 (0.004) & 0.667 (0.004) \\
                  & 0.2   & 1.034 (0.007) & 0.790 (0.004) \\
                  & 0.3   & 1.188 (0.008) & 1.084 (0.007) \\ \hline
            3200  & 0.005 & 0.639 (0.005) & 0.369 (0.002) \\
                  & 0.1   & 0.724 (0.006) & 0.466 (0.003) \\
                  & 0.2   & 0.784 (0.003) & 0.553 (0.003) \\
                  & 0.3   & 0.825 (0.004) & 0.729 (0.004) \\ \hline
            6400  & 0.005 & 0.478 (0.003) & 0.271 (0.001) \\
                  & 0.1   & 0.499 (0.002) & 0.327 (0.002) \\
                  & 0.2   & 0.551 (0.005) & 0.392 (0.002) \\
                  & 0.3   & 0.613 (0.005) & 0.495 (0.003) \\ \hline
            12800 & 0.005 & 0.336 (0.003) & 0.195 (0.002) \\
                  & 0.1   & 0.377 (0.003) & 0.238 (0.001) \\
                  & 0.2   & 0.399 (0.002) & 0.280 (0.001) \\
                  & 0.3   & 0.419 (0.002) & 0.343 (0.002) \\ \hline
        \end{tabular}
        \label{tab:q20_rhoWj}
    \end{minipage}
\end{table}

\begin{table}[ht]
    \centering
    \begin{minipage}{0.48\linewidth}
        \centering
        \renewcommand{\arraystretch}{1.0}
        \caption{Frobenius norm error for setups with varying $p$ and $\rho_{\varepsilon}$ with $n=400$, $q=20,$ $\rho_{W,j} = 0.05$ and $s_{\max}=5$. Averages and standard errors in parentheses are based on 100 replications.}
        \begin{tabular}{cccc}
            \hline
            $p$ & $\rho_{\varepsilon}$ & $\VERT \widehat{\bB}^{(1)} - \bB^{*}\VERT_{F}$ & $\VERT \widehat{\bB}^{(2)} - \bB^{*}\VERT_{F}$ \\ \hline
        50  & 0   & 1.810 (0.012) & 1.534 (0.010) \\
            & 0.3 & 1.815 (0.012) & 1.525 (0.010) \\
            & 0.7 & 1.795 (0.012) & 1.314 (0.010) \\
            & 0.9 & 1.830 (0.015) & 1.056 (0.009) \\ \hline
        100 & 0   & 1.858 (0.011) & 1.653 (0.011) \\
            & 0.3 & 1.853 (0.011) & 1.624 (0.011) \\
            & 0.7 & 1.877 (0.009) & 1.262 (0.009) \\
            & 0.9 & 1.841 (0.013) & 0.982 (0.007) \\ \hline
        200 & 0   & 1.879 (0.006) & 1.665 (0.007) \\
            & 0.3 & 1.892 (0.010) & 1.626 (0.007) \\
            & 0.7 & 1.872 (0.005) & 1.297 (0.006) \\
            & 0.9 & 1.881 (0.010) & 0.986 (0.007) \\ \hline
        400 & 0   & 1.887 (0.016) & 1.780 (0.008) \\
            & 0.3 & 1.874 (0.014) & 1.736 (0.008) \\
            & 0.7 & 1.884 (0.016) & 1.379 (0.008) \\
            & 0.9 & 1.915 (0.017) & 1.072 (0.007) \\ \hline
        800 & 0   & 2.148 (0.004) & 1.857 (0.006) \\
            & 0.3 & 2.147 (0.004) & 1.822 (0.006) \\
            & 0.7 & 2.144 (0.005) & 1.415 (0.010) \\
            & 0.9 & 2.144 (0.005) & 1.147 (0.008) \\ \hline
        \end{tabular}
        \label{tab:q20_rho_eps_1}
    \end{minipage}%
    \hfill
    \begin{minipage}{0.48\linewidth}
        \centering
        \renewcommand{\arraystretch}{1.0}
        \caption{Frobenius norm error for setups with varying $p$ and $\rho_{W,j}$ with $n=400$, $q=20$, $\rho_{\varepsilon} = 0.7$ and $s_{\max}=5$. Averages and standard errors in parentheses are based on 100 replications.}
        \begin{tabular}{cccc}
            \hline
            $p$ & $\rho_{W,j}$ & $\VERT \widehat{\bB}^{(1)} - \bB^{*} \VERT_{F}$ & $\VERT \widehat{\bB}^{(2)} - \bB^{*} \VERT_{F}$ \\ \hline
        50  & 0.005 & 1.745 (0.008) & 1.199 (0.010) \\
            & 0.1   & 1.871 (0.015) & 1.404 (0.012) \\
            & 0.2   & 2.010 (0.015) & 1.842 (0.013) \\
            & 0.3   & 2.193 (0.012) & 2.432 (0.020) \\ \hline
        100 & 0.005 & 1.794 (0.015) & 1.125 (0.007) \\
            & 0.1   & 1.928 (0.007) & 1.397 (0.010) \\
            & 0.2   & 2.012 (0.011) & 1.869 (0.017) \\
            & 0.3   & 2.280 (0.018) & 2.621 (0.018) \\ \hline
        200 & 0.005 & 1.853 (0.005) & 1.181 (0.009) \\
            & 0.1   & 1.932 (0.012) & 1.406 (0.009) \\
            & 0.2   & 2.261 (0.014) & 1.794 (0.019) \\
            & 0.3   & 2.363 (0.009) & 2.934 (0.030) \\ \hline
        400 & 0.005 & 1.798 (0.006) & 1.156 (0.006) \\
            & 0.1   & 2.093 (0.017) & 1.527 (0.010) \\
            & 0.2   & 2.244 (0.007) & 1.985 (0.023) \\
            & 0.3   & 2.355 (0.014) & 3.641 (0.061) \\ \hline
        800 & 0.005 & 2.083 (0.013) & 1.244 (0.006) \\
            & 0.1   & 2.162 (0.005) & 1.586 (0.007) \\
            & 0.2   & 2.241 (0.010) & 2.005 (0.019) \\
            & 0.3   & 2.664 (0.010) & 2.969 (0.033) \\ \hline
        \end{tabular}
        \label{tab:q20_rhoWj_1}
    \end{minipage}
\end{table}

Finally, in Tables \ref{tab:q20_rho_eps}, \ref{tab:q20_rhoWj}, \ref{tab:q20_rho_eps_1}, and \ref{tab:q20_rhoWj_1}, we present the average Frobenius norm estimation errors for $\bB^{*}$ in Stage 1, ${\VERT\widehat{\bB}^{(1)}-\bB^{*}\VERT_{F}}$, and Stage 3, ${\VERT\widehat{\bB}^{(2)}-\bB^{*}\VERT_{F}}$ over 100 replications, including their standard errors for scenarios S3A and S3B for $q=20$. Results for $ q=10$ are provided in Table \ref{tab:q10_rho_eps}, \ref{tab:q10_rhoWj}, \ref{tab:q10_rho_eps_1}, and \ref{tab:q10_rhoWj_1} in \ref{q10Sim}. In Table \ref{tab:q20_rho_eps}, where $n$ and $\rho_\varepsilon$ are varied, there is a consistent improvement in the estimation error in the final stage across all values of $n$ and $\rho_\varepsilon$ compared to the initial stage.  The standard errors, shown in parentheses, are also reduced in the final estimation of $\bB^{*}$. In Table \ref{tab:q20_rhoWj}, with variations in $n$ and $\rho_{W,j}$, a similar trend is observed when 20\% of the data is missing. For smaller sample sizes (below $n=800$) and 30\% missing data in each column of the outcomes, the performance of $\widehat{\bB}^{(2)}$ is worse than that of $\widehat{\bB}^{(1)}$, but it improves as $n$ increases. Table \ref{tab:q20_rho_eps_1} shows the effects of varying $p$ and $\rho_{\varepsilon}$, where substantial gaps between the first stage and final stage estimation errors are evident in all cases. Similar trends are observed for $q=10$ scenarios as well. In Table \ref{tab:q20_rhoWj_1}, with variations in $p$ and $\rho_{W,j}$, it is observed that as the number of $p$ and $q$ increases along with a higher degree of missingness (above 30\%), the performance of $\widehat{\bB}^{(2)}$ deteriorates compared to $\widehat{\bB}^{(1)}$.

\subsection{Simulations to compare \texttt{missoNet} with other methods \label{missVother}}

We compared our method to several established approaches, including Multivariate Regression with Covariance Estimation (\textit{MRCE}) \citep{rothman2010sparse}, Conditional Graphical Lasso with Missing Values (\textit{cglasso}) \citep{augugliaro2023cglasso}, and separate \textit{Lasso} fits by CpG.

Following \citet{rothman2010sparse}, we generated the true sparse coefficient matrix $\mathbf{B}^{*}$ as follows:
\begin{align}
\mathbf{B}^{*} = \mathbf{B} \odot \mathbf{K} \odot \mathbf{R}, \nonumber
\end{align}
where $\odot$ denotes elementwise multiplication. The matrix $\mathbf{B} \in \mathbb{R}^{p\times q}$ is dense, with elements uniformly distributed on [0.3, 0.7] and element signs assigned based on a Bernoulli trial with probability 0.5. The sparsity in $\mathbf{B}^{*}$ is induced by the matrices $\mathbf{K}$ and $\mathbf{R}$. Specifically, $\mathbf{K}$ is an elementwise binary mask, with each value drawn independently from a Bernoulli distribution with nonzero probability $s_1$. Row-wise sparsity is introduced through $\mathbf{R}$, with each row being entirely zeros or ones based on a Bernoulli trial with nonzero probability $s_2$. This combination yields both elementwise and row-wise sparsity, with an expected $s_1 s_2 pq$ nonzero entries in $\mathbf{B}^{*}$.

The predictor matrix, $\bX$, of size $n \times p$ is generated with rows drawn independently from a multivariate Gaussian distribution $\mathcal{N}_{p}(\mathbf{0},\bSigma^{*}_{xx})$, where $[\bSigma^{*}_{xx}]_{i,j}=0.7^{|i-j|}$ represents an AR(1) covariance structure. The response matrix, $\bY \in \mathbb{R}^{n \times q}$, is sampled from $\mathcal{N}_{q}(\bX\bB^{*},\bSigma^{*}_{\varepsilon\varepsilon})$.

We considered two inverse covariance structures for errors:
\begin{enumerate}
    \item \textbf{Type 1}: The inverse of an AR(1) error covariance $[\bSigma^{*}_{\varepsilon\varepsilon}]_{i,j}=\rho_{\varepsilon}^{|i-j|}$, corresponding to a chain structure with sparse tri-diagonal precision matrix $\bTheta^{*}_{\varepsilon\varepsilon}=(\bSigma^{*}_{\varepsilon\varepsilon})^{-1}$.
    
    \item \textbf{Type 2}: A more complex structure combining an inverse AR(1) matrix with a block-diagonal structure, where each block follows one of the three graph types: (i) an independent vertex set, (ii) a weak complete graph (with edge weights drawn from the uniform distribution on $[0.1, 0.4]$, or (iii) a strong complete graph (with edge weights drawn from the uniform distribution on $[0.5, 1]$.
\end{enumerate}
The details of the Type 1 and Type 2 structures are provided in \ref{precison_struc}. To introduce missing data in $\bY$, we generated a Bernoulli matrix $\bW$ with missingness probability $\rho_{W}$ for each column, and the observed responses were computed as $\bZ = \bY \odot \bW$. We considered the following models by varying the sample size, the number of variables, the conditional dependency strength, the network structure across responses, and the missing data probability.
\begin{enumerate}
    \item \textbf{Model 1}: $n=400$, $(p,q) = (30,30)$, $(s_1,s_2) = (0.2, 0.2)$, Type 1 structure for $\bTheta^{*}_{\varepsilon\varepsilon}$ with $\rho_{\varepsilon}=0.7$, and $\rho_{W} = \{0.01, 0.10, 0.20\}$.
    
    \item \textbf{Model 2}:  $n=200$; $(p,q) = (60,60)$; $(s_1,s_2) = (0.1, 0.1)$, Type 1 structure for $\bTheta^{*}_{\varepsilon\varepsilon}$ with $\rho_{\varepsilon}=0.7$, and $\rho_{W} = \{0.01, 0.10, 0.20\}$. This model evaluates robustness under increased dimensionality.
        
    \item \textbf{Model 3}: $n=200$; $(p,q) = (30,30)$; $(s_1,s_2) = (0.2, 0.2)$, Type 1 structure for $\bTheta^{*}_{\varepsilon\varepsilon}$ with $\rho_{\varepsilon}=0.4$ (weaker dependencies), and $\rho_{W} = \{0.01, 0.10, 0.20\}$.
    
    \item \textbf{Model 4}: $n=200$; $(p,q) = (30,30)$; $(s_1,s_2) = (0.2, 0.2)$. Type 2 structure for $\bTheta^{*}_{\varepsilon\varepsilon}$, and $\rho_{W} = \{0.01, 0.10, 0.20\}$.
\end{enumerate}

Model 1 serves as a baseline with a large sample size and strong dependencies within a simple network structure. Models 2-4 reduce the sample size to amplify the effect of missing data. Model 2 increases the dimensionality, while Model 3 reduces dependency strength to assess the benefits of multivariate methods under weaker conditional dependencies. Model 4 evaluates the performance of the methods under complex interdependencies.

Each simulation was repeated 100 times for stability and reliability.

\subsubsection{Results \label{Comp_res}}

The detailed results of the simulations for Models 1-4, each with varying degrees of missing data ($\rho_{W}$), are presented in \ref{Res_comp}. Tables \ref{ch1:tab_model1}–\ref{ch1:tab_model4} present the means and standard errors for the key evaluation metrics, including overall parameter-estimation accuracy (Prediction Error, PE; Kullback–Leibler Loss, KLL) and the ability to recover sparsity structures (True Positive Rate, TPR; True Negative Rate, TNR; Matthews Correlation Coefficient, MCC), all defined in Section \ref{ME_missVother}.

Further details on the tuning parameter selection strategies can be found in Section~\ref{Tuning_missVother}. Across all tested strategies
({cv.min}, {cv.1se}, {BIC}) and levels of missingness (1\%, 10\%, 20\%), \texttt{missoNet} consistently exhibits strong performance: 

\textbf{Prediction accuracy.} For parameter estimation accuracy, \texttt{missoNet} [cv.min] consistently achieved the lowest Prediction Error (PE) for $\bB^{*}$, outperforming all competitors.

\textbf{Sparsity recovery.} Regarding sparsity recovery, although \texttt{missoNet} [cv.min] shows a slightly higher false positive rate, it still surpasses \textit{Lasso} [cv.min] and \textit{MRCE} [cv.min] in terms of TNR and MCC for \(\bB^{*}\). \texttt{missoNet} configured with cv.1se or BIC excelled, delivering the highest or near-highest TPR, TNR, and MCC for  \(\mathbf{B}^{*}\).  While \textit{Lasso} [cv.1se] occasionally yields a slightly
higher TNR, it is generally overly conservative, leading to substantially lower TPR and up to tenfold higher PE (e.g., in Model~1) compared to \texttt{missoNet}. This suggests a considerable risk of omitting important features and reducing predictive accuracy when employing \textit{Lasso} in complex settings.

\textbf{Network estimation.} In estimating the precision matrix  $\bTheta_{\varepsilon\varepsilon}^*$, \texttt{missoNet} [BIC]
achieves the lowest KLL in most settings, outperforming all methods except
itself with cross-validation\footnote{The performance difference between
\texttt{missoNet} [cv.min] and \texttt{missoNet} [cv.1se] when estimating
$\bTheta_{\varepsilon\varepsilon}^*$ is minimal because the one-standard-error
rule applies only to $\lambda_{\bB}$ rather than $\lambda_{\bTheta}$, consistent
with the definition of \textit{Lasso} [cv.1se].} and \textit{cglasso} in certain
high-missingness cases (e.g., Models~1 and~4). Nevertheless,
\texttt{missoNet} [BIC] consistently outperforms \textit{MRCE}. BIC tuning also
appears more sensitive to $\lambda_{\bTheta}$ and is more affected by missing
data than cross-validation. As model complexity and missingness increase,
cross-validation-tuned models often surpass those tuned via BIC, a trend that
later real-data analyses further support. Although \textit{cglasso} [BIC], which
uses an EM algorithm, maintains relatively stable estimation for
$\bTheta_{\varepsilon\varepsilon}^*$ in high-missingness scenarios, it tends to
select overly complex network structures that compromise interpretability.

\textbf{Network sparsity recovery.} In terms of sparsity recovery for $\bTheta_{\varepsilon\varepsilon}^*$, \texttt{missoNet} performs robustly under both cross-validation and BIC tuning. Generally, the BIC setting provides stronger results, while cross-validation remains competitive or superior across all missing-data levels. Although \texttt{missoNet} [BIC] is somewhat sensitive to higher missingness when measuring KLL, it remains reliable in edge identification for $\bTheta_{\varepsilon\varepsilon}^*$. \textit{MRCE} shows moderate performance for selecting variables in $\bTheta_{\varepsilon\varepsilon}^*$, occasionally surpassing \texttt{missoNet} under low missingness (e.g., 1\% in Model~1), but incurs markedly higher KLL (approximately fourfold that of \texttt{missoNet}). Meanwhile, \textit{cglasso} [BIC] often chooses a $\lambda_{\bTheta}$ that yields an overly complex network, lowering TNR and MCC relative to other methods and indicating overfitting. This issue also appears in subsequent real-data applications, emphasizing the importance of carefully tuning regularization parameters for \textit{cglasso} in practice.

\textbf{Computation efficiency.} \texttt{missoNet} [BIC] is significantly
faster than other multivariate methods. Although \texttt{missoNet} with
cross-validation requires multiple sub-model evaluations and is therefore
slower, it remains competitive. Meanwhile, \textit{cglasso} is generally slower
than \texttt{missoNet}, but can be more efficient in higher-dimensional
scenarios (e.g., Model~2) thanks to its Fortran-based core, which suits
large-scale matrix operations. By contrast, \textit{MRCE} is noticeably slower
in all settings, especially with many variables, due to its iterative estimation
procedure and less-optimized implementation.

In summary, \textit{Lasso} offers fast computation, making it a practical choice
for preliminary dimensionality reduction in ultrahigh-dimensional datasets.
However, its conservative selection strategy, limited predictive performance,
and inability to estimate network structures constrain its suitability for
comprehensive multivariate modeling across multiple related tasks. Meanwhile,
\textit{MRCE} lacks automated missing-data handling and lags behind in both
estimation accuracy and computational efficiency, making it less appropriate for
modern analytical demands. Although \textit{cglasso} demonstrates commendable
performance, it remains slightly inferior to \texttt{missoNet}. In certain
scenarios, \textit{cglasso} can marginally outperform \texttt{missoNet} in PE
under the same BIC tuning, yet it generally underperforms in KLL for
$\bTheta_{\varepsilon\varepsilon}^{*}$ and in variable-selection accuracy for
both $\bB^{*}$ and $\bTheta_{\varepsilon\varepsilon}^{*}$. 

The capacity of \textit{cglasso} to maintain stable performance despite
increasing missingness is noteworthy, due in part to its EM-based iterative
updating of missing data. However, its reliance on
information-criterion-driven model selection frequently leads to overly complex
network structures for $\bTheta_{\varepsilon\varepsilon}^*$, thereby
undermining both interpretability and predictive accuracy. By contrast,
\texttt{missoNet} achieves robust performance through multiple tuning
approaches: cv.min focuses on predictive accuracy, cv.1se
produces sparser solutions with high MCC, and BIC excels at recovering
$\bTheta_{\varepsilon\varepsilon}^*$ in demanding scenarios (e.g., Models~2
and~3).

Given \texttt{missoNet}'s computational efficiency under information-criteria
methods and the broader generalizability enabled by cross-validation, we suggest
employing both strategies in practical applications. Parameter tuning should
align with validation outcomes and specific research goals, and in certain
contexts, ensemble approaches may further enhance performance.

\section{A real data illustration}

\subsection{Introduction}

We previously reported that DNA-methylation at many cytosines is linked to blood levels of anti-citrullinated-protein antibody (ACPA), a key marker of rheumatoid arthritis (RA) \citep{shao2019rheumatoid, zeng2021thousands}. 
Using binomial regression models, we identified 19,472 ACPA-associated differentially methylated cytosines (DMCs) across the genome in serum samples from two \href{www.cartagene.qc.ca}{CARTaGENE} subsamples that were profiled by targeted bisulfite sequencing \citep{zeng2021thousands}. To test whether nearby genetic variants help shape these methylation signals, we re-analysed these data with \texttt{missoNet}, grouping CpG sites and single-nucleotide polymorphisms (SNPs) by chromosomal region. This multivariate strategy uncovered regional mQTLs that influence methylation while also modelling the epigenetic network within each segment.

\subsection{Data preprocessing}

The dataset analyzed here includes 202 adults selected from the CARTaGENE, with 115 
from a cohort processed using an updated bisulfite sequencing protocol that increased read depth and lowered the proportion of missing methylation calls \citep{zeng2021thousands, shao2019rheumatoid}, improvements that have been reported to raise statistical power in similar studies \citep{seiler2021characterizing}. 
Genome-wide genotypes were retrieved from the same biobank and filtered to retain variants that met Hardy–Weinberg equilibrium ($p > 5 \times 10^{-5}$), had a minor allele frequency (MAF) above 5\%, and showed high imputation quality (INFO score > 0.8), standards that are common in large mQTL surveys \citep{oliva2023dna}. 
Because methylation proportions follow a binomial distribution, we applied an arcsine transformation, which stabilises their variance and helps normalize the data to satisfy Gaussian model assumptions \citep{park2016differential, zheng2017estimating}.

DNA methylation at these DMCs is known to differ in people who carry ACPA, and it varies systematically among blood-cell types \citep{wu2011global}, with additional shifts seen across age and between the sexes \citep{ryan2020systematic}. 
Since some comparative methods lack flexibility for variable-specific adjustments (e.g., setting $\lambda_{\mathbf{B}} = 0$ for age), we accounted for these anticipated sources of variation by regressing the transformed methylation values at every CpG on five covariates: study cohort (original versus updated protocol), ACPA status, age, sex, and the first principal component of estimated cell-type proportions, a strategy recommended to remove batch and cellular-composition effects in whole-blood studies \citep{salas2022enhanced, kaushal2017comparison}. The residuals from this regression were then carried forward to the mQTL scan.

To simplify downstream steps, ACPA-linked DMCs were merged into regions whenever adjacent sites lay within 250 bp of one another, reflecting the local correlation of cytosine methylation. 
We next restricted attention to SNPs located within ±250 kb of each DMC region, which is a  definition of the \textit{cis} window used in many mQTL studies \citep{hannon2018leveraging, shao2019rheumatoid}. 
Each CpG–SNP pair was preliminarily tested with an univariate linear model with additive genetic effects; only variants achieving a false discovery rate below 0.10 entered the multivariate mQTL analysis using \texttt{missoNet} and other methods.

\subsection{Experimental design}

In this section we describe how we examined whether common variants that flank ACPA-associated DMCs can account for the coordinated methylation shifts across neighbouring CpGs.

\paragraph{Statistical framework}
For every multi-CpG region we created an $N\times p$ response matrix $\mathbf{Y}$ whose $p$ columns hold former regressed residuals of individual CpG sites. The corresponding predictor matrix $\mathbf{X}$ contained SNP genotypes drawn from a ±250 kb \textit{cis}-window around the region. \texttt{missoNet} and its competitors estimate the regression coefficient matrix $\mathbf{B}$ while jointly learning a sparse precision matrix $\mathbf{\Theta}$, so that correlations between CpGs can inform one another.

\paragraph{Data partition, regional aggregation, and missing calls}
Among 202 CARTaGENE participants, the 115 sequenced with the newer, higher-depth protocol formed the training set, whereas the remaining 87 sequenced earlier constituted an independent validation set. Of 19,472 previously reported ACPA-associated DMCs, 991 lay near at least one candidate SNP; merging any three or more consecutive DMCs (of these 991) within 250 bp yielded 62 CpG regions (mean = 8.2 CpGs). For illustration purposes and tractable computations, we focused on 16 moderately-sized regions that each held 5–50 CpGs and were linked to 10–100 SNPs in earlier single-marker scans (full coordinates are provided in \textbf{Supplementary Table S1}, Region~1--16, as detailed in Section~\ref{sup_tab_S1}). The average missing methylation call rates were 12\% in training and 28\% in validation; for CpGs that were complete, we removed an additional 5\% of observed values at random to mimic data that is missing completely at random.

\paragraph{Competing models and tuning procedures}
We compared \texttt{missoNet} against \textit{Lasso}, \textit{MRCE}, and \textit{cglasso}. 
For \textit{Lasso} and \textit{MRCE}, missing values were imputed by column means before fitting; \texttt{missoNet} and \textit{cglasso} modelled missingness directly. 
All methods were given identical grids for the penalties on $\mathbf{B}$ and $\mathbf{\Theta}$ (if applicable), and convergence tolerances were matched. Five-fold cross-validation selected tuning parameters for \textit{Lasso}, \textit{MRCE} and one flavour of \texttt{missoNet}, whereas \textit{cglasso} and a second \texttt{missoNet} flavour relied on the BIC (as \texttt{missoNet} supports tuning through both criteria).

\paragraph{Performance metrics and safeguards against false discoveries}
Predictive accuracy and goodness of fit on the held-out 87 samples were assessed by mean squared prediction error (MSE) and out-of-sample log-likelihood (Log.Lik). To gauge robustness against false positives, we spiked each region with 10 synthetic "null" SNPs whose genotype counts (0, 1, 2) were randomly drawn to match the allele-frequency spectrum of nearby real variants; such negative controls were used for GWAS as a practical check on false-discovery rate control \citep{sesia2021false}.

\subsection{Results of mQTL data analyses}

Table \ref{tab_mQTL} summarizes the comparative analysis of different methods for modeling DNA methylation variability influenced by SNPs in \textit{cis}-regions. Prediction accuracy on unseen data is assessed using mean squared error (MSE), while log-likelihood (Log.Lik) evaluates the fit of the assumed multivariate normal model to the joint distribution of CpG values. The proportions of nonzero elements for synthetic (P$(\widehat{\mathbf{B}}^{\text{null}})$) and real (P$(\widehat{\mathbf{B}}^{\text{true}})$) SNPs provide insights into each method's ability to distinguish relevant from redundant variables; ratios above one indicate a method retains more real than dummy variables.

\renewcommand{\arraystretch}{0.7}
\begin{xltabular}{\textwidth}{@{}>{\arraybackslash}p{3.0cm} >{\centering\arraybackslash}p{1.5cm} >{\centering\arraybackslash}p{1.5cm} >{\centering\arraybackslash}p{2.6cm} >{\centering\arraybackslash}p{1.5cm} >{\centering\arraybackslash}p{1.5cm} >{\centering\arraybackslash}p{1.5cm}@{}}
  \caption{Means and standard errors (in parentheses) of evaluation metrics across 16 tested regions for all comparative methods. Mean squared error (\textnormal{MSE}) and log-likelihood (\textnormal{Log.Lik}) were calculated on the independent validation cohort. The best-performing method for each of the first three metrics is highlighted in red, while results within $\pm 5$\% of the best are highlighted in teal. \textnormal{P}$(\widehat{\mathbf{B}}^{\textnormal{null}})$ and \textnormal{P}$(\widehat{\mathbf{B}}^{\textnormal{true}})$ represent the proportions of nonzero associations for synthetic and true predictors, respectively. \textnormal{P}$(\widehat{\mathbf{\Theta}}_{\varepsilon\varepsilon})$ indicates the proportion of nonzero off-diagonal elements in the estimated inverse covariance matrix, $\widehat{\mathbf{\Theta}}_{\varepsilon\varepsilon}$.\\
  \textnormal{\textbf{Abbreviations}:}\ \textnormal{cv.min}: cross-validation with the minimum error; \textnormal{cv.1se}: cross-validation using the one-standard-error rule; \textnormal{BIC}: Bayesian Information Criterion.}\\
  \toprule
    Method & MSE & Log.Lik & \text{P}$(\widehat{\mathbf{B}}^{\text{true}})$\big/ \text{P}$(\widehat{\mathbf{B}}^{\text{null}})$ & P$(\widehat{\mathbf{B}}^{\text{null}})$ & P$(\widehat{\mathbf{B}}^{\text{true}})$ & P$(\widehat{\mathbf{\Theta}}_{\varepsilon\varepsilon})$\\
  \midrule
  \endfirsthead

  \multicolumn{7}{c}%
  {{Table \thetable\ Continued from Previous Page}} \\
  \toprule
  Method & MSE & Log.Lik & \text{P}$(\widehat{\mathbf{B}}^{\text{true}})$\big/ \text{P}$(\widehat{\mathbf{B}}^{\text{null}})$ & P$(\widehat{\mathbf{B}}^{\text{null}})$ & P$(\widehat{\mathbf{B}}^{\text{true}})$ & P$(\widehat{\mathbf{\Theta}}_{\varepsilon\varepsilon})$\\
  \midrule
  \endhead

  \midrule \multicolumn{7}{r}{{Continued on next page}} \\
  \endfoot

  \endlastfoot

  \textit{Lasso} [cv.min] & 0.160 & -$1.769^{1}$ & 0.636 & 0.202 & 0.129 \\
  & (0.022) & (0.413) & & (0.032) & (0.020) \\
  \cline{2-7}
  \textit{Lasso} [cv.1se] & 0.156 & -1.689 & 3.864 & 0.013 & 0.049 \\
  & (0.019) & (0.317) & & (0.008) & (0.011) \\
  \midrule
  \texttt{missoNet} [cv.min] & \color{Teal}0.155 & 5.281 & 1.222 & 0.247 & 0.302 & 0.148 \\
  & (0.024) & (2.593) & & (0.055) & (0.056) & (0.072) \\
  \cline{2-7}
  \texttt{missoNet} [cv.1se] & \color{Red}\textbf{0.149} & 8.196 & \color{Red}\textbf{12.536} & 0.011 & 0.137 & 0.268 \\
  & (0.019) & (1.997) & & (0.007) & (0.033) & (0.096) \\
  \cline{2-7}
  \texttt{missoNet} [BIC] & 0.179 & 15.081 & 0.596 & 0.002 & 0.001 & 0.757 \\
  & (0.024) & (3.292) & & (0.001) & (0.001) & (0.033) \\
  \midrule
  \textit{MRCE} [cv.min] & \color{Teal}0.152 & 6.264 & 0.556 & 0.256 & 0.142 & 0.049 \\
  & (0.017) & (2.309) & & (0.058) & (0.024) & (0.019) \\
  \midrule
  \textit{cglasso} [BIC] & 0.176 & \color{Red}\textbf{17.092} & 1.023 & 0.010 & 0.010 & 0.784 \\
  & (0.023) & (3.218) & & (0.004) & (0.007) & (0.039) \\
  \bottomrule
  \multicolumn{7}{l}{\footnotesize ${}^{1}$For \textit{Lasso}, an identity matrix was employed for $\widehat{\mathbf{\Theta}}_{\varepsilon\varepsilon}$ in the computation of the log-likelihood.}
    
\label{tab_mQTL}
\end{xltabular}

\textit{Lasso}, a widely used tool in QTL mapping due to its scalability and ease of implementation \citep{usai2009lasso}, yielded a moderate MSE (0.160) when using the minimum-error criterion (cv.min). However, it selected 20\% of dummy SNP associations, resulting in a true-to-dummy ratio below one, indicative of overfitting. Employing the one-standard-error rule (cv.1se) substantially reduced the false-positive rate, increased the ratio to 3.9, and slightly improved MSE, suggesting enhanced generalizability and feature selection capability. This improvement came at the cost of retaining fewer true predictors, implying mild underfitting and the potential exclusion of important variables.

In contrast, \texttt{missoNet}, which incorporates network relationships among CpG sites, outperformed \textit{Lasso} and other methods across most metrics when appropriately tuned via cross-validation.  Specifically, \texttt{missoNet} [cv.1se] achieved the best MSE and the highest true-to-dummy ratio (12.54) while retaining a manageable 13.7\% of real SNPs and misidentifying only 1.1\% of null associations. The cv.min variant retained more edges in both the coefficient and precision matrices but at the expense of significantly worse performance in prediction accuracy, model fit, and feature selection. When \texttt{missoNet} was tuned using the BIC, it selected an overly dense network (75\% of precision-matrix edges) and exhibited the poorest MSE among the three \texttt{missoNet} variants, illustrating overfitting by likelihood-based criteria when the sample size is small.

\textit{MRCE} and \textit{cglasso}, in comparison, produced less favorable results. \textit{MRCE} delivered a competitive MSE (0.152) but showed the weakest discrimination, retaining 26\% of dummy SNP associations and the lowest true-to-dummy ratio. The \textit{cglasso} method, optimized using BIC, achieved the highest log-likelihood (17.1) yet the largest MSE (0.176) and a near-unity true-to-dummy ratio, indicating that likelihood alone does not guarantee predictive utility and feature selection capacity.

In summary, incorporating CpG–CpG correlation, as \texttt{missoNet} does, can substantially improve detection of \textit{cis}-mQTLs, aligning with large-scale catalogs that report extensive regional sharing of methylation signals \cite{gaunt2016systematic}.  When tuned with the one-standard-error cross-validation rule, \texttt{missoNet} provides the most accurate and parsimonious model of \textit{cis}-mQTL effects. Across 16 CpG regions, it achieved the lowest MSE, a favorable out-of-sample log-likelihood, and the strongest separation between genuine and synthetic predictors. This indicates that its network-aware shrinkage and unbiased surrogate estimators under missing values effectively recover true genetic drivers of regional methylation while suppressing noise. Although the cv.min version of \texttt{missoNet} showed greater model complexity and signs of overfitting, it still outperformed other cross-validation-based methods such as \textit{Lasso} and \textit{MRCE} in either prediction accuracy or support set recovery. In addition, this application study suggests that while likelihood-based criteria like BIC may lead to good data fit, they can inflate network complexity when sample size is modest, ultimately compromising predictive accuracy and identification ability by overfitting to intricate network structures in $\widehat{\mathbf{\Theta}}_{\varepsilon\varepsilon}$.

\textbf{Supplementary Table S1}, Region~1--16 (described in Section~\ref{sup_tab_S1}) details the performance of the best-performing method -- \texttt{missoNet} [cv.1se] -- across the 16 regions concerning the performance metrics in Table \ref{tab_mQTL}. They also list all tested CpGs and SNPs for each region analyzed, as well as the identified associations. For CpG-SNP associations, these supplemental tables report whether the association was previously documented in the life-course mQTL map of \citet{gaunt2016systematic} (catalog available at \url{http://www.mqtldb.org}). In nearly all regions, \texttt{missoNet} [cv.1se] identified known mQTLs or discovered new CpG-SNP associations where the same SNPs alter methylation levels of CpGs within 100 bp of a known mQTL association. Specifically, we tested a total of 9,238 CpG-SNP associations across 16 regions; \texttt{missoNet} [cv.1se] identified 957 as nonzero, with 841 (\textbf{87.88\%}) confirmed as either replicates of known mQTL associations (30/841) or new links where the SNPs alter methylation of CpGs within 100 bp of a documented site. The majority of these close mappings result from the targeted bisulfite panel not perfectly overlapping the 450K array used by \citet{gaunt2016systematic}. Two regions merit note: in Region 14 (chr20:32255222–32255294, GRCh37), no associations were detected, and none were reported by \citet{gaunt2016systematic}, suggesting this window lacks strong cis control in whole blood. In Region 1 (chr1:8958863–8959033), \texttt{missoNet} identified 25 new CpG–SNP pairs not previously cataloged, representing candidates for follow-up validation. Collectively, these findings support \texttt{missoNet} as a robust tool for \textit{cis}-mQTL discovery in modestly sized sequencing studies, especially when missing data and correlated methylation patterns are present.

\section{Discussion and conclusion}

In this work, we proposed a three-stage convex estimation technique to jointly estimate the regression coefficients and the precision matrix in a multivariate regression setting with missing response data. This method is specifically designed to address challenges arising from missingness, employing a structured approach to ensure both accurate and interpretable estimation. Theoretical bounds were established for each stage of the estimation process, quantifying the statistical error in terms of key factors such as the degree of missingness, error variance, dimensions of the response and covariate matrices, and sample size. These bounds provide a rigorous theoretical foundation for the proposed approach.

To validate our theoretical findings, we conducted a comprehensive simulation study. The results demonstrated strong alignment between the theoretical error bounds and the observed performance, reinforcing the robustness of the proposed method. Additionally, we developed a computationally efficient proximal gradient descent algorithm to implement the estimation procedure. This algorithm is publicly available as an \texttt{R} package, enhancing accessibility and reproducibility for practitioners.

Extensive simulations were conducted to benchmark our method against other established techniques. The comparisons revealed that our approach consistently outperforms competing methods across a wide range of scenarios, particularly in settings characterized by high missingness or complex dependency structures. To further illustrate the practical utility of our method, we applied it to a real world analysis for identifying cis-SNPs that alter regional DNA methylation levels, showcasing its applicability and relevance to real data. This analysis both confirmed previously-reported SNP-CpG associations, and identified many new ones, in spite of the missingness at each region.  

Notably, our algorithm serves as a powerful tool for addressing multivariate regression problems with missing responses. Furthermore, since missing data can be viewed as a special case of multiplicative noise, the method can be generalized to accommodate scenarios where responses are subject to other kinds of multiplicative noise. Assuming the rows \( w_{i\bullet} \) of the noise matrix \( \mathbf{W} \in \mathbb{R}^{n \times q} \) are drawn independently and identically from a multivariate distribution with first and second-order expectations \( \boldsymbol{\mu}_{W} = \mathbb{E}[\mathbf{W}] \in \mathbb{R}^{q} \) and \( \mathbb{E}[\mathbf{W}\mathbf{W}^{\top}] \in \mathbb{R}^{q \times q} \), respectively, the theoretical derivations remain valid for this more general case.

In conclusion, the proposed method’s theoretical guarantees, computational efficiency, and strong empirical performance make it a valuable addition to the statistical toolkit. Future research directions may include extending the framework to more complex missing data mechanisms and developing theoretical results for settings where missing data occurs in both responses and covariates.

\section*{Acknowledgments}
The data that support the DNA methylation analysis are available from CARTaGENE. Restrictions apply to the availability of these data, which were used under license for this study. Data are available from \href{https://cartagene.qc.ca/en/researchers.html}{https://cartagene.qc.ca/en/researchers.html} with the permission of CARTaGENE.

\newpage
\makeatletter
\renewcommand\maketitle{
  \null
  \thispagestyle{plain}
  \vskip 2em
  \begin{center}%
    \let \footnote \thanks
    {\LARGE \@title \par}%
    \vskip 1.5em
    {\large
      \lineskip .5em
      \begin{tabular}[t]{c}%
        \@author
      \end{tabular}\par}%
    \vskip 1em
    {\large \@date \par}%
    \vskip 1em
  \end{center}%
  \par
  \vskip 1.5em}
\makeatother

\title{Appendices for "Multivariate regression with missing response data for modelling regional DNA methylation QTLs"}
\author{Shomoita Alam, Yixiao Zeng, Sasha Bernatsky, Marie Hudson, Inés Colmegna, \\David A. Stephens, Celia M. T. Greenwood, Archer Y. Yang}

\maketitle

\setcounter{page}{1}
\normalsize
\phantomsection
\addcontentsline{toc}{section}{APPENDICES}
\setcounter{table}{0} 
\renewcommand{\thetable}{A\arabic{table}} 
\setcounter{figure}{0} 
\renewcommand{\thefigure}{A\arabic{figure}}

\renewcommand{\thesubsection}{Appendix \Alph{subsection}}
\setcounter{subsection}{0}


\subsection{Technical proofs \label{P2_Supp}}

\subsubsection{Derivation of $\widehat{\bS}_{\varepsilon\varepsilon}\coloneqq\widehat{\bS}_{yy}-\widehat{\bB}^{(1)\top}\bS_{xx}\widehat{\bB}^{(1)}$: \label{deriv_See}}

We can define the population variance of the error, $\varepsilon$, 
\begin{align*}
    \bSigma_{\varepsilon\varepsilon}^{*} &= \mathrm{Var}(Y-X\bB^{*})\\
    &=\mathrm{Var}(Y) - 2 \mathrm{Cov}(Y,X)\bB^{*} + \bB^{*\top} \mathrm{Var}(X) \bB^{*}\\
    &= \bSigma_{yy}^{*} - 2 \bSigma_{yx}^{*} \bB^{*}+\bB^{*\top} \bSigma_{xx}^{*} \bB^{*}
\end{align*}

We can also simplify the covariance between $Y$ and $X$ as, 
\begin{align*}
\bSigma_{yx}^{*} &= \mathrm{Cov}(X, X\bB^{*}+\varepsilon)\\
& = \mathrm{Cov}(X, X\bB^{*})+\mathrm{Cov}(X, \varepsilon)\\
& = \bB^{*\top} \bSigma_{xx}^{*},
\end{align*}
since $\mathrm{Cov}(X, \varepsilon)=0$. Therefore, $\bSigma_{\varepsilon\varepsilon}^{*}=\bSigma_{yy}^{*} -\bB^{*\top} \bSigma_{xx}^{*} \bB^{*}$. A plug-in estimate of $\bSigma_{\varepsilon\varepsilon}^{*}$ is given by $\widehat{\bS}_{\varepsilon\varepsilon}\coloneqq\widehat{\bS}_{yy}-\widehat{\bB}^{(1)\top}\bS_{xx}\widehat{\bB}^{(1)}$.

\subsubsection{Proof of Lemma \ref{lem:1} \label{Pr_Lem1}}
\begin{proof}
Recall that, for a multiplicative measurement error model, we assume
the observed matrix is $\mathbf{Z}=\mathbf{Y}\odot\mathbf{W}$ where
$\mathbf{W}=(\mathbf{w}_{1}\ddd\mathbf{w}_{n})^{\top}$ is a matrix
of multiplicative error. Given $\mathbf{S}_{xy}=\frac{1}{n}\bX^{\top}\mathbf{Y}$
as the matrix that represents the covariance between $\bX$ and uncontaminated
$\bY$, we have 
\begin{align}
\widehat{\bS}_{xy}-\mathbf{S}_{xy} & =\frac{1}{n}\mathbf{X}^{\top}\mathbf{Z}\oslash[\mathbb{E}[W],...,\mathbb{E}[W]]^{\top}-\frac{1}{n}\bX^{\top}\mathbf{Y}\nonumber \\
 & =\frac{1}{n}\mathbf{X}^{\top}\mathbf{(\mathbf{Y}\odot\mathbf{W})}\oslash[\mathbf{1}_{q}-\boldsymbol{\rho},\ldots,\mathbf{1}_{q}-\boldsymbol{\rho}]^{\top}-\frac{1}{n}\bX^{\top}\mathbf{Y}.\label{eq:sigma_mult_matrix_format}
\end{align}
Let $Y_{ij}$ be the $i$th row and $j$th column of $\bY$ for $j=1,\ldots,q$
and $k=1,\ldots,p$, also $W_{ij}$ and $X_{ik}$ are defined similarly.
Followed by \eqref{eq:sigma_mult_matrix_format}, we have, 
\begin{align*}
(\widehat{\bS}_{xy})_{kj}-\left(\mathbf{S}_{xy}\right)_{kj} & =\frac{1}{n(1-\rho_{j})}\sum_{i=1}^{n}Y_{ij}W_{ij}X_{ik}-\frac{1}{n}\sum_{i=1}^{n}Y_{ij}X_{ik}\\
 & =\frac{1}{n(1-\rho_{j})}\sum_{i=1}^{n}Y_{ij}W_{ij}X_{ik}-\frac{1}{n(1-\rho_{j})}\sum_{i=1}^{n}Y_{ij}X_{ik}\mathbb{E}W_{ij}\\
 & =\frac{1}{n(1-\rho_{j})}\sum_{i=1}^{n}Y_{ij}X_{ik}(W_{ij}-\mathbb{E}W_{ij}),
\end{align*}
with $\mathbb{E}W_{ij}=1-\rho_{j}$. Now we plug in the true model
$Y_{ij}=\sum_{k'\in S_{j}}X_{ik'}\beta_{k'j}^{*}+\varepsilon_{ij}$
and get 
\begin{align*}
\left|(\widehat{\bS}_{xy})_{kj}-\left(\mathbf{S}_{xy}\right)_{kj}\right| & =\left|\frac{1}{n(1-\rho_{j})}\sum_{i=1}^{n}(\sum_{k'=1}^{s_{j}}X_{ik'}\beta_{k'j}^{*}+\varepsilon_{ij})X_{ik}(W_{ij}-\mathbb{E}W_{ij})\right|\\
 & =\frac{1}{(1-\rho_{j})}\left[\left|\frac{1}{n}\sum_{i=1}^{n}\sum_{k'=1}^{s_{j}}\beta_{k'j}^{*}X_{ik'}X_{ik}(W_{ij}-\mathbb{E}W_{ij})\right|+\left|\frac{1}{n}\sum_{i=1}^{n}\varepsilon_{ij}X_{ik}(W_{ij}-\mathbb{E}W_{ij})\right|\right]\\
 & \leq\frac{1}{\mu_{\min}}\left[\left|\frac{1}{n}\sum_{i=1}^{n}\sum_{k'=1}^{s_{j}}\beta_{k'j}^{*}X_{ik'}X_{ik}(W_{ij}-\mathbb{E}W_{ij})\right|+\left|\frac{1}{n}\sum_{i=1}^{n}\varepsilon_{ij}X_{ik}(W_{ij}-\mathbb{E}W_{ij})\right|\right]\\
 & \leq\frac{1}{\mu_{\min}}\left[s_{\max}B_{\max}X_{\max}^{2}\left|\underbrace{\frac{1}{n}\sum_{i=1}^{n}(W_{ij}-\mathbb{E}W_{ij})}_{T_{1}}\right|+X_{\max}\left|\underbrace{\frac{1}{n}\sum_{i=1}^{n}\varepsilon_{ij}(W_{ij}-\mathbb{E}W_{ij})}_{T_{2}}\right|\right]
\end{align*}
where $X_{\max}=\max_{i,k}|X_{ik}|<\infty$, $B_{\max}=\max_{k,j}|\beta_{kj}^{*}|$,
$\mu_{\min}=\min_{j}\mu_{j}=\min_{j}1-\rho_{j}$. We find bounds for
$T_{1}$ and $T_{2}$ separately. For $T_{1}$, applying Theorem 2.6.3
general Hoeffding inequality on page 27 of \citet{vershynin2018high},
we get
\begin{align*}
 & \qquad\Pr(|T_{1}|\geq t)\\
 & =\Pr\left[\left|\sum_{i=1}^{n}\frac{1}{n}(W_{ij}-\mathbb{E}W_{ij})\right|\geq t\right]\\
 & \leq2\exp\left\{ -\frac{ct^{2}}{K_{W}^{2}\sum_{i=1}^{n}(1/n)^{2}}\right\} \\
 & \overset{\text{(i)}}{\leq}2\exp\left\{ -\frac{cm_{1}t^{2}}{\sigma_{W}^{2}\sum_{i=1}^{n}(1/n)^{2}}\right\} \\
 & \overset{\text{(ii)}}{\leq}2\exp\left\{ -\frac{cnt^{2}}{\sigma_{W}^{2}}\right\} 
\end{align*}
where $K_{W}=\max_{i,j}\|W_{ij}-\mathbb{E}W_{ij}\|_{\psi_{2}}$. Inequality
(i) is due to the implication of Lemma 5.5 of \citet{vershynin2010introduction}
that for sub-Gaussian random variable $W_{ij}-\mathbb{E}W_{ij}$ there
exist universal constants $m_{1}$ and $M_{1}$ such that $m_{1}\|W_{ij}-\mathbb{E}W_{ij}\|_{\psi_{2}}^{2}\leq\sigma_{W}^{2}\leq M_{1}\|W_{ij}-\mathbb{E}W_{ij}\|_{\psi_{2}}^{2}$
hold. It can be simplified to inequality (ii) since $m_{1}$ is a
constant and consequently gets absorbed into the universal constant
$c$.

Now we will find the bound for $T_{2}$. Follow by Lemma 2.7.7 of
\citet{vershynin2018high}, we notice that $\varepsilon_{ij}(W_{ij}-\mathbb{E}W_{ij})$
as the product of two independent, centered sub-Gaussian random variables
follows sub-exponential distribution since 
\[
\|\varepsilon_{ij}(W_{ij}-\mathbb{E}W_{ij})\|_{\psi_{1}}\leq\|\varepsilon_{ij}\|_{\psi_{2}}\|W_{ij}-\mathbb{E}W_{ij}\|_{\psi_{2}}.
\]
Also we can see that $\varepsilon_{ij}(W_{ij}-\mathbb{E}W_{ij})$
is centered since
\[
\E[\varepsilon_{ij}(W_{ij}-\mathbb{E}W_{ij})]=\E\varepsilon_{ij}\E(W_{ij}-\mathbb{E}W_{ij})=0
\]
with Now apply Theorem 2.8.2 of \citet{vershynin2018high}, we get
\begin{align*}
 & \qquad\Pr(|T_{2}|\geq t)\\
 & =\Pr\left[\left|\sum_{i=1}^{n}\frac{1}{n}\varepsilon_{ij}(W_{ij}-\mathbb{E}W_{ij})\right|>t\right]\\
 & \leq2\exp\left\{ -c\min\left(\frac{t^{2}}{\max_{i}\|\varepsilon_{ij}(W_{ij}-\mathbb{E}W_{ij})\|_{\psi_{1}}^{2}\sum_{i=1}^{n}(1/n)^{2}},\frac{t}{\max_{i}\|\varepsilon_{ij}(W_{ij}-\mathbb{E}W_{ij})\|_{\psi_{1}}\max_{i}(1/n)}\right)\right\} \\
 & \overset{(\text{i})}{\leq}2\exp\left\{ -c\min\left(\frac{nt^{2}}{\max_{i}\|\varepsilon_{ij}\|_{\psi_{2}}^{2}\|W_{ij}-\mathbb{E}W_{ij}\|_{\psi_{2}}^{2}},\frac{nt}{\max_{i}\|\varepsilon_{ij}\|_{\psi_{2}}\|W_{ij}-\mathbb{E}W_{ij}\|_{\psi_{2}}}\right)\right\} \\
 & \leq2\exp\left\{ -cn\min\left(\frac{t^{2}}{\sigma_{\varepsilon}^{2}\sigma_{W}^{2}},\frac{t}{\sigma_{\varepsilon}\sigma_{W}}\right)\right\} 
\end{align*}
where inequality (i) is again due to Lemma 5.5 of \citet{vershynin2010introduction}
that there exist universal constants $m_{1}$, $M_{1}$, $m_{2}$
and $M_{2}$ such that $m_{1}\|W_{ij}-\mathbb{E}W_{ij}\|_{\psi_{2}}^{2}\leq\sigma_{W}^{2}\leq M_{1}\|W_{ij}-\mathbb{E}W_{ij}\|_{\psi_{2}}^{2}$
and $m_{2}\|\varepsilon_{ij}\|_{\psi_{2}}^{2}\leq\sigma_{\varepsilon}^{2}\leq M_{2}\|\varepsilon_{ij}\|_{\psi_{2}}^{2}$
hold. 

We can combine the pieces together as follows: 
\begin{align*}
\mathrm{Pr}\left[\Big|(\widehat{\bS}_{xy})_{kj}-(\mathbf{S}_{xy})_{kj}\Big|\geq t\right] & \leq\mathrm{Pr}\left[\frac{s_{\max}B_{\max}X_{\max}^{2}}{\mu_{\min}}|T_{1}|+\frac{X_{\max}}{\mu_{\min}}|T_{2}|\geq t\right]\\
 & =\mathrm{Pr}\left[s_{\max}B_{\max}X_{\max}^{2}|T_{1}|+X_{\max}|T_{2}|\geq t\mu_{\min}\right]\\
 & \overset{\text{(i)}}{\leq}\mathrm{Pr}\left[2\max\left\{ s_{\max}B_{\max}X_{\max}^{2}|T_{1}|,X_{\max}|T_{2}|\right\} \geq t\mu_{\min}\right]\\
 & =\mathrm{Pr}\left[\max\left\{ s_{\max}B_{\max}X_{\max}^{2}|T_{1}|,X_{\max}|T_{2}|\right\} \geq t\mu_{\min}/2\right]\\
 & =\mathrm{Pr}\left[\left(s_{\max}B_{\max}X_{\max}^{2}|T_{1}|\geq t\mu_{\min}/2\right)\cup\left(X_{\max}|T_{2}|\geq t\mu_{\min}/2\right)\right]\\
 & =\mathrm{Pr}\left[\left(|T_{1}|\geq t\mu_{\min}/(2s_{\max}B_{\max}X_{\max}^{2})\right)\cup\left(|T_{2}|\geq t\mu_{\min}/(2X_{\max})\right)\right]\\
 & \leq\mathrm{Pr}\left[|T_{1}|\geq t\mu_{\min}/(2s_{\max}B_{\max}X_{\max}^{2})\right]+\mathrm{Pr}\left[|T_{2}|\geq t\mu_{\min}/(2X_{\max})\right]\\
 & =C\exp\left\{ -cn\frac{\mu_{\min}^{2}t^{2}}{s_{\max}^{2}\sigma_{W}^{2}X_{\max}^{4}B_{\max}^{2}}\right\} \\
 & \qquad\qquad+C\exp\left\{ -cn\min\left(\frac{\mu_{\min}^{2}t^{2}}{\sigma_{\varepsilon}^{2}\sigma_{W}^{2}X_{\max}^{2}},\frac{\mu_{\min}t}{\sigma_{\varepsilon}\sigma_{W}X_{\max}}\right)\right\} \\
 & \overset{\text{(ii)}}{=}C\exp\left\{ -cn\frac{\mu_{\min}^{2}t^{2}}{s_{\max}^{2}\sigma_{W}^{2}X_{\max}^{4}B_{\max}^{2}}\right\} +C\exp\left\{ -cn\frac{\mu_{\min}^{2}t^{2}}{\sigma_{\varepsilon}^{2}\sigma_{W}^{2}X_{\max}^{2}}\right\} \\
 & \leq2C\max\left\{ \exp\left(-cn\frac{\mu_{\min}^{2}t^{2}}{s_{\max}^{2}\sigma_{W}^{2}X_{\max}^{4}B_{\max}^{2}}\right),\exp\left(-cn\frac{\mu_{\min}^{2}t^{2}}{\sigma_{\varepsilon}^{2}\sigma_{W}^{2}X_{\max}^{2}}\right)\right\} \\
 & \leq C\exp\left(-cn\frac{\mu_{\min}^{2}t^{2}}{\sigma_{W}^{2}X_{\max}^{2}\max(s_{\max}^{2}X_{\max}^{2}B_{\max}^{2},\sigma_{\varepsilon}^{2})}\right)
\end{align*}
Inequality (i) is due to the relationship $A+B+|A-B|=2\max(A,B)$,
which implies $A+B\leq2\max(A,B)$. Equality (ii) follows by assuming
that $t\leq t_{0}^{(1)}$ with 
\begin{equation}
t_{0}^{(1)}=\sigma_{\varepsilon}\sigma_{W}X_{\max}/\mu_{\min}.\label{eq:t_restriction_1}
\end{equation}
Applying the union bound to find an upper bound for the elementwise
max norm, we get, 
\begin{align*}
\mathrm{Pr}[\VERT\widehat{\bS}_{xy}-\mathbf{S}_{xy}\VERT_{\max}\geq t] & \leq\mathrm{Pr}\left[\max_{j,k}\Big|(\widehat{\bS}_{xy})_{kj}-(\mathbf{S}_{xy})_{kj}\Big|\geq t\right]\\
 & \leq\sum_{j,k}\mathrm{Pr}\left[\Big|(\widehat{\bS}_{xy})_{kj}-(\mathbf{S}_{xy})_{kj}\Big|\geq t\right]\\
 & \leq pqC\exp\left(-cn\frac{\mu_{\min}^{2}t^{2}}{\sigma_{W}^{2}X_{\max}^{2}\max(s_{\max}^{2}X_{\max}^{2}B_{\max}^{2},\sigma_{\varepsilon}^{2})}\right).
\end{align*}
\end{proof}

\subsubsection{Proof of Lemma \ref{lem:2} \label{Pr_Lem2}}
\begin{proof}
Given the definition of $\mathbf{S}_{xy}$ 
\begin{align}
\bS_{xy}-\bS_{xx}\bB^{*} & =\frac{1}{n}\bX^{\top}\mathbf{Y}-\frac{1}{n}\bX^{\top}\mathbf{X}\mathbf{B}^{*}=\frac{1}{n}\bX^{\top}(\bX\bB^{*}+\boldsymbol{\varepsilon})-\frac{1}{n}\bX^{\top}\mathbf{X}\mathbf{B}^{*}=\frac{1}{n}\bX^{\top}\boldsymbol{\varepsilon}\label{eq:Matrix_Sxy2Sigma}
\end{align}
    Hence if we consider the $k$th row and $j$th column of this difference
for $j=1,\ldots,q$ and $k=1,\ldots,p$, it is
\begin{align*}
\left(\bS_{xy}\right)_{kj}-\left(\bS_{xx}\bB^{*}\right)_{kj} & =\frac{1}{n}\sum_{i=1}^{n}\varepsilon_{ij}X_{ik}
\end{align*}
We can bound it by applying Theorem 2.6.3 general Hoeffding inequality
on page 27 of \citet{vershynin2018high}
\begin{align*}
\mathrm{Pr}\left[\Big|\sum_{i=1}^{n}(X_{ik}/n)\varepsilon_{ij}\Big|\geq t\right] & <2\exp\left\{ -\frac{ct^{2}}{(\max_{i}\|\varepsilon_{ij}\|_{\psi_{2}})^{2}\sum_{i=1}^{n}(X_{ik}/n)^{2}}\right\} \\
 & <2\exp\left\{ -\frac{cnt^{2}}{\sigma_{\varepsilon}^{2}X_{\max}^{2}}\right\} .
\end{align*}
Applying the union bound, we get
\[
\mathrm{Pr}(\VERT\bS_{xy}-\bS_{xx}\bB^{*}\VERT_{\max}\geq t)\leq pqC\exp\left[-\frac{cnt^{2}}{\sigma_{\varepsilon}^{2}X_{\max}^{2}}\right]
\]
and 
\[
\mathrm{Pr}(\|(\bS_{xy})_{\bullet l}-\bS_{xx}\bbeta_{l}^{*}\|_{\infty}\geq t)\leq pC\exp\left[-\frac{cnt^{2}}{\sigma_{\varepsilon}^{2}X_{\max}^{2}}\right].
\]
\end{proof}

\subsubsection{Proof of Proposition \ref{prop:1} \label{Pr_Prop1}}
\begin{proof}
From the definition of \eqref{eq:Beta_obj}, we have 
\begin{align*}
\mathcal{L}(\widehat{\bbeta}_{l})-\mathcal{L}(\bbeta_{l}^{*}) & =\frac{\widehat{\bbeta}_{l}^{\top}\bS_{xx}\widehat{\bbeta}_{l}}{2}-\frac{\bbeta_{l}^{*}{}^{\top}\bS_{xx}\bbeta_{l}^{*}}{2}-\left\{ (\widehat{\bS}_{xy})_{\bullet l}\right\} ^{\top}(\widehat{\bbeta}_{l}-\bbeta_{l}^{*})-(\widehat{\bbeta}_{l}-\bbeta_{l}^{*})^{\top}\bS_{xx}\bbeta_{l}^{*}+(\widehat{\bbeta}_{l}-\bbeta_{l}^{*})^{\top}\bS_{xx}\bbeta_{l}^{*}\\
 & =\frac{\widehat{\bbeta}_{l}^{\top}\bS_{xx}\widehat{\bbeta}_{l}}{2}+\frac{\bbeta_{l}^{*}{}^{\top}\bS_{xx}\bbeta_{l}^{*}}{2}-\widehat{\bbeta}_{l}^{\top}\bS_{xx}\bbeta_{l}^{*}-\left\{ (\widehat{\bS}_{xy})_{\bullet l}\right\} ^{\top}(\widehat{\bbeta}_{l}-\bbeta_{l}^{*})+(\widehat{\bbeta}_{l}-\bbeta_{l}^{*})^{\top}\bS_{xx}\bbeta_{l}^{*}\\
 & =\frac{1}{2}(\widehat{\bbeta}_{l}^{\top}\bS_{xx}\widehat{\bbeta}_{l}+\bbeta_{l}^{*}{}^{\top}\bS_{xx}\bbeta_{l}^{*}-2\widehat{\bbeta}_{l}^{\top}\bS_{xx}\bbeta_{l}^{*})-((\widehat{\bS}_{xy})_{\bullet l}-\bS_{xx}\bbeta_{l}^{*})^{\top}(\widehat{\bbeta}_{l}-\bbeta_{l}^{*})\\
 & \overset{\text{(i)}}{=}\frac{1}{2}[(\widehat{\bbeta}_{l}-\bbeta_{l}^{*})^{\top}\bS_{xx}(\widehat{\bbeta}_{l}-\bbeta_{l}^{*})]-((\widehat{\bS}_{xy})_{\bullet l}-\bS_{xx}\bbeta_{l}^{*})^{\top}(\widehat{\bbeta}_{l}-\bbeta_{l}^{*})\\
 & =\frac{1}{2}[\bdelta^{\top}\bS_{xx}\bdelta]-((\widehat{\bS}_{xy})_{\bullet l}-\bS_{xx}\bbeta_{l}^{*})^{\top}\bdelta
\end{align*}
where $\bdelta=\widehat{\bbeta}_{l}-\bbeta_{l}^{*}$ and equality
(i) follows by completing the square. Since $\widehat{\bbeta}_{l}$
is the solution of \eqref{eq:Beta_obj}
\[
\mathcal{L}(\widehat{\bbeta}_{l})+\lambda_{l}\|\widehat{\bbeta}_{l}\|_{1}\leq\mathcal{L}(\bbeta_{l}^{*})+\lambda_{l}\|\bbeta_{l}^{*}\|_{1}.
\]
Plugging in $\mathcal{L}(\widehat{\bbeta}_{l})-\mathcal{L}(\bbeta_{l}^{*})$
yields 
\begin{align}
\frac{1}{2}[\bdelta^{\top}\bS_{xx}\bdelta]+\lambda_{l}\|\widehat{\bbeta}_{l}\|_{1} & \leq((\widehat{\bS}_{xy})_{\bullet l}-\bS_{xx}\bbeta_{l}^{*})^{\top}\bdelta+\lambda_{l}\|\bbeta_{l}^{*}\|_{1}\nonumber \\
 & \leq\|\bdelta\|_{1}\|(\widehat{\bS}_{xy})_{\bullet l}-\bS_{xx}\bbeta_{l}^{*}\|_{\infty}+\lambda_{l}\|\bbeta_{l}^{*}\|_{1}\label{eq:obj_diff}
\end{align}
The second inequality follows from Hölder's inequality. In order to
obtain an upper bound for the left-hand side, we first bound the quantity
$\|(\widehat{\bS}_{xy})_{\bullet l}-\bS_{xx}\bbeta_{l}^{*}\|_{\infty}$.
Using the triangular inequality, we get
\[
\|(\widehat{\bS}_{xy})_{\bullet l}-\bS_{xx}\bbeta_{l}^{*}\|_{\infty}\leq\|(\widehat{\bS}_{xy})_{\bullet l}-(\bS_{xy})_{\bullet l}\|_{\infty}+\|(\bS_{xy})_{\bullet l}-\bS_{xx}\bbeta_{l}^{*}\|_{\infty}.
\]
The first term can be bounded by applying Lemma \ref{lem:1} by setting
$t=\lambda_{l}/4$. We see that for $\lambda_{l}\leq4t_{0}^{(1)}$,
we have 
\[
\mathrm{Pr}[\|(\widehat{\bS}_{xy})_{\bullet l}-(\mathbf{S}_{xy})_{\bullet l}\|_{\infty}\geq\lambda_{l}/4]\leq pC\exp\left(-\frac{cn\mu_{\min}^{2}\lambda_{l}^{2}}{\sigma_{W}^{2}X_{\max}^{2}\max(s_{\max}^{2}X_{\max}^{2}B_{\max}^{2},\sigma_{\varepsilon}^{2})}\right).
\]
The second term can be bounded by applying Lemma \ref{lem:2} by setting
$t=\lambda_{l}/4$
\[
\mathrm{Pr}(\|(\bS_{xy})_{\bullet l}-\bS_{xx}\bbeta_{l}^{*}\|_{\infty}\geq\lambda_{l}/4)\leq pC\exp\left[-\frac{cn\lambda_{l}^{2}}{\sigma_{\varepsilon}^{2}X_{\max}^{2}}\right].
\]
Define the event $\mathcal{E}_{1}=\{\|(\widehat{\bS}_{xy})_{\bullet l}-(\mathbf{S}_{xy})_{\bullet l}\|_{\infty}\geq\lambda_{l}/4\}$
and $\mathcal{E}_{2}=\{\|(\bS_{xy})_{\bullet l}-\bS_{xx}\bbeta_{l}^{*}\|_{\infty}\geq\lambda_{l}/4\}$,
and the ``good'' event $\mathcal{G}(\lambda_{l})=\{\lambda_{l}/2\geq\|(\widehat{\bS}_{xy})_{\bullet l}-\bS_{xx}\bbeta_{l}^{*}\|_{\infty}\}$.
Then by Boole's inequality
\begin{align}
\mathrm{Pr}(\mathcal{G}(\lambda_{l})) & =\mathrm{Pr}(\|(\widehat{\bS}_{xy})_{\bullet l}-\bS_{xx}\bbeta_{l}^{*}\|_{\infty}<\lambda_{l}/2)\nonumber \\
 & \geq\mathrm{Pr}(\mathcal{E}_{1}^{c}\cap\mathcal{E}_{2}^{c})=\mathrm{Pr}\left[(\mathcal{E}_{1}\cup\mathcal{E}_{2})^{c}\right]=1-\mathrm{Pr}\left[(\mathcal{E}_{1}\cup\mathcal{E}_{2})\right]\nonumber \\
 & \geq1-\mathrm{Pr}(\mathcal{E}_{1})-\mathrm{Pr}(\mathcal{E}_{2})\nonumber \\
 & \geq1-pC\exp\left(-\frac{cn\mu_{\min}^{2}\lambda_{l}^{2}}{\sigma_{W}^{2}X_{\max}^{2}\max(s_{\max}^{2}X_{\max}^{2}B_{\max}^{2},\sigma_{\varepsilon}^{2})}\right)-pC\exp\left[-\frac{cn\lambda_{l}^{2}}{\sigma_{\varepsilon}^{2}X_{\max}^{2}}\right]\nonumber \\
 & \geq1-2pC\max\left\{ \exp\left(-\frac{cn\lambda_{l}^{2}}{\sigma_{W}^{2}X_{\max}^{2}\max(s_{\max}^{2}X_{\max}^{2}B_{\max}^{2},\sigma_{\varepsilon}^{2})/\mu_{\min}^{2}}\right),\exp\left[-\frac{cn\lambda_{l}^{2}}{\sigma_{\varepsilon}^{2}X_{\max}^{2}}\right]\right\} \nonumber \\
 & =1-2pC\exp\left(-\frac{cn\lambda_{l}^{2}}{\max\left[\sigma_{W}^{2}X_{\max}^{2}\max(s_{\max}^{2}X_{\max}^{2}B_{\max}^{2},\sigma_{\varepsilon}^{2})/\mu_{\min}^{2},\sigma_{\varepsilon}^{2}X_{\max}^{2}\right]}\right)\nonumber \\
 & =1-pC\exp\left(-\frac{cn\lambda_{l}^{2}}{X_{\max}^{2}\max\left[\sigma_{W}^{2}s_{\max}^{2}X_{\max}^{2}B_{\max}^{2}/\mu_{\min}^{2},\sigma_{W}^{2}\sigma_{\varepsilon}^{2}/\mu_{\min}^{2},\sigma_{\varepsilon}^{2}\right]}\right)\label{eq:prob_ge}
\end{align}

Returning to \eqref{eq:obj_diff}, when the ``good'' event $\mathcal{G}(\lambda_{l})$
holds, we have 
\begin{align}
\frac{1}{2}[\bdelta^{\top}\bS_{xx}\bdelta] & \leq\|\bdelta\|_{1}\|(\widehat{\bS}_{xy})_{\bullet l}-\bS_{xx}\bbeta_{l}^{*}\|_{\infty}+\lambda_{l}(\|\bbeta_{l}^{*}\|_{1}-\|\widehat{\bbeta}_{l}\|_{1})\nonumber \\
 & \leq\frac{\lambda_{l}}{2}\|\bdelta\|_{1}+\lambda_{l}(\|\bbeta_{l}^{*}\|_{1}-\|\bbeta_{l}^{*}+\bdelta\|_{1}).\label{eq:obj_diff_bound}
\end{align}
Now since $(\bbeta_{l}^{*})_{S_{l}^{c}}=\mathbf{0}$, we have $\|\bbeta_{l}^{*}\|_{1}=\|(\bbeta_{l}^{*})_{S_{l}}\|_{1}$,
and using the reverse triangle inequality
\[
\|\bbeta_{l}^{*}+\bdelta\|_{1}=\|(\bbeta_{l}^{*})_{S_{l}}+\bdelta_{S_{l}}\|_{1}+\|\bdelta_{S_{l}^{c}}\|_{1}\geq\|(\bbeta_{l}^{*})_{S_{l}}\|_{1}-\|\bdelta_{S_{l}}\|_{1}+\|\bdelta_{S_{l}^{c}}\|_{1}.
\]
Substituting these relations into inequality \eqref{eq:obj_diff_bound}
yields
\begin{align}
\frac{1}{2}[\bdelta^{\top}\bS_{xx}\bdelta] & \leq\frac{\lambda_{l}}{2}\|\bdelta\|_{1}+\lambda_{l}(\|\bbeta_{l}^{*}\|_{1}-\|\bbeta_{l}^{*}+\bdelta\|_{1}).\nonumber \\
 & \leq\frac{\lambda_{l}}{2}\|\bdelta\|_{1}+\lambda_{l}(\|\bbeta_{l}^{*}\|_{1}-\|(\bbeta_{l}^{*})_{S_{l}}\|_{1}+\|\bdelta_{S_{l}}\|_{1}-\|\bdelta_{S_{l}^{c}}\|_{1})\nonumber \\
 & =\frac{\lambda_{l}}{2}(\|\bdelta_{S_{l}}\|_{1}+\|\bdelta_{S_{l}^{c}}\|_{1})+\lambda_{l}(\|\bdelta_{S_{l}}\|_{1}-\|\bdelta_{S_{l}^{c}}\|_{1})\nonumber \\
 & \overset{\text{(i)}}{=}\frac{3\lambda_{l}}{2}\|\bdelta_{S_{l}}\|_{1}-\frac{\lambda_{l}}{2}\|\bdelta_{S_{l}^{c}}\|_{1}\label{eq:cone_eq}\\
 & \leq\frac{3\lambda_{l}}{2}\|\bdelta_{S_{l}}\|_{1}\leq\frac{3}{2}\sqrt{s_{l}}\lambda_{l}\|\bdelta_{S_{l}}\|_{2}\nonumber 
\end{align}
This allows us to apply the restricted eigenvalue condition to $\bdelta$,
which ensures that $\kappa_{l}\|\bdelta\|_{2}^{2}\leq\bdelta^{\top}\bS_{xx}\bdelta.$
Combining this lower bound with our earlier inequality yields
\begin{align*}
\frac{\kappa_{l}}{2}\|\bdelta\|_{2}^{2} & \leq\frac{1}{2}[\bdelta^{\top}\bS_{xx}\bdelta]\leq\frac{3}{2}\sqrt{s_{l}}\lambda_{l}\|\bdelta_{S_{l}}\|_{2}\leq\frac{3}{2}\sqrt{s_{l}}\lambda_{l}\|\bdelta\|_{2}
\end{align*}
and rearranging yields the bound
\[
\|\widehat{\bbeta}_{l}-\bbeta_{l}^{*}\|_{2}\leq3\sqrt{s_{l}}\lambda_{l}/\kappa_{l}.
\]
Returning to equality (i) in \eqref{eq:cone_eq}, we can now show
that the condition for $\mathbb{C}(S_{l})$ holds
\[
0\leq\frac{1}{2}[\bdelta^{\top}\bS_{xx}\bdelta]\leq\frac{3\lambda_{l}}{2}\|\bdelta_{S_{l}}\|_{1}-\frac{\lambda_{l}}{2}\|\bdelta_{S_{l}^{c}}\|_{1}.
\]
Hence $\|\bdelta_{S_{l}^{c}}\|_{1}\leq3\|\bdelta_{S_{l}}\|_{1}$.
Now, we can also derive the $\ell_{1}$-norm bound for the estimation
error. 
\begin{align*}
\|\widehat{\bbeta}_{l}-\bbeta_{l}^{*}\|_{1} & =\|\bdelta\|_{1}\\
 & =\|\bdelta_{S_{l}}\|_{1}+\|\bdelta_{S_{l}^{c}}\|_{1}\\
 & \leq\|\bdelta_{S_{l}}\|_{1}+3\|\bdelta_{S_{l}}\|_{1}\qquad\forall\bdelta\in\mathbb{C}(S_{l})\\
 & \leq4\|\bdelta_{S_{l}}\|_{1}\\
 & \leq4\sqrt{s_{l}}\|\bdelta_{S_{l}}\|_{2}\\
 & \leq12s_{l}\lambda_{l}/\kappa_{l}.
\end{align*}
When we set 
\begin{equation}
\lambda_{l}/2\geq\lambda_{0}/2\coloneqq X_{\max}\max\left[\sigma_{W}s_{\max}X_{\max}B_{\max}/\mu_{\min},\sigma_{W}\sigma_{\varepsilon}/\mu_{\min},\sigma_{\varepsilon}\right]\sqrt{\frac{\log p}{n}}\label{eq:good_event_bound}
\end{equation}
in inequality \eqref{eq:prob_ge}, then
\begin{align}
\Pr(\mathcal{G}_{l}) & =\Pr(\|(\widehat{\bS}_{xy})_{\bullet l}-\bS_{xx}\bbeta_{l}^{*}\|_{\infty}\leq\lambda_{l}/2)\nonumber \\
 & \geq\Pr(\|(\widehat{\bS}_{xy})_{\bullet l}-\bS_{xx}\bbeta_{l}^{*}\|_{\infty}\leq\lambda_{0}/2)\nonumber \\
 & \geq1-pC\exp\left(-\frac{cn(X_{\max}\max\left[\sigma_{W}s_{\max}X_{\max}B_{\max}/\mu_{\min},\sigma_{W}\sigma_{\varepsilon}/\mu_{\min},\sigma_{\varepsilon}\right])^{2}(\log p/n)}{X_{\max}^{2}\max\left[\sigma_{W}^{2}s_{\max}^{2}B_{\max}^{2}/\mu_{\min}^{2},\sigma_{W}^{2}\sigma_{\varepsilon}^{2}/\mu_{\min}^{2},\sigma_{\varepsilon}^{2}\right]}\right)\nonumber \\
 & =1-pC\exp\left(-c\log p\right)\nonumber \\
 & =1-C\exp\left(\log p-c\log p\right)\nonumber \\
 & =1-C\exp((1-c)\log p)\nonumber \\
 & =1-C\exp(-c\log p)\label{eq:good_event_sxy}
\end{align}
With slight abuse of the notation, we define a new constant $-c$
that is equal to $1-c$. Notice that the lower bound $\lambda_{0}$
of $\lambda_{l}$ in \eqref{eq:good_event_bound} must satisfy the
condition assumed in \eqref{eq:t_restriction_1}
\[
\lambda_{0}/2\coloneqq X_{\max}\max\left[\sigma_{W}s_{\max}X_{\max}B_{\max}/\mu_{\min},\sigma_{W}\sigma_{\varepsilon}/\mu_{\min},\sigma_{\varepsilon}\right]\sqrt{\frac{\log p}{n}}\leq t_{0}^{(1)}\coloneqq X_{\max}\sigma_{W}\sigma_{\varepsilon}/\mu_{\min}
\]
which implies that sample size $n$ must be sufficiently large such
that 
\begin{align*}
\sqrt{\frac{\log p}{n}} & \leq\frac{X_{\max}\sigma_{W}\sigma_{\varepsilon}/\mu_{\min}}{X_{\max}\max\left[\sigma_{W}s_{\max}X_{\max}B_{\max}/\mu_{\min},\sigma_{W}\sigma_{\varepsilon}/\mu_{\min},\sigma_{\varepsilon}\right]}\\
 & =\frac{1}{\max(s_{\max}B_{\max}X_{\max}/\sigma_{\varepsilon},1,\mu_{\min}/\sigma_{W})}\\
 & =\min(\sigma_{\varepsilon}/(s_{\max}B_{\max}X_{\max}),\sigma_{W}/\mu_{\min},1)
\end{align*}
Therefore, Lemma \ref{lem:1} can be applied as in \eqref{eq:prob_ge}
of this proof. 
\end{proof}

\subsubsection{Proof of Lemma \ref{lem3} \label{Pr_Lem3}}
\begin{proof}
Recall that, for a multiplicative measurement error model, we assume
the observed matrix is $\mathbf{Z}=\mathbf{Y}\odot\mathbf{\mathbf{W}}$
where $\mathbf{W}=(\mathbf{w}_{1}\ddd\mathbf{w}_{n})^{\top}$ is a
matrix of multiplicative error. Let $\boldsymbol{\Sigma}_{W}$ be
the known population covariance matrix of the measurement errors $\mathbf{W}$
for the multiplicative model. Given $\mathbf{S}_{yy}$ as the sample
covariance matrix for the data without any corruption, we have 
\begin{align}
\widehat{\bS}_{yy}-\mathbf{S}_{yy} & =\frac{1}{n}\mathbf{Z}^{\top}\mathbf{Z}\oslash(\bSigma_{W}+\mu_{W}\mu_{W}^{\top})-\frac{1}{n}\mathbf{Y}^{\top}\mathbf{Y}\nonumber \\
 & =\frac{1}{n}\mathbf{(\mathbf{Y}\odot\mathbf{W})}^{\top}(\mathbf{Y}\odot\mathbf{W})\oslash(\bSigma_{W}+\mu_{W}\mu_{W}^{\top})-\frac{1}{n}\mathbf{Y}^{\top}\mathbf{Y}\nonumber \\
 & =\frac{1}{n}\mathbf{(\mathbf{Y}\odot\mathbf{W})}^{\top}(\mathbf{Y}\odot\mathbf{W})\oslash\mathbb{E}[W{W}^{\top}]-\frac{1}{n}\mathbf{Y}^{\top}\mathbf{Y},\label{eq:Sigma_Y_mult}
\end{align}
where the last line is due to the fact that $\mathbb{E}[W{W}^{\top}]=\mathrm{Cov}(W)+\mathbb{E}[W]\mathbb{E}[W^{\top}]=\bSigma_{W}+\mu_{W}\mu_{W}^{\top}$.

Let $Y_{ij}$ and $Y_{ik}$ be the $i$th row and $j$th and $k$th
column of $\bY$ for $i=1\ddd n$ and $j,k=1,\ldots,q$, and $W_{ij}$
and $W_{ik}$ can be defined similarly. Let $X_{ik'}$ and $X_{ik''}$
be the $i$th row and $k'$th and $k''$th column of $\bX$ for $k',k''=1,\ldots,p$.
Then we can express the $(k,j)$ element of $(\widehat{\bS}_{yy}-\mathbf{S}_{yy})$
as
\begin{align*}
(\widehat{\bS}_{yy})_{kj}-\left(\mathbf{S}_{yy}\right)_{kj} & =\frac{1}{n}\sum_{i=1}^{n}\frac{Y_{ij}W_{ij}Y_{ik}W_{ik}}{\mathbb{E}(W_{j}W_{k})}-\frac{1}{n}\sum_{i=1}^{n}Y_{ij}Y_{ik}
\end{align*}
where $\mathbb{E}(W_{j}W_{k})=(\mathbb{E}[W{W}^{\top}])_{jk}$. Now
by plugging in the true model 
\begin{align*}
Y_{ij} & =\sum_{k'\in S_{j}}X_{ik'}\beta_{k'j}^{*}+\varepsilon_{ij},\\
Y_{ik} & =\sum_{k''\in S_{k}}X_{ik''}\beta_{k''k}^{*}+\varepsilon_{ik}
\end{align*}
we get
\begin{align*}
 & (\widehat{\bS}_{yy})_{kj}-(\mathbf{S}_{yy})_{kj}\\
= & \frac{1}{\mathbb{E}(W_{j}W_{k})}\Bigg[\frac{1}{n}\sum_{i=1}^{n}(\sum_{k'\in S_{j}}X_{ik'}\beta_{k'j}^{*}+\varepsilon_{ij})(\sum_{k''\in S_{k}}X_{ik''}\beta_{k''k}^{*}+\varepsilon_{ik})W_{ij}W_{ik}\\
 & -\frac{1}{n}\sum_{i=1}^{n}(\sum_{k'\in S_{j}}X_{ik'}\beta_{k'j}^{*}+\varepsilon_{ij})(\sum_{k''\in S_{k}}X_{ik''}\beta_{k''k}^{*}+\varepsilon_{ik})\mathbb{E}(W_{j}W_{k})\Bigg]\\
= & \frac{1}{\mathbb{E}(W_{j}W_{k})}\Bigg[\frac{1}{n}\sum_{i=1}^{n}(\sum_{k'\in S_{j}}X_{ik'}\beta_{k'j}^{*}+\varepsilon_{ij})(\sum_{k''\in S_{k}}X_{ik''}\beta_{k''k}^{*}+\varepsilon_{ik})(W_{ij}W_{ik}-\mathbb{E}(W_{j}W_{k}))\Bigg].
\end{align*}
Since both $W_{ij}$ and $W_{ik}$ are sub-Gaussian with parameter
$\sigma_{W}^{2}$, their product $W_{ij}W_{ik}$ is sub-exponential
with 
\[
\|W_{ij}W_{ik}\|_{\psi_{1}}\leq\|W_{ij}\|_{\psi_{2}}\|W_{ik}\|_{\psi_{2}}.
\]
Therefore, 
\[
U_{ikj}\coloneqq W_{ij}W_{ik}-\mathbb{E}(W_{j}W_{k})
\]
 is a centered sub-exponential random variable, which has
\[
\|U_{ikj}\|_{\psi_{1}}=\|W_{ij}W_{ik}-\mathbb{E}(W_{j}W_{k})\|_{\psi_{1}}=C\|W_{ij}W_{ik}\|_{\psi_{1}}\leq C\|W_{ij}\|_{\psi_{2}}\|W_{ik}\|_{\psi_{2}}.
\]
 Hence
\begin{align*}
 & \left|(\widehat{\bS}_{yy})_{kj}-(\mathbf{S}_{yy})_{kj}\right|\\
 & \leq\frac{1}{\mathbb{E}(W_{j}W_{k})}\Bigg|\frac{1}{n}\sum_{i=1}^{n}\Bigg\{ U_{ikj}\sum_{k'\in S_{j}}\sum_{k''\in S_{k}}X_{ik'}\beta_{k'j}^{*}X_{ik''}\beta_{k''k}^{*}\\
 & \qquad+U_{ikj}\sum_{k'\in S_{j}}X_{ik'}\beta_{k'j}^{*}\varepsilon_{ik}+U_{ikj}\sum_{k''\in S_{k}}X_{ik''}\beta_{k''k}^{*}\varepsilon_{ij}+U_{ikj}\varepsilon_{ij}\varepsilon_{ik})\Bigg\}\Bigg|\\
 & \leq\frac{1}{m_{\min}}\Bigg\{ X_{\max}^{2}B_{\max}^{2}s_{\max}^{2}\Bigg|\underbrace{\sum_{i=1}^{n}(1/n)U_{ikj}}_{T_{1}}\Bigg|+X_{\max}B_{\max}s_{\max}\Bigg|\underbrace{\sum_{i=1}^{n}(1/n)U_{ikj}\varepsilon_{ik}}_{T_{2}}\Bigg|\\
 & \qquad+X_{\max}B_{\max}s_{\max}\Bigg|\underbrace{\sum_{i=1}^{n}(1/n)U_{ikj}\varepsilon_{ij}}_{T_{3}}\Bigg|+\Bigg|\underbrace{\sum_{i=1}^{n}(1/n)U_{ikj}\varepsilon_{ij}\varepsilon_{ik}}_{T_{4}}\Bigg|\Bigg\}
\end{align*}
where $m_{\min}=\min_{j,k}|\mathbb{E}(W_{j}W_{k})|>0$, $s_{\max}=\max_{j}s_{j}$,
$X_{\max}=\max_{i,k}|X_{ik}|<\infty$, and $B_{\max}=\max_{k',j}|\beta_{k'j}^{*}|$.
Notice that in the above formula, within each term, we have multiple
products of sub-Gaussian random variables. Now we bound terms $T_{1}$,
$T_{2}$, $T_{3}$ and $T_{4}$, separately. 

Term $T_{1}$ is the average of $n$ independent, mean zero, sub-exponential
random variables. Therefore
\begin{align*}
\mathrm{Pr}\left(\left|T_{1}\right|\geq t\right) & =\mathrm{Pr}\left(\left|\sum_{i=1}^{n}(1/n)U_{ikj}\right|\geq t\right)\\
 & \leq2\exp\left[-cn\min\left(\frac{t^{2}}{\max_{i}\|U_{ikj}\|_{\psi_{1}}^{2}},\frac{t}{\max_{i}\|U_{ikj}\|_{\psi_{1}}}\right)\right]\\
 & \leq C\exp\left[-cn\min\left(\frac{t^{2}}{\max_{i}\|W_{ij}\|_{\psi_{2}}^{2}\|W_{ik}\|_{\psi_{2}}^{2}},\frac{t}{\max_{i}\|W_{ij}\|_{\psi_{2}}\|W_{ik}\|_{\psi_{2}}}\right)\right]\\
 & \leq C\exp\left[-cn\min\left(\frac{t^{2}}{\sigma_{W}^{4}},\frac{t}{\sigma_{W}^{2}}\right)\right]
\end{align*}

Now we look at term $T_{2}$ and $T_{3}$. We need to bound the product
$U_{ikj}\varepsilon_{ik}$. We know that sub-Gaussian variable $\varepsilon_{ik}$
is sub-exponential, since 
\[
\mathrm{Pr}\left(\left|\varepsilon_{ik}\right|\geq t\right)\leq2\exp(-t^{2}/\sigma_{\varepsilon}^{2})\leq2\exp(-t/\sigma_{\varepsilon}^{2}).
\]
We can use Lemma A.1 of \citet{gotze2021concentration}, which states
that the product of $D$ sub-exponential random variables has a $\alpha/D$-sub-exponential
tail with the Orlicz norm
\begin{equation}
\Big\|\prod_{i=1}^{D}X_{i}\Big\|_{\psi_{\frac{\alpha}{D}}}\leq\prod_{i=1}^{D}\|X_{i}\|_{\psi_{\alpha}.}\label{eq:alpha_sub_exp-1}
\end{equation}
Specifically, in our case $D=2$ and $\alpha=1$ in \eqref{eq:alpha_sub_exp-1},
we obtain that the product of two sub-exponential is a $1/2$-sub-exponential
with 
\[
\|U_{ikj}\varepsilon_{ik}\|_{\psi_{1/2}}\leq\|U_{ikj}\|_{\psi_{1}}\|\varepsilon_{ik}\|_{\psi_{1}}=\|W_{ij}\|_{\psi_{2}}\|W_{ij}\|_{\psi_{2}}\|\varepsilon_{ik}\|_{\psi_{2}}^{2}.
\]
Meanwhile, we know $U_{ikj}\varepsilon_{ik}$ is centered due to $\E(U_{ikj}\varepsilon_{ik})=\E(U_{ikj})\E(\varepsilon_{ik})=0$
and independence of $U_{ikj}$ and $\varepsilon_{ik}$. We apply Corollary
1.4 of \citet{gotze2021concentration}
\begin{align*}
\mathrm{Pr}\left(\left|T_{2}\right|\geq t\right) & =\mathrm{Pr}\left(\left|\sum_{i=1}^{n}(1/n)U_{ikj}\varepsilon_{ik}\right|\geq t\right)\\
 & \leq2\exp\left(-c\min\left(\frac{nt^{2}}{\|U_{ikj}\varepsilon_{ik}\|_{\psi_{1/2}}^{2}},\frac{n^{1/2}t^{1/2}}{\|U_{ikj}\varepsilon_{ik}\|_{\psi_{1/2}}^{1/2}}\right)\right)\\
 & \leq2\exp\left(-c\min\left(\frac{nt^{2}}{\|W_{ij}\|_{\psi_{2}}^{2}\|W_{ij}\|_{\psi_{2}}^{2}\|\varepsilon_{ik}\|_{\psi_{2}}^{4}},\frac{n^{1/2}t^{1/2}}{\|W_{ij}\|_{\psi_{2}}^{1/2}\|W_{ij}\|_{\psi_{2}}^{1/2}\|\varepsilon_{ik}\|_{\psi_{2}}}\right)\right)\\
 & \leq2\exp\left(-c\min\left(\frac{nt^{2}}{\sigma_{W}^{4}\sigma_{\varepsilon}^{4}},\frac{n^{1/2}t^{1/2}}{\sigma_{W}\sigma_{\varepsilon}}\right)\right)
\end{align*}
The bound for term $T_{3}$ is the same
\begin{align*}
\mathrm{Pr}\left(\left|T_{3}\right|\geq t\right) & =2\exp\left(-c\min\left(\frac{nt^{2}}{\sigma_{W}^{4}\sigma_{\varepsilon}^{4}},\frac{n^{1/2}t^{1/2}}{\sigma_{W}\sigma_{\varepsilon}}\right)\right).
\end{align*}

Now we look at term $T_{4}$, in which $\varepsilon_{ij}\varepsilon_{ik}$
as the product of two sub-Gaussian is sub-exponential with 
\[
\|\varepsilon_{ij}\varepsilon_{ik}\|_{\psi_{1}}\leq\|\varepsilon_{ij}\|_{\psi_{2}}\|\varepsilon_{ik}\|_{\psi_{2}}
\]
Therefore, we can view $U_{ikj}\varepsilon_{ij}\varepsilon_{ik}$
as the product of two sub-exponential random variables with
\[
\|U_{ikj}\varepsilon_{ij}\varepsilon_{ik}\|_{\psi_{1/2}}\leq\|U_{ikj}\|_{\psi_{1}}\|\varepsilon_{ij}\varepsilon_{ik}\|_{\psi_{1}}\leq\|W_{ij}\|_{\psi_{2}}\|W_{ik}\|_{\psi_{2}}\|\varepsilon_{ij}\|_{\psi_{2}}\|\varepsilon_{ik}\|_{\psi_{2}}.
\]

\begin{align*}
\mathrm{Pr}\left(\left|T_{4}\right|\geq t\right) & =\mathrm{Pr}\left(\left|\sum_{i=1}^{n}(1/n)U_{ikj}\varepsilon_{ij}\varepsilon_{ik}\right|\geq t\right)\\
 & \leq2\exp\left(-c\min\left(\frac{nt^{2}}{\|U_{ikj}\varepsilon_{ij}\varepsilon_{ik}\|_{\psi_{1/2}}^{2}},\frac{n^{1/2}t^{1/2}}{\|U_{ikj}\varepsilon_{ij}\varepsilon_{ik}\|_{\psi_{1/2}}^{1/2}}\right)\right)\\
 & \leq2\exp\left(-c\min\left(\frac{nt^{2}}{\|W_{ij}\|_{\psi_{2}}^{2}\|W_{ik}\|_{\psi_{2}}^{2}\|\varepsilon_{ij}\|_{\psi_{2}}^{2}\|\varepsilon_{ik}\|_{\psi_{2}}^{2}},\frac{n^{1/2}t^{1/2}}{\|W_{ij}\|_{\psi_{2}}^{1/2}\|W_{ik}\|_{\psi_{2}}^{1/2}\|\varepsilon_{ij}\|_{\psi_{2}}^{1/2}\|\varepsilon_{ik}\|_{\psi_{2}}^{1/2}}\right)\right)\\
 & \leq2\exp\left(-c\min\left(\frac{nt^{2}}{\sigma_{W}^{4}\sigma_{\varepsilon}^{4}},\frac{n^{1/2}t^{1/2}}{\sigma_{W}\sigma_{\varepsilon}}\right)\right).
\end{align*}
Define the event 
\begin{align*}
\mathcal{A}_{1} & =\left\{ \frac{X_{\max}^{2}B_{\max}^{2}s_{\max}^{2}}{m_{\min}}\left|T_{1}\right|\geq t/4\right\} \\
\mathcal{A}_{2} & =\left\{ \frac{X_{\max}B_{\max}s_{\max}}{m_{\min}}\left|T_{2}\right|\geq t/4\right\} \\
\mathcal{A}_{3} & =\left\{ \frac{X_{\max}B_{\max}s_{\max}}{m_{\min}}\left|T_{3}\right|\geq t/4\right\} \\
\mathcal{A}_{4} & =\left\{ \frac{1}{m_{\min}}\left|T_{4}\right|\geq t/4\right\} 
\end{align*}
and the event $\mathcal{G}=\{|(\widehat{\bS}_{yy})_{kj}-(\mathbf{S}_{yy})_{kj}|<t\}$.
Then by Boole's inequality
\begin{align*}
\mathrm{Pr}(\mathcal{G}) & =\mathrm{Pr}(|(\widehat{\bS}_{yy})_{kj}-(\mathbf{S}_{yy})_{kj}|<t)\\
 & \geq\Pr(\mathcal{A}_{1}^{c}\cap\mathcal{A}_{2}^{c}\cap\mathcal{A}_{3}^{c}\cap\mathcal{A}_{4}^{c})\\
 & =\Pr[(\mathcal{A}_{1}\cup\mathcal{A}_{2}\cup\mathcal{A}_{3}\cup\mathcal{A}_{4})^{c}]\\
 & \geq1-\Pr[(\mathcal{A}_{1}\cup\mathcal{A}_{2}\cup\mathcal{A}_{3}\cup\mathcal{A}_{4})]\\
 & \geq1-\Pr(\mathcal{A}_{1})-\Pr(\mathcal{A}_{2})-\Pr(\mathcal{A}_{3})-\Pr(\mathcal{A}_{4})\\
 & \geq1-\Pr\left\{ \left|T_{1}\right|\geq\frac{tm_{\min}}{4X_{\max}^{2}B_{\max}^{2}s_{\max}^{2}}\right\} -\Pr\left\{ \left|T_{2}\right|\geq\frac{tm_{\min}}{4X_{\max}B_{\max}s_{\max}}\right\} \\
 & \qquad\qquad-\Pr\left\{ \left|T_{3}\right|\geq\frac{tm_{\min}}{4X_{\max}B_{\max}s_{\max}}\right\} -\Pr\left\{ \left|T_{4}\right|\geq tm_{\min}/4\right\} \\
 & \geq1-C\exp\left(-c\min\left(\frac{nt^{2}m_{\min}^{2}}{X_{\max}^{4}B_{\max}^{4}s_{\max}^{4}\sigma_{W}^{4}},\frac{ntm_{\min}}{X_{\max}^{2}B_{\max}^{2}s_{\max}^{2}\sigma_{W}^{2}}\right)\right)\\
 & \qquad\qquad-C\exp\left(-c\min\left(\frac{nt^{2}m_{\min}^{2}}{X_{\max}^{2}B_{\max}^{2}s_{\max}^{2}\sigma_{W}^{4}\sigma_{\varepsilon}^{4}},\frac{n^{1/2}t^{1/2}m_{\min}^{1/2}}{X_{\max}^{1/2}B_{\max}^{1/2}s_{\max}^{1/2}\sigma_{W}\sigma_{\varepsilon}}\right)\right)\\
 & \qquad\qquad-C\exp\left(-c\min\left(\frac{nt^{2}m_{\min}^{2}}{X_{\max}^{2}B_{\max}^{2}s_{\max}^{2}\sigma_{W}^{4}\sigma_{\varepsilon}^{4}},\frac{n^{1/2}t^{1/2}m_{\min}^{1/2}}{X_{\max}^{1/2}B_{\max}^{1/2}s_{\max}^{1/2}\sigma_{W}\sigma_{\varepsilon}}\right)\right)\\
 & \qquad\qquad-C\exp\left(-c\min\left(\frac{nt^{2}m_{\min}^{2}}{\sigma_{W}^{4}\sigma_{\varepsilon}^{4}},\frac{n^{1/2}t^{1/2}m_{\min}^{1/2}}{\sigma_{W}\sigma_{\varepsilon}}\right)\right)
\end{align*}
Now, we can further simplify the terms involved inside the exponentiation
of the right hand side of the above inequality if we assume $t$ satisfies
some additional conditions. Specifically, if $t\leq t_{0}^{(a)}\coloneqq X_{\max}^{2}B_{\max}^{2}s_{\max}^{2}\sigma_{W}^{2}/m_{\min}$,
then we have 
\begin{align*}
\frac{nt^{2}m_{\min}^{2}}{X_{\max}^{4}B_{\max}^{4}s_{\max}^{4}\sigma_{W}^{4}} & \leq\frac{ntm_{\min}}{X_{\max}^{2}B_{\max}^{2}s_{\max}^{2}\sigma_{W}^{2}}.
\end{align*}
If $t\leq t_{0}^{(b)}\coloneqq X_{\max}B_{\max}s_{\max}\sigma_{W}^{2}\sigma_{\varepsilon}^{2}/(n^{1/3}m_{\min})$,
then 
\begin{align*}
\frac{nt^{2}m_{\min}^{2}}{X_{\max}^{2}B_{\max}^{2}s_{\max}^{2}\sigma_{W}^{4}\sigma_{\varepsilon}^{4}} & \leq\frac{n^{1/2}t^{1/2}m_{\min}^{1/2}}{X_{\max}^{1/2}B_{\max}^{1/2}s_{\max}^{1/2}\sigma_{W}\sigma_{\varepsilon}}
\end{align*}
and if $t\leq t_{0}^{(c)}\coloneqq\sigma_{W}^{2}\sigma_{\varepsilon}^{2}/(n^{1/3}m_{\min})$,
then 
\begin{align*}
\frac{nt^{2}m_{\min}^{2}}{\sigma_{W}^{4}\sigma_{\varepsilon}^{4}} & \leq\frac{n^{1/2}t^{1/2}m_{\min}^{1/2}}{\sigma_{W}\sigma_{\varepsilon}}
\end{align*}
Therefore, if we assume that 
\begin{align*}
t\leq & t_{0}^{(2)}\coloneqq\min(t_{0}^{(a)},t_{0}^{(b)},t_{0}^{(c)})\\
 & =\min(X_{\max}^{2}B_{\max}^{2}s_{\max}^{2}\sigma_{W}^{2}/m_{\min},X_{\max}B_{\max}s_{\max}\sigma_{W}^{2}\sigma_{\varepsilon}^{2}/m_{\min}n^{1/3},\sigma_{W}^{2}\sigma_{\varepsilon}^{2}/m_{\min}n^{1/3})
\end{align*}
then the lower bound for $\Pr(\mathcal{G})$ can be simplified as
\begin{align*}
\mathrm{Pr}\left[\mathcal{G}\right] & \geq1-C\exp\left(-\frac{cnt^{2}m_{\min}^{2}}{X_{\max}^{4}B_{\max}^{4}s_{\max}^{4}\sigma_{W}^{4}}\right)-C\exp\left(-\frac{cnt^{2}m_{\min}^{2}}{X_{\max}^{2}B_{\max}^{2}s_{\max}^{2}\sigma_{W}^{4}\sigma_{\varepsilon}^{4}}\right)\\
 & \qquad\qquad-C\exp\left(-\frac{cnt^{2}m_{\min}^{2}}{X_{\max}^{2}B_{\max}^{2}s_{\max}^{2}\sigma_{W}^{4}\sigma_{\varepsilon}^{4}}\right)-C\exp\left(-\frac{cnt^{2}m_{\min}^{2}}{\sigma_{W}^{4}\sigma_{\varepsilon}^{4}}\right)\\
 & \geq1-4C\max\Bigg\{\exp\left(-\frac{cnt^{2}m_{\min}^{2}}{X_{\max}^{4}B_{\max}^{4}s_{\max}^{4}\sigma_{W}^{4}}\right),\exp\left(-\frac{cnt^{2}m_{\min}^{2}}{X_{\max}^{2}B_{\max}^{2}s_{\max}^{2}\sigma_{W}^{4}\sigma_{\varepsilon}^{4}}\right),\exp\left(-\frac{cnt^{2}m_{\min}^{2}}{\sigma_{W}^{4}\sigma_{\varepsilon}^{4}}\right)\Bigg\}\\
 & \geq1-4C\Bigg\{\exp\left(-\frac{cnt^{2}m_{\min}^{2}}{\max\left\{ X_{\max}^{4}B_{\max}^{4}s_{\max}^{4}\sigma_{W}^{4},X_{\max}^{2}B_{\max}^{2}s_{\max}^{2}\sigma_{W}^{4}\sigma_{\varepsilon}^{4},\sigma_{W}^{4}\sigma_{\varepsilon}^{4}\right\} }\right)\Bigg\}\\
 & \geq1-4C\Bigg\{\exp\left(-\frac{cnt^{2}m_{\min}^{2}}{\sigma_{W}^{4}\max\left\{ X_{\max}^{4}B_{\max}^{4}s_{\max}^{4},X_{\max}^{2}B_{\max}^{2}s_{\max}^{2}\sigma_{\varepsilon}^{4},\sigma_{\varepsilon}^{4}\right\} }\right)\Bigg\}.
\end{align*}
Applying the union bound, we get
\[
\mathrm{Pr}(\VERT\widehat{\bS}_{yy}-\mathbf{S}_{yy}\VERT_{\max}\geq t)\leq4q^{2}C\Bigg\{\exp\left(-\frac{cnt^{2}m_{\min}^{2}}{\sigma_{W}^{4}\max\left\{ X_{\max}^{4}B_{\max}^{4}s_{\max}^{4},X_{\max}^{2}B_{\max}^{2}s_{\max}^{2}\sigma_{\varepsilon}^{4},\sigma_{\varepsilon}^{4}\right\} }\right)\Bigg\}.
\]
\end{proof}

\subsubsection{Proof of Lemma \ref{lem4} \label{Pr_Lem4} }
\begin{proof}
The deviation $\VERT\widehat{\bS}_{\varepsilon\varepsilon}-\bSigma_{\varepsilon\varepsilon}^{*}\VERT_{\max}$
can be decomposed and upper bounded as 
\begin{align*}
\VERT\widehat{\bS}_{\varepsilon\varepsilon}-\bSigma_{\varepsilon\varepsilon}^{*}\VERT_{\max} & =\VERT\widehat{\bS}_{yy}-\widehat{\bB}^{(1)\top}\bS_{xx}\widehat{\bB}^{(1)}-\bSigma_{\varepsilon\varepsilon}^{*}\VERT_{\max}\\
 & =\VERT\widehat{\bS}_{yy}-\bS_{yy}+\bS_{yy}-\widehat{\bB}^{(1)\top}\bS_{xx}\widehat{\bB}^{(1)}-\bSigma_{\varepsilon\varepsilon}^{*}\VERT_{\max}\\
 & =\VERT\widehat{\bS}_{yy}-\bS_{yy}+(1/n)\bY^{\top}\bY-\widehat{\bB}^{(1)\top}\bS_{xx}\widehat{\bB}^{(1)}-\bSigma_{\varepsilon\varepsilon}^{*}\VERT_{\max}\\
 & \leq\VERT\widehat{\bS}_{yy}-\bS_{yy}\VERT_{\max}+\VERT(1/n)(\bX\bB^{*}+\bepsilon)^{\top}(\bX\bB^{*}+\bepsilon)-\widehat{\bB}^{(1)\top}\bS_{xx}\widehat{\bB}^{(1)}-\bSigma_{\varepsilon\varepsilon}^{*}\VERT_{\max}\\
 & \leq\underbrace{\VERT(\widehat{\bS}_{yy}-\bS_{yy})\VERT_{\max}}_{T_{1}}+\underbrace{\VERT\widehat{\bB}^{(1)\top}\bS_{xx}\widehat{\bB}^{(1)}-\bB^{*\top}\text{\ensuremath{\bS}}_{xx}\bB^{*}\VERT_{\max}}_{T_{2}}\\
 & \qquad+\underbrace{\VERT(2/n)\bepsilon^{\top}\bX\bB^{*}\VERT_{\max}}_{T_{3}}+\underbrace{\VERT(1/n)\bepsilon^{\top}\bepsilon-\bSigma_{\varepsilon\varepsilon}^{*}\VERT_{\max}}_{T_{4}},
\end{align*}
by applying a series of triangular inequalities. In the following,
we bound each term on the right hand side of the above inequality
separately. The first term $T_{1}$ can be bound by applying Lemma
\ref{lem3} by setting 
\begin{align*}
{t_{1}} & {\coloneqq\sigma_{W}^{2}\max\left\{ X_{\max}^{2}B_{\max}^{2}s_{\max}^{2}/m_{\min},X_{\max}B_{\max}s_{\max}\sigma_{\varepsilon}^{2}/m_{\min},\sigma_{\varepsilon}^{2}/m_{\min}\right\} \sqrt{\frac{\log(q^{2})}{n}}}\\
{\leq t_{0}^{(2)}} & {\coloneqq\min(t_{0}^{(a)},t_{0}^{(b)},t_{0}^{(c)})}
\end{align*}
we obtain the following tail bound
\begin{align*}
 & \mathrm{Pr}(T_{1}\leq t_{1})\\
\geq & \mathrm{Pr}(\VERT\widehat{\bS}_{yy}-\mathbf{S}_{yy}\VERT_{\max}\leq t_{1})\\
\geq & 1-C\exp(-c\log q^{2})
\end{align*}

The second term $T_{2}$ can be simplified as follows:
\begin{align*}
\VERT(\widehat{\bB}^{(1)\top}\bS_{xx}\widehat{\bB}^{(1)}-\bB^{*\top}\bS_{xx}\bB^{*})\VERT_{\max} & =\VERT(\widehat{\bB}^{(1)}-\bB^{*})^{\top}\bS_{xx}(\widehat{\bB}^{(1)}-\bB^{*})+2[(\widehat{\bB}^{(1)}-\bB^{*})^{\top}\bS_{xx}\bB^{*}]\VERT_{\max}\\
 & \leq\VERT(\widehat{\bB}^{(1)}-\bB^{*})^{\top}\bS_{xx}(\widehat{\bB}^{(1)}-\bB^{*})\VERT_{\max}\\
 & \qquad+2\VERT(\widehat{\bB}^{(1)}-\bB^{*})^{\top}\bS_{xx}\bB^{*}\VERT_{\max}.
\end{align*}
We treat each term on the right hand side in sequence. Starting with
the first term, we have 
\begin{align*}
\VERT(\widehat{\bB}^{(1)}-\bB^{*})^{\top}\bS_{xx}(\widehat{\bB}^{(1)}-\bB^{*})\VERT_{\max} & \overset{(\text{i})}{\leq}\VERT(\widehat{\bB}^{(1)}-\bB^{*})^{\top}\VERT_{\infty}\VERT\bS_{xx}(\widehat{\bB}^{(1)}-\bB^{*})\VERT_{\max}\\
 & \leq\VERT(\widehat{\bB}^{(1)}-\bB^{*})^{\top}\VERT_{\infty}\VERT(\widehat{\bB}^{(1)}-\bB^{*})^{\top}\bS_{xx}^{\top}\VERT_{\max}\\
 & \leq\VERT(\widehat{\bB}^{(1)}-\bB^{*})^{\top}\VERT_{\infty}\VERT(\widehat{\bB}^{(1)}-\bB^{*})^{\top}\VERT_{\infty}\VERT\bS_{xx}\VERT_{\max}\\
 & \overset{(\text{ii})}{=}\VERT\widehat{\bB}^{(1)}-\bB^{*}\VERT_{1}\VERT\widehat{\bB}^{(1)}-\bB^{*}\VERT_{1}\VERT\bS_{xx}\VERT_{\max}\\
 & =\VERT\bS_{xx}\VERT_{\max}\VERT\widehat{\bB}^{(1)}-\bB^{*}\VERT_{1}^{2}\\
 & \leq X_{\max}^{2}\VERT\widehat{\bB}^{(1)}-\bB^{*}\VERT_{1}^{2}\\
 & =X_{\max}^{2}\left(\max_{l\in[q]}\|\widehat{\bbeta}_{l}-\bbeta_{l}^{*}\|_{1}\right)^{2}\\
 & \leq X_{\max}^{2}\left(\max_{l\in[q]}12s_{l}\lambda_{l}/\kappa_{l}\right)^{2}
\end{align*}
Inquality (i) follows from that for two matrices $\mathbf{A}$ and
$\mathbf{B}$
\begin{align*}
\VERT\bA\bB\VERT_{\max} & =\max_{i,k}|\sum_{j}A_{ij}B_{jk}|\\
 & \leq\max_{i,k}|\sum_{j}A_{ij}(\max_{j}|B_{jk}|)|\\
 & =\max_{i}|\sum_{j}A_{ij}|\max_{j,k}|B_{jk}|\\
 & =\max_{k}[\max_{j}|B_{jk}|]\max_{i}|\sum_{j}A_{ij}|\\
 & =\VERT\bA\VERT_{\infty}\VERT\bB\VERT_{\max},
\end{align*}
and equality (ii) follows from the relationship $\VERT\mathbf{A}\VERT_{1}=\VERT\mathbf{A}^{\top}\VERT_{\infty}$.
Since by Proposition \ref{prop:1}
\[
\Pr(\|\widehat{\bbeta}_{l}-\bbeta_{l}^{*}\|_{1}\le12s_{l}\lambda_{l}/\kappa_{l})\geq1-C\exp(-c\log p),
\]
Hence
\begin{align*}
 & \qquad\Pr\left[\VERT(\widehat{\bB}^{(1)}-\bB^{*})^{\top}\bS_{xx}(\widehat{\bB}^{(1)}-\bB^{*})\VERT_{\max}\leq X_{\max}^{2}\left(\max_{l\in[q]}12s_{l}\lambda_{l}/\kappa_{l}\right)^{2}\right]\\
 & \geq\Pr\left[X_{\max}^{2}\left(\max_{l\in[q]}\|\widehat{\bbeta}_{l}-\bbeta_{l}^{*}\|_{1}\right)^{2}\leq X_{\max}^{2}\left(\max_{l\in[q]}12s_{l}\lambda_{l}/\kappa_{l}\right)^{2}\right]\\
 & =\Pr\left[\max_{l\in[q]}\|\widehat{\bbeta}_{l}-\bbeta_{l}^{*}\|_{1}\le\max_{l\in[q]}12s_{l}\lambda_{l}/\kappa_{l}\right]\\
 & =\Pr\left[\cap_{l=1}^{q}\{\|\widehat{\bbeta}_{l}-\bbeta_{l}^{*}\|_{1}\le\max_{l\in[q]}12s_{l}\lambda_{l}/\kappa_{l}\}\right]\\
 & =\Pr\left[\left[\cup_{l=1}^{q}\{\|\widehat{\bbeta}_{l}-\bbeta_{l}^{*}\|_{1}\geq\max_{l\in[q]}12s_{l}\lambda_{l}/\kappa_{l}\}\right]^{c}\right]\\
 & =1-\Pr\left[\cup_{l=1}^{q}\{\|\widehat{\bbeta}_{l}-\bbeta_{l}^{*}\|_{1}\geq\max_{l\in[q]}12s_{l}\lambda_{l}/\kappa_{l}\}\right]\\
 & \geq1-\sum_{l=1}^{q}\Pr\left[\|\widehat{\bbeta}_{l}-\bbeta_{l}^{*}\|_{1}\geq\max_{l\in[q]}12s_{l}\lambda_{l}/\kappa_{l}\right]\\
 & \geq1-\sum_{l=1}^{q}\Pr\left[\|\widehat{\bbeta}_{l}-\bbeta_{l}^{*}\|_{1}\geq12s_{l}\lambda_{l}/\kappa_{l}\right]\\
 & \geq1-qC\exp(-c\log p)\\
 & =1-C\exp(-c\log(pq)).
\end{align*}

Next, applying Proposition \ref{prop:1} again, we get
\begin{align*}
\VERT(\widehat{\bB}^{(1)}-\bB^{*})^{\top}\bS_{xx}\bB^{*}\VERT_{\max} & \leq\VERT(\widehat{\bB}^{(1)}-\bB^{*})^{\top}\bS_{xx}\VERT_{\max}\VERT\bB^{*}\VERT_{1}\\
 & \leq\VERT(\widehat{\bB}^{(1)}-\bB^{*})^{\top}\VERT_{\infty}\VERT\bS_{xx}\VERT_{\max}\VERT\bB^{*}\VERT_{1}\\
 & \leq\VERT\widehat{\bB}^{(1)}-\bB^{*}\VERT_{1}\VERT\bS_{xx}\VERT_{\max}\VERT\bB^{*}\VERT_{1}\\
 & \leq X_{\max}^{2}s_{\max}B_{\max}\VERT\widehat{\bB}^{(1)}-\bB^{*}\VERT_{1}\\
 & =X_{\max}^{2}s_{\max}B_{\max}\left(\max_{l\in[q]}\|\widehat{\bbeta}_{l}-\bbeta_{l}^{*}\|_{1}\right)\\
 & \leq X_{\max}^{2}s_{\max}B_{\max}\max_{l\in[q]}12s_{l}\lambda_{l}/\kappa_{l}
\end{align*}
with probability
\begin{align*}
 & \qquad\Pr\left[\VERT(\widehat{\bB}^{(1)}-\bB^{*})^{\top}\bS_{xx}\bB^{*}\VERT_{\max}\leq X_{\max}^{2}s_{\max}B_{\max}\max_{l\in[q]}12s_{l}\lambda_{l}/\kappa_{l}\right]\\
 & \geq\Pr\left[X_{\max}^{2}s_{\max}B_{\max}\max_{l\in[q]}\|\widehat{\bbeta}_{l}-\bbeta_{l}^{*}\|_{1}\leq X_{\max}^{2}s_{\max}B_{\max}\max_{l\in[q]}12s_{l}\lambda_{l}/\kappa_{l}\right]\\
 & \geq\Pr\left[\max_{l\in[q]}\|\widehat{\bbeta}_{l}-\bbeta_{l}^{*}\|_{1}\leq\max_{l\in[q]}12s_{l}\lambda_{l}/\kappa_{l}\right]\\
 & =\Pr\left[\cap_{l=1}^{q}\{\|\widehat{\bbeta}_{l}-\bbeta_{l}^{*}\|_{1}\le\max_{l\in[q]}12s_{l}\lambda_{l}/\kappa_{l}\}\right]\\
 & =1-\Pr\left[\cup_{l=1}^{q}\{\|\widehat{\bbeta}_{l}-\bbeta_{l}^{*}\|_{1}\geq\max_{l\in[q]}12s_{l}\lambda_{l}/\kappa_{l}\}\right]\\
 & \geq1-\sum_{l=1}^{q}\Pr\left[\|\widehat{\bbeta}_{l}-\bbeta_{l}^{*}\|_{1}\geq\max_{l\in[q]}12s_{l}\lambda_{l}/\kappa_{l}\right]\\
 & \geq1-\sum_{l=1}^{q}\Pr\left[\|\widehat{\bbeta}_{l}-\bbeta_{l}^{*}\|_{1}\geq12s_{l}\lambda_{l}/\kappa_{l}\right]\\
 & \leq1-C\exp(-c\log(pq)).
\end{align*}
Now, we bound the third term $T_{3}$. Let us consider the $j$th
row and $k$th column for $j,k=1,\ldots,q$, 
\[
\left|\left(\frac{2}{n}\bepsilon^{\top}\bX\bB^{*}\right)_{jk}\right|=\left|\frac{2}{n}\sum_{i=1}^{n}\sum_{k'\in S_{k}}X_{ik'}\beta_{k'k}^{*}\varepsilon_{ij}\right|\leq\left|2X_{\max}s_{\max}B_{\max}\sum_{i=1}^{n}(1/n)\varepsilon_{ij}\right|.
\]
Since $\varepsilon_{ij}$ is sub-Gaussian with parameter $\sigma_{\varepsilon}^{2}$,
we can bound it by applying Theorem 2.6.3 general Hoeffding inequality
on page 27 of \citet{vershynin2018high}
\begin{align*}
\mathrm{Pr}\left[\Big|\sum_{i=1}^{n}(1/n)\varepsilon_{ij}\Big|\geq t\right] & <2\exp\left\{ -\frac{ct^{2}}{(\max_{i}\|\varepsilon_{ij}\|_{\psi_{2}})^{2}\sum_{i=1}^{n}(1/n)^{2}}\right\} \\
 & <2\exp\left\{ -\frac{cnt^{2}}{\sigma_{\varepsilon}^{2}}\right\} .
\end{align*}
Hence, 
\begin{align*}
\mathrm{Pr}\left[\Big|\sum_{i=1}^{n}(2X_{\max}s_{\max}B_{\max}/n)\varepsilon_{ij}\Big|\geq t\right] & <2\exp\left\{ -\frac{cnt^{2}}{\sigma_{\varepsilon}^{2}X_{\max}^{2}s_{\max}^{2}B_{\max}^{2}}\right\} .
\end{align*}
Applying the union bound, we get
\[
\mathrm{Pr}(\VERT(2/n)\bepsilon^{\top}\bX\bB^{*}\VERT_{\max}\geq t)\leq q^{2}C\exp\left\{ -\frac{cnt^{2}}{\sigma_{\varepsilon}^{2}X_{\max}^{2}s_{\max}^{2}B_{\max}^{2}}\right\} .
\]
By setting 
\[
t_{2}\coloneqq\sigma_{\varepsilon}X_{\max}s_{\max}B_{\max}\sqrt{\frac{\log(q^{2})}{n}}
\]
we get the tail bound for $T_{3}$
\[
\mathrm{Pr}(\VERT(2/n)\bepsilon^{\top}\bX\bB^{*}\VERT_{\max}\leq t_{2})\geq1-C\exp(-c\log(q^{2})).
\]
Now, we bound $T_{4}$. For the $j$th row and $k$th column of $(1/n)\bepsilon^{\top}\bepsilon-\bSigma_{\varepsilon\varepsilon}^{*}$,
we have 
\[
((1/n)\bepsilon^{\top}\bepsilon-\bSigma_{\varepsilon\varepsilon}^{*})_{jk}=\sum_{i=1}^{n}(1/n)[\varepsilon_{ij}\varepsilon_{ik}-\mathbb{E}(\varepsilon_{j}\varepsilon_{k})]
\]
where $\varepsilon_{ij}$ and $\varepsilon_{ik}$ are sub-Gaussian
with parameter $\sigma_{\varepsilon}^{2}$. Their product $\varepsilon_{ij}\varepsilon_{ik}$
is sub-exponential with 
\[
\|\varepsilon_{ij}\varepsilon_{ik}\|_{\psi_{1}}\leq\|\varepsilon_{ij}\|_{\psi_{2}}\|\varepsilon_{ik}\|_{\psi_{2}}.
\]
Therefore, define
\[
V_{ikj}\coloneqq\varepsilon_{ij}\varepsilon_{ik}-\mathbb{E}(\varepsilon_{j}\varepsilon_{k})
\]
as a centered sub-exponential random variable, which has
\[
\|V_{ikj}\|_{\psi_{1}}=\|\varepsilon_{ij}\varepsilon_{ik}-\mathbb{E}(\varepsilon_{j}\varepsilon_{k})\|_{\psi_{1}}=C\|\varepsilon_{ij}\varepsilon_{ik}\|_{\psi_{1}}\leq C\|\varepsilon_{ij}\|_{\psi_{2}}\|\varepsilon_{ik}\|_{\psi_{2}}.
\]
Therefore
\begin{align*}
\mathrm{Pr}\left(\left|\sum_{i=1}^{n}((1/n)\bepsilon^{\top}\bepsilon-\bSigma_{\varepsilon\varepsilon}^{*})_{jk}\right|\geq t\right) & =\mathrm{Pr}\left(\left|\sum_{i=1}^{n}(1/n)V_{ikj}\right|\geq t\right)\\
 & \leq2\exp\left[-cn\min\left(\frac{t^{2}}{\max_{i}\|V_{ikj}\|_{\psi_{1}}^{2}},\frac{t}{\max_{i}\|V_{ikj}\|_{\psi_{1}}}\right)\right]\\
 & \leq C\exp\left[-cn\min\left(\frac{t^{2}}{\max_{i}\|\varepsilon_{ij}\|_{\psi_{2}}^{2}\|\varepsilon_{ik}\|_{\psi_{2}}^{2}},\frac{t}{\max_{i}\|\varepsilon_{ij}\|_{\psi_{2}}\|\varepsilon_{ik}\|_{\psi_{2}}}\right)\right]\\
 & \leq C\exp\left[-cn\min\left(\frac{t^{2}}{\sigma_{\varepsilon}^{4}},\frac{t}{\sigma_{\varepsilon}^{2}}\right)\right]\\
 & \leq C\exp\left[-cn\frac{t^{2}}{\sigma_{\varepsilon}^{4}}\right]
\end{align*}
The last inequality holds if we assume that $t$ is chosen satisfying
$t\leq t_{0}^{(3)}\coloneqq\sigma_{\varepsilon}^{2}$. Therefore,
applying union bound, we get
\[
\Pr(\VERT(1/n)\bepsilon^{\top}\bepsilon-\bSigma_{\varepsilon\varepsilon}^{*}\VERT_{\max}\geq t)\leq q^{2}C\exp\left[-cn\left(\frac{t^{2}}{\sigma_{\varepsilon}^{4}}\right)\right].
\]
Define, 
\[
t_{3}\coloneqq\sigma_{\varepsilon}^{2}\sqrt{\frac{\log(q^{2})}{n}}
\]
the requirement $t_{3}\leq t_{0}^{(3)}$ implies that ${\sqrt{\frac{\log(q^{2})}{n}}\leq1}$
must be satisfied. Hence, we get the tail bound for $T_{4}$
\[
\Pr(\VERT(1/n)\bepsilon^{\top}\bepsilon-\bSigma_{\varepsilon\varepsilon}^{*}\VERT_{\max}\leq t_{3})\geq1-C\exp(-c\log(q^{2})).
\]
Define the event 
\begin{align*}
\mathcal{B}_{1} & =\left\{ \VERT\widehat{\bS}_{yy}-\bS_{yy}\VERT_{\max}\leq t_{1}\right\} \\
\mathcal{B}_{2} & =\left\{ \VERT(\widehat{\bB}^{(1)}-\bB^{*})^{\top}\bS_{xx}(\widehat{\bB}^{(1)}-\bB^{*})\VERT_{\max}\leq X_{\max}^{2}\left(\max_{l\in[q]}12s_{l}\lambda_{l}/\kappa_{l}\right)^{2}\right\} \\
\mathcal{B}_{3} & =\left\{ \VERT(\widehat{\bB}^{(1)}-\bB^{*})^{\top}\bS_{xx}\bB^{*}\VERT_{\max}\leq X_{\max}^{2}s_{\max}B_{\max}\max_{l\in[q]}12s_{l}\lambda_{l}/\kappa_{l}\right\} \\
\mathcal{B}_{4} & =\left\{ \VERT(2/n)\bepsilon^{\top}\bX\bB^{*}\VERT_{\max}\leq t_{2}\right\} \\
\mathcal{B}_{5} & =\left\{ \VERT(1/n)\bepsilon^{\top}\bepsilon-\bSigma_{\varepsilon\varepsilon}^{*}\VERT_{\max}\leq t_{3}\right\} ,
\end{align*}
and
\begin{align}
\Delta & =t_{1}+X_{\max}^{2}\left(\max_{l\in[q]}12s_{l}\lambda_{l}/\kappa_{l}\right)^{2}+X_{\max}^{2}s_{\max}B_{\max}\max_{l\in[q]}12s_{l}\lambda_{l}/\kappa_{l}\label{eq:Delta}\\
 & \qquad t_{2}+t_{3}\nonumber \\
 & =\sigma_{W}^{2}\max\left\{ X_{\max}^{2}B_{\max}^{2}s_{\max}^{2}/m_{\min},X_{\max}B_{\max}s_{\max}\sigma_{\varepsilon}^{2}/m_{\min},\sigma_{\varepsilon}^{2}/m_{\min}\right\} \sqrt{\frac{\log(q^{2})}{n}}\nonumber \\
 & \qquad+X_{\max}^{2}\left(\max_{l\in[q]}12s_{l}\lambda_{l}/\kappa_{l}\right)^{2}+X_{\max}^{2}s_{\max}B_{\max}\max_{l\in[q]}12s_{l}\lambda_{l}/\kappa_{l}+\sigma_{\varepsilon}X_{\max}s_{\max}B_{\max}\sqrt{\frac{\log(q^{2})}{n}}\nonumber \\
 & \qquad+\sigma_{\varepsilon}^{2}\sqrt{\frac{\log(q^{2})}{n}}\nonumber 
\end{align}
and the event $\mathcal{H}=\{\VERT\widehat{\bS}_{\varepsilon\varepsilon}-\bSigma_{\varepsilon\varepsilon}^{*}\VERT_{\max}<\Delta\}$.
Then by Boole's inequality
\begin{align*}
\mathrm{Pr}(\mathcal{H}) & =\mathrm{Pr}(\VERT\widehat{\bS}_{\varepsilon\varepsilon}-\bSigma_{\varepsilon\varepsilon}^{*}\VERT_{\max}\leq\Delta)\\
 & \geq\mathrm{Pr}(\VERT\widehat{\bS}_{yy}-\bS_{yy}\VERT_{\max}+\VERT(\widehat{\bB}^{(1)}-\bB^{*})^{\top}\bS_{xx}(\widehat{\bB}^{(1)}-\bB^{*})\VERT_{\max}\\
 & \qquad+\VERT(\widehat{\bB}^{(1)}-\bB^{*})^{\top}\bS_{xx}\bB^{*}\VERT_{\max}+\VERT(2/n)\bepsilon^{\top}\bX\bB^{*}\VERT_{\max}+\VERT(1/n)\bepsilon^{\top}\bepsilon-\bSigma_{\varepsilon\varepsilon}^{*}\VERT_{\max}\leq\Delta)\\
 & \geq\Pr(\mathcal{B}_{1}\cap\mathcal{B}_{2}\cap\mathcal{B}_{3}\cap\mathcal{B}_{4}\cap\mathcal{B}_{5})\\
 & =\Pr[(\mathcal{B}_{1}^{c}\cup\mathcal{B}_{2}^{c}\cup\mathcal{B}_{3}^{c}\cup\mathcal{B}_{4}^{c}\cup\mathcal{B}_{5}^{c})^{c}]\\
 & \geq1-\Pr[(\mathcal{B}_{1}^{c}\cup\mathcal{B}_{2}^{c}\cup\mathcal{B}_{3}^{c}\cup\mathcal{B}_{4}^{c}\cup\mathcal{B}_{5}^{c})]\\
 & \geq1-\Pr(\mathcal{B}_{1}^{c})-\Pr(\mathcal{B}_{2}^{c})-\Pr(\mathcal{B}_{3}^{c})-\Pr(\mathcal{B}_{4}^{c})-\Pr(\mathcal{B}_{5}^{c})\\
 & \geq1-C\exp(-c\log q^{2})-C\exp(-c\log(pq))-C\exp(-c\log(pq))\\
 & \qquad-C\exp(-c\log(q^{2}))-C\exp(-c\log(q^{2}))\\
 & =1-C\exp(-c\log q^{2})-C\exp(-c\log(pq))
\end{align*}
\end{proof}

\subsubsection{Proof of Proposition \ref{prop:2} \label{Pr_Prop2}}
\begin{proof}
Since $\widehat{\bS}_{\varepsilon\varepsilon}$ is not necessarily
positive semi-definite, we can produce an update $\widetilde{\bS}_{\varepsilon\varepsilon}$
as described in \citet{datta2017} by projecting it onto the nearest
semi-definite cone and substituting $\widehat{\bS}_{\varepsilon\varepsilon}$
by $\widetilde{\bS}_{\varepsilon\varepsilon}$. Then, $\widehat{\bS}_{\varepsilon\varepsilon}$
would satisfy an inequality analogous to \eqref{eq:Proj_conv_rate}.
Specifically, 
\[
\VERT\widetilde{\bS}_{\varepsilon\varepsilon}-\bSigma_{\varepsilon\varepsilon}^{*}\VERT_{\max}\leq2\VERT\widehat{\bS}_{\varepsilon\varepsilon}-\bSigma_{\varepsilon\varepsilon}^{*}\VERT_{\max}\leq2\Delta.
\]

Applying Lemma \ref{lem4}, we obtain $\VERT\widetilde{\bS}_{\varepsilon\varepsilon}-\bSigma_{\varepsilon\varepsilon}^{*}\VERT_{\max}\leq2\Delta$
occurs with probability $1-C\exp(-c\log q^{2})-C\exp(-c\log(pq))$.
The conclusion of this Proposition follows from a slight variation
of Theorem 1 from \citet{ravikumar2011high} where the observed sample
was uncontaminated and one could calculate the sample covariance matrix
$\widehat{\bSigma}_{\varepsilon\varepsilon}$ with $\VERT\widehat{\bSigma}_{\varepsilon\varepsilon}-\bSigma_{\varepsilon\varepsilon}^{*}\VERT_{\max}\leq\bar{\delta}_{f}(n,p^{\tau})$.
In our case, $\widehat{\bSigma}_{\varepsilon\varepsilon}$ is replaced
by $\widetilde{\bS}_{\varepsilon\varepsilon}$ and $\bar{\delta}_{f}(n,p^{\tau})=2\Delta$.
Suppose $\widetilde{\bS}_{\varepsilon\varepsilon}$ satisfies the
error bound $\VERT\widetilde{\bS}_{\varepsilon\varepsilon}-\bSigma_{\varepsilon\varepsilon}^{*}\VERT_{\max}\leq2\Delta$
on the intersection of events, that is, $\mathcal{B}_{1}\cap\mathcal{B}_{2}\cap\mathcal{B}_{3}\cap\mathcal{B}_{4}\cap\mathcal{B}_{5}$
defined in Lemma \ref{lem4} and if the tuning parameter $\lambda_{\bTheta}$
satisfies \eqref{eq:tuning_theta}, then by Theorem 1 from \citet{ravikumar2011high},
we have
\[
\VERT\widehat{\bTheta}_{\varepsilon\varepsilon}-\bTheta_{\varepsilon\varepsilon}^{*}\VERT_{\max}\leq\left\{ 2\kappa_{\bGamma}(1+\frac{8}{\alpha})\right\} \Delta\coloneqq\Delta_{\infty}(\bTheta_{\varepsilon\varepsilon}^{*}).
\]
 We can also show that \eqref{eq:Theta_op_bound} holds as follows
\[
\VERT\widehat{\bTheta}_{\varepsilon\varepsilon}-\bTheta_{\varepsilon\varepsilon}^{*}\VERT_{\mathrm{2}}\leq\sqrt{d_{q}d_{q}}\VERT\widehat{\bTheta}_{\varepsilon\varepsilon}-\bTheta_{\varepsilon\varepsilon}^{*}\VERT_{\max}=d_{q}\VERT\widehat{\bTheta}_{\varepsilon\varepsilon}-\bTheta_{\varepsilon\varepsilon}^{*}\VERT_{\max}\leq d_{q}\Delta_{\infty}(\bTheta_{\varepsilon\varepsilon}^{*})\coloneqq\Delta_{2}(\bTheta_{\varepsilon\varepsilon}^{*})
\]
and 
\[
\VERT\widehat{\bTheta}_{\varepsilon\varepsilon}-\bTheta_{\varepsilon\varepsilon}^{*}\VERT_{\mathrm{1}}\leq\sqrt{d_{q}}\VERT\widehat{\bTheta}_{\varepsilon\varepsilon}-\bTheta_{\varepsilon\varepsilon}^{*}\VERT_{2}=\sqrt{d_{q}}\Delta_{2}(\bTheta_{\varepsilon\varepsilon}^{*})\coloneqq\Delta_{1}(\bTheta_{\varepsilon\varepsilon}^{*}).
\]
\end{proof}

\subsubsection{Proof of Theorem 1 \label{Pr_Thm1}}

\begin{proof}
The proof relies partly on Proposition \ref{prop:RSC_final} stated
below. The proposition verifies that the empirical loss function at
this stage satisfies the RE condition.
\begin{prop}
\label{prop:RSC_final} Suppose that Assumption \ref{assu:3} and
the assumptions of Proposition \ref{prop:2} hold. Then the empirical
loss $\mathcal{L}(\cdot;\bS_{xx},\widehat{\bS}_{xy},\widehat{\bTheta}_{\varepsilon\varepsilon})$
satisfies RE condition with curvature $\kappa^{\prime}$ introduced
in \eqref{eq:kappa_prime} and tolerance function equal to zero over
the cone set $\mathbb{C}(S)$. 
\end{prop}
\begin{proof}
We fix arbitrary $\bDelta\in\mathbb{C}(S)$. we have 
\begin{align}
\mathcal{E}\mathcal{L}(\bDelta,\bB^{*};\bS_{xx},\widehat{\bS}_{xy},\widehat{\bTheta}_{\varepsilon\varepsilon}) & =\vec{(}\bDelta)^{\top}\bTheta_{\varepsilon\varepsilon}^{*}\otimes\bS_{xx}\vec{(}\bDelta)/2\nonumber \\
 & \qquad+\langle\bDelta^{\top}\bS_{xx}\bDelta,\widehat{\bTheta}_{\varepsilon\varepsilon}-\bTheta_{\varepsilon\varepsilon}^{*}\rangle/2\label{eq:Err1}
\end{align}
 For the first term, by Assumption \ref{assu:3}, 
\[
\vec{(}\bDelta)^{\top}\bTheta_{\varepsilon\varepsilon}^{*}\otimes\bS_{xx}\vec{(}\bDelta)/2\geq\kappa\VERT\bDelta\VERT_{F}^{2}
\]

For the second term, we have
\begin{align}
\Big|\langle\bDelta^{\top}\bS_{xx}\bDelta,\widehat{\bTheta}_{\varepsilon\varepsilon}-\bTheta_{\varepsilon\varepsilon}^{*}\rangle\Big| & =\Big|\langle\bS_{xx}\bDelta,\bDelta(\widehat{\bTheta}_{\varepsilon\varepsilon}-\bTheta_{\varepsilon\varepsilon}^{*})\rangle\Big|\nonumber \\
 & \overset{(\text{i})}{\leq}\VERT\bS_{xx}\bDelta\VERT_{F}\times\VERT\bDelta(\widehat{\bTheta}_{\varepsilon\varepsilon}-\bTheta_{\varepsilon\varepsilon}^{*})\VERT_{F}\nonumber \\
 & \overset{(\text{ii})}{\leq}\VERT\bS_{xx}\VERT_{F}\VERT\bDelta\VERT_{F}\times\VERT\bDelta\VERT_{F}\VERT\widehat{\bTheta}_{\varepsilon\varepsilon}-\bTheta_{\varepsilon\varepsilon}^{*}\VERT_{F}\nonumber \\
 & \overset{(\text{iii})}{\leq}{\VERT\bS_{xx}\VERT_{2}\VERT\widehat{\bTheta}_{\varepsilon\varepsilon}-\bTheta_{\varepsilon\varepsilon}^{*}\VERT_{2}\VERT\bDelta\VERT_{F}^{2}},\label{eq:Er1a}
\end{align}
where inequality (i) follows from Hölder's inequality, inequality
(ii) follows from the submultiplicative property of the Frobenius
norm and inequality (iii) follows from the fact $\VERT\bA\VERT_{F}=\VERT\bA\VERT_{2}$
since the trace of a matrix is equal to the sum of its eigenvalues.

Therefore, \eqref{eq:Err1} can be bounded as
\begin{align}
\mathcal{E}\mathcal{L}(\bDelta,\bB^{*};\bS_{xx},\widehat{\bS}_{xy},\widehat{\bTheta}_{\varepsilon\varepsilon}) & \geq\kappa\VERT\bDelta\VERT_{F}^{2}-(\VERT\bS_{xx}\VERT_{2}\VERT\widehat{\bTheta}_{\varepsilon\varepsilon}-\bTheta_{\varepsilon\varepsilon}^{*}\VERT_{2}\VERT\bDelta\VERT_{F}^{2})/2\label{eq:Simp_err}
\end{align}
For the regularized Lasso problem, since $\bDelta\in\mathbb{C}(S)$,
and therefore, we can write 
\begin{equation}
\VERT\bDelta\VERT_{1,1}=\VERT\bDelta_{S}\VERT_{1,1}+\VERT\bDelta_{S^{c}}\VERT_{1,1}\leq4\VERT\bDelta_{S}\VERT_{1,1}\leq4\sqrt{s}\VERT\bDelta\VERT_{F}.\label{eq:Delta_cone}
\end{equation}
Combining \eqref{eq:Simp_err} and \eqref{eq:Delta_cone}, we conclude
that 
\begin{align}
\mathcal{E}\mathcal{L}(\bDelta,\bB^{*};\bS_{xx},\widehat{\bS}_{xy},\widehat{\bTheta}_{\varepsilon\varepsilon}) & \geq\left\{ \kappa-(\VERT\bS_{xx}\VERT_{2}\VERT\widehat{\bTheta}_{\varepsilon\varepsilon}-\bTheta_{\varepsilon\varepsilon}^{*}\VERT_{2})/2\right\} \VERT\bDelta\VERT_{F}^{2}\nonumber \\
 & \geq(\kappa-\VERT\bS_{xx}\VERT_{2}\Delta_{1}(\bTheta_{\varepsilon\varepsilon}^{*})/2)\VERT\bDelta\VERT_{F}^{2}\label{eq:Simp_err2}
\end{align}
If the assumptions of Proposition \ref{prop:2} are satisfied, then
\eqref{eq:Theta_max_bound} and \eqref{eq:Theta_op_bound} also hold.
Then we can further bound the right-hand side of \eqref{eq:Simp_err2}
from below. This concludes the proof of Proposition \ref{prop:RSC_final}. 
\end{proof}

Next, we want to apply Theorem 7.13 from \citet{wainwright2019high}.
To do so, first, we check if the conditions mentioned in Proposition
\ref{prop:1} holds. We set, 
\[
\mathcal{M}\coloneqq\left\{ \bDelta\in\mathbb{R}^{p\times q}:(\bDelta)_{kl}=0\qquad\forall(k,l)\in S^{c}\right\} .
\]
Then \textbf{$\bB^{*}\in\mathcal{M}$}, and the penalty $\mathcal{R}$
is decomposable with respect to $(\mathcal{M},\mathcal{M}^{\perp})$.
Next, by Proposition\textbf{ }\ref{prop:RSC_final} and the assumption
on $\kappa^{\prime}$, over the cone set $\mathbb{C}(S)$, the loss
$\mathcal{L}(\cdot;\bS_{xx},\widehat{\bS}_{xy},\widehat{\bTheta}_{\varepsilon\varepsilon})$
satisfies RE condition with tolerance function equal to zero and curvature
$\kappa^{\prime}\geq\kappa>0$. Finally, the dual norm of $\mathcal{R}$
is $\mathcal{R}^{*}(\cdot)=\VERT\cdot\VERT_{\max}$, and 
\begin{align*}
\mathcal{R}^{*}\left\{ \nabla_{\bB}\mathcal{L}(\bB^{*};\bS_{xx},\widehat{\bS}_{xy},\widehat{\bTheta}_{\varepsilon\varepsilon})\right\}  & =\VERT\bS_{xx}\bB^{*}\widehat{\bTheta}_{\varepsilon\varepsilon}-\widehat{\bS}_{xy}\widehat{\bTheta}_{\varepsilon\varepsilon}\VERT_{\max}\\
 & =\VERT(\widehat{\bS}_{xy}-\bS_{xx}\bB^{*})(\widehat{\bTheta}_{\varepsilon\varepsilon}-\bTheta_{\varepsilon\varepsilon}^{*})+(\widehat{\bS}_{xy}-\bS_{xx}\bB^{*})\bTheta_{\varepsilon\varepsilon}^{*}\VERT_{\max}\\
 & \leq\VERT(\widehat{\bS}_{xy}-\bS_{xx}\bB^{*})(\widehat{\bTheta}_{\varepsilon\varepsilon}-\bTheta_{\varepsilon\varepsilon}^{*})\VERT_{\max}+\VERT(\widehat{\bS}_{xy}-\bS_{xx}\bB^{*})\bTheta_{\varepsilon\varepsilon}^{*}\VERT_{\max}\\
 & \overset{\text{(i)}}{\leq}\VERT\widehat{\bS}_{xy}-\bS_{xx}\bB^{*}\VERT_{\max}\VERT\widehat{\bTheta}_{\varepsilon\varepsilon}-\bTheta_{\varepsilon\varepsilon}^{*}\VERT_{\infty}+\VERT\widehat{\bS}_{xy}-\bS_{xx}\bB^{*}\VERT_{\max}\VERT\bTheta_{\varepsilon\varepsilon}^{*}\VERT_{\infty}\\
 & =\VERT\widehat{\bS}_{xy}-\bS_{xx}\bB^{*}\VERT_{\max}(\VERT\widehat{\bTheta}_{\varepsilon\varepsilon}-\bTheta_{\varepsilon\varepsilon}^{*}\VERT_{\infty}+\VERT\bTheta_{\varepsilon\varepsilon}^{*}\VERT_{\infty})\\
 & =\VERT\widehat{\bS}_{xy}-\bS_{xx}\bB^{*}\VERT_{\max}(\VERT\widehat{\bTheta}_{\varepsilon\varepsilon}-\bTheta_{\varepsilon\varepsilon}^{*}\VERT_{1}+\VERT\bTheta_{\varepsilon\varepsilon}^{*}\VERT_{1})\\
 & \leq(\lambda_{0}/2)(\VERT\widehat{\bTheta}_{\varepsilon\varepsilon}-\bTheta_{\varepsilon\varepsilon}^{*}\VERT_{1}+\VERT\bTheta_{\varepsilon\varepsilon}^{*}\VERT_{1})
\end{align*}
Inequality (i) follows from the fact $\VERT\bA\bB\VERT_{\max}\leq\VERT\bA\VERT_{\infty}\VERT\bB\VERT_{\max}$
and equality (ii) follows from the relationship $\VERT\mathbf{A}\VERT_{1}=\VERT\mathbf{A}^{\top}\VERT_{\infty}$.
Denote $C_{\Theta}=\VERT\widehat{\bTheta}_{\varepsilon\varepsilon}-\bTheta_{\varepsilon\varepsilon}^{*}\VERT_{1}+\VERT\bTheta_{\varepsilon\varepsilon}^{*}\VERT_{1}$,
and choose 
\[
\lambda_{0}/2\coloneqq X_{\max}\max\left[\sigma_{W}s_{\max}X_{\max}B_{\max}/\mu_{\min},\sigma_{W}\sigma_{\varepsilon}/\mu_{\min},\sigma_{\varepsilon}\right]\sqrt{\frac{\log p}{n}}
\]
then we have 
\begin{align*}
 & \qquad\Pr(\mathcal{R}^{*}\{\nabla_{\bB}\mathcal{L}(\bB^{*};\bS_{xx},\widehat{\bS}_{xy},\widehat{\bTheta}_{\varepsilon\varepsilon})\}\leq(\lambda_{0}/2)C_{\Theta})\\
 & \geq\Pr(\VERT\widehat{\bS}_{xy}-\bS_{xx}\bB^{*}\VERT_{\max}C_{\Theta}\leq(\lambda_{0}/2)C_{\Theta})\\
 & =\Pr(\VERT\widehat{\bS}_{xy}-\bS_{xx}\bB^{*}\VERT_{\max}\leq\lambda_{0}/2)\\
 & =\Pr(\max_{l}\|(\widehat{\bS}_{xy})_{\bullet l}-\bS_{xx}\bbeta_{l}^{*}\|_{\infty}\leq\lambda_{0}/2)\\
 & =\Pr(\cap_{l}\{\|(\widehat{\bS}_{xy})_{\bullet l}-\bS_{xx}\bbeta_{l}^{*}\|_{\infty}\leq\lambda_{0}/2\})\\
 & =1-\Pr(\cup_{l}\{\|(\widehat{\bS}_{xy})_{\bullet l}-\bS_{xx}\bbeta_{l}^{*}\|_{\infty}\geq\lambda_{0}/2\})\\
 & \overset{\text{(i)}}{\geq}1-\sum_{l=1}^{q}\Pr(\|(\widehat{\bS}_{xy})_{\bullet l}-\bS_{xx}\bbeta_{l}^{*}\|_{\infty}\geq\lambda_{0}/2)\\
 & =1-qC\exp(-c\log p)
\end{align*}
Inequality (i) follows from the tail probability bound in \eqref{eq:good_event_sxy}.
The conclusion from \eqref{eq:B_final_est} follows from Theorem 7.13
in \citet{wainwright2019high} with probability $1-qC\exp(-c\log p)$. 
\end{proof}

\subsection{Details of the algorithm \label{Comp_details}}

\begin{algorithm}[H]
\caption{The FISTA with Backtracking Line Search}
\label{alg_FISTA}
\DontPrintSemicolon
\SetKwInput{KwData}{Data}
\SetKwInput{KwResult}{Output}
\SetKwInput{KwIn}{Input}
\SetKwInput{KwInit}{Initialize}

\KwData{$\mathbf{S}_{xx}$, $\widehat{\mathbf{S}}_{xy}$ \tcp*[r]{Unbiased surrogates}}
\KwIn{$\lambda_{\mathbf{B}} \geq 0$, $0 < \eta < 1$, $t_{\text{init}} > 0$}
\KwInit{$\widehat{\mathbf{B}}_0 = \widehat{\mathbf{B}}_1 = \widehat{\mathbf{B}}^{(1)}$, $\mathbf{\Theta}_{\varepsilon\varepsilon} = \widehat{\mathbf{\Theta}}_{\varepsilon\varepsilon}$, $k = 2$}
\KwResult{$\widehat{\mathbf{B}}^{(2)}$}

\Begin{
  \Repeat{stop condition is met}{
    $t_k = t_{\text{init}}$ \tcp*[r]{Initialize the step size}
    \tcc{Backtracking line search}
    \Repeat{the Armijo condition (\ref{Armijo}) is met}{
      $t_k \leftarrow \eta t_k$ \tcp*[r]{Reduce step size} 
    }
    \tcc{FISTA}
    \textbf{update} $\widehat{\mathbf{B}}_k$ via Equation (\ref{FISTA})\;
    $k \leftarrow k + 1$\;  
  }
  \Return $\widehat{\mathbf{B}}^{(2)} = \widehat{\mathbf{B}}_k$\;
}
\end{algorithm}

\subsection{Computational details of the \texttt{missoNet} package \label{missonet_comp}}
This section outlines the main computational aspects of the R package \texttt{missoNet}, available at \url{https://CRAN.R-project.org/package=missoNet}. The package provides core functions for data simulation, model fitting and selection, result visualization, and making predictions on new data. It is designed with function arguments similar to the \texttt{glmnet} package, offering a familiar interface for existing \texttt{glmnet} users.

\subsubsection{Tuning parameter generation}

Our method uses two tuning parameters, $\lambda_{\mathbf{B}}$ and $\lambda_{\mathbf{\Theta}}$, to control the penalties on the coefficient matrix $\bB^{*}$ and the precision matrix $\bTheta_{\varepsilon\varepsilon}^{*}$. These parameters help balance model fit and complexity. Typically, values of $\lambda_{\mathbf{B}}$ and $\lambda_{\mathbf{\Theta}}$ are chosen by exploring a range that minimizes the error loss function.

A notable feature of \texttt{missoNet} is the automated generation of sequences for $\lambda_{\mathbf{B}}$ and $\lambda_{\mathbf{\Theta}}$. The sequence begins with maximum values, $\lambda_{\mathbf{B}}^{\text{max}}$ and $\lambda_{\mathbf{\Theta}}^{\text{max}}$, large enough to ensure that the initial model is essentially null (i.e., $\widehat{\mathbf{B}} = \mathbf{0}$ and $\widehat{\mathbf{\Theta}}_{\varepsilon\varepsilon} = \mathbf{I}$), providing a baseline for increasing model complexity. Following \cite{friedman2008sparse}, $\lambda_{\mathbf{\Theta}}^{\text{max}}$ is set to match the largest absolute value of the off-diagonal elements of the empirical covariance matrix $\widetilde{\mathbf{S}}_{\varepsilon\varepsilon}$, assuming $\widehat{\mathbf{B}} = \mathbf{0}$, is strong enough to cancel all empirical correlations:
\[
\lambda_{\mathbf{\Theta}}^{\text{max}} = \underset{k \neq k'}{\text{max}}|\tilde{s}_{kk'}|,
\]
where $\tilde{s}_{kk'}$ represents the $(k,k')$th element of $\widetilde{\mathbf{S}}_{\varepsilon\varepsilon}$. For $\lambda_{\mathbf{B}}$, the value required to yield $\widehat{\mathbf{B}} = \mathbf{0}$ is:
\[
\lambda_{\mathbf{B}}^{\text{max}} = \|2\widehat{\mathbf{S}}_{xy}\|_{\text{max}},
\]
assuming $\widehat{\mathbf{\Theta}}_{\varepsilon\varepsilon} = \mathbf{I}$, as described in \cite{friedman2010regularization} for \texttt{glmnet}. Both $\lambda_{\mathbf{\Theta}}^{\text{max}}$ and $\lambda_{\mathbf{B}}^{\text{max}}$ depend on surrogate estimators that account for missing data and are sensitive to standardization.

The sequence for each parameter ranges from $\lambda^{\text{max}}$ to a small fraction of $\lambda^{\text{max}}$, following a logarithmic scale. This setup allows thorough exploration from highly regularized models to more flexible ones. Users can adjust the sequence by specifying the number of values (\texttt{`nlambda`}), the minimum ratio (\texttt{`lambda.min.ratio`}), or provide a custom sequence if needed.

\subsubsection{Cross-validation}

To select the best tuning parameters, cross-validation divides the data into subsets for training and validation, assessing the model's predictive performance iteratively. Typically, $K$-fold cross-validation is employed, where the dataset is split into $K$ subsets. For each subset $k$, let ($\mathbf{X}_{k}$, $\mathbf{Y}_{k}$) represent the validation set, and ($\mathbf{X}_{-k}$, $\mathbf{Y}_{-k}$) represent the training set. The prediction error for fold $k$ is calculated as:
\[
\text{err}_{k}(\lambda_{\mathbf{B}}, \lambda_{\mathbf{\Theta}}) = \frac{1}{n_k}\VERT \mathbf{Y}_{k} - \mathbf{X}_{k} \widehat{\mathbf{B}}_{k}(\lambda_{\mathbf{B}}, \lambda_{\mathbf{\Theta}})\VERT^2_F,
\]
where $n_k$ is the number of observations in the $k$th fold, and $\widehat{\mathbf{B}}_{k}(\lambda_{\mathbf{B}}, \lambda_{\mathbf{\Theta}})$ is the coefficient estimate trained on ($\mathbf{X}_{-k}$, $\mathbf{Y}_{-k}$) with the tuning parameters $\lambda_{\mathbf{B}}$ and $\lambda_{\mathbf{\Theta}}$. The optimal combination of $\lambda_{\mathbf{B}}$ and $\lambda_{\mathbf{\Theta}}$ minimizes the average cross-validation error:
\[
(\widehat{\lambda}_{\mathbf{B}}, \widehat{\lambda}_{\mathbf{\Theta}}) = \argmin_{\lambda_{\mathbf{B}}, \lambda_{\mathbf{\Theta}} \geq 0}\ \frac{1}{K} \sum_{k=1}^{K} \text{err}_{k}(\lambda_{\mathbf{B}}, \lambda_{\mathbf{\Theta}}).
\]

For datasets with missing values, where elements of $\mathbf{Y}$ are partially observable, standard cross-validation can be biased if only complete cases are used. We adapt by using surrogate estimators to account for missing data:
\[
(\widehat{\lambda}_{\mathbf{B}}, \widehat{\lambda}_{\mathbf{\Theta}}) = \argmin_{\lambda_{\mathbf{B}}, \lambda_{\mathbf{\Theta}} \geq 0}\ \frac{1}{K} \sum_{k=1}^{K} \text{Tr}[\widehat{\mathbf{S}}^{k}_{yy} - 2{\widehat{\mathbf{S}}^{k}_{xy}}{}^{\top} \widehat{\mathbf{B}}_{k}(\lambda_{\mathbf{B}}, \lambda_{\mathbf{\Theta}}) + \widehat{\mathbf{B}}_{k}(\lambda_{\mathbf{B}}, \lambda_{\mathbf{\Theta}})^{\top} \mathbf{S}^{k}_{xx} {\widehat{\mathbf{B}}_{k}(\lambda_{\mathbf{B}}, \lambda_{\mathbf{\Theta}})}].
\]
This approach, called calibrated cross-validation, ensures accurate performance metrics despite incomplete data.

\subsubsection{Warm starts}

Similar to \texttt{glmnet}, \texttt{missoNet} uses warm starts to improve computational efficiency. The model starts with the most restrictive settings ($\lambda_{\mathbf{B}} = \lambda_{\mathbf{B}}^{\text{max}}$, $\lambda_{\mathbf{\Theta}} = \lambda_{\mathbf{\Theta}}^{\text{max}}$) and gradually reduces these values, using each previous solution as a starting point for the next iteration, as shown in Figure \ref{fig:BICexample} (a).

\begin{figure}[htbp]
    \centering
    \includegraphics[width=0.9\textwidth,origin=c]{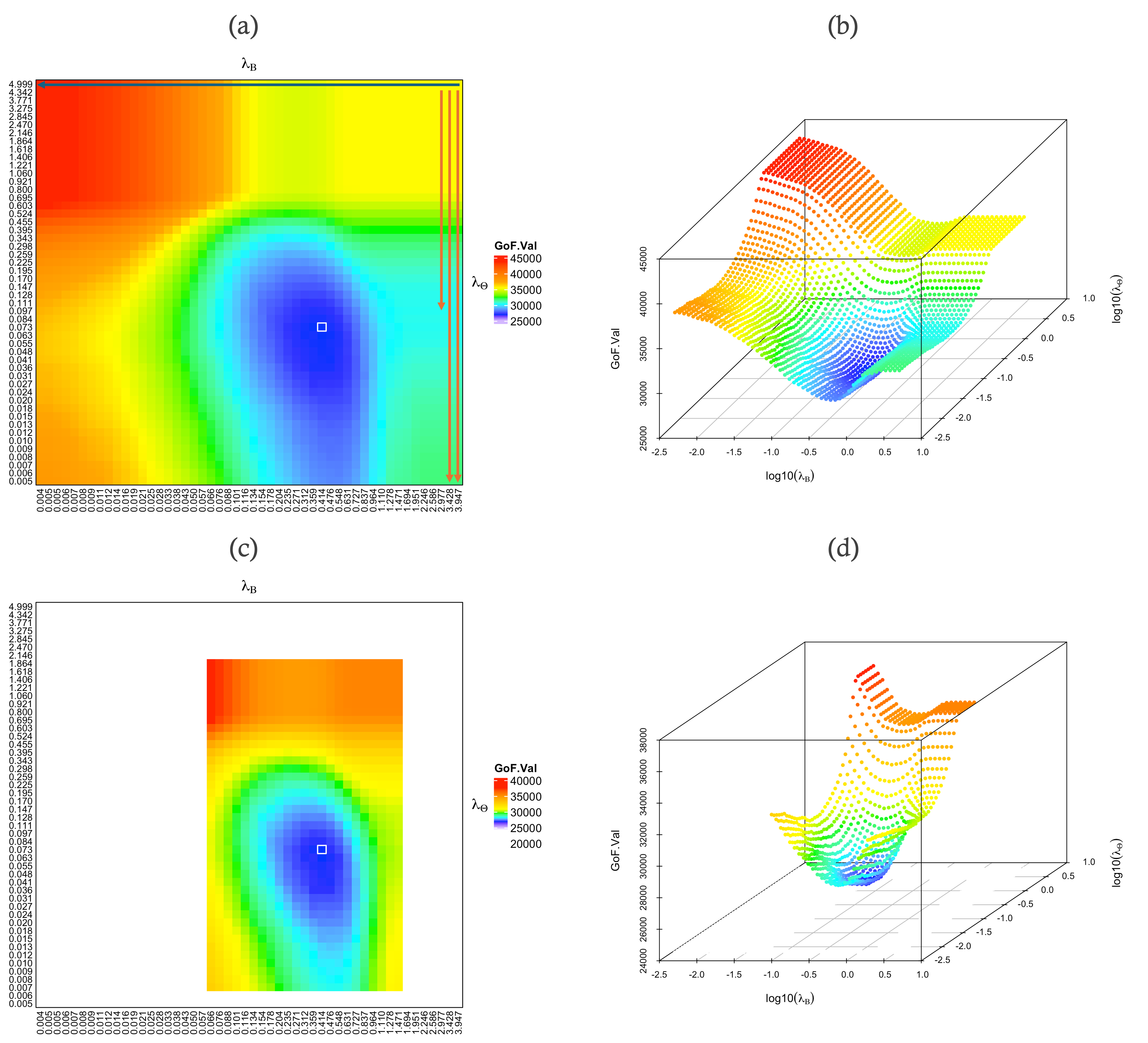}
    \caption{Visualization of a grid search over $\lambda_{\mathbf{\Theta}}$ and $\lambda_{\mathbf{B}}$, with model selection based on the BIC for goodness-of-fit. The plots are generated using the \texttt{missoNet} package's plotting function.
    (\textit{a}) Heatmap of BIC values for different combinations of $\lambda_{\mathbf{\Theta}}$ and $\lambda_{\mathbf{B}}$. The optimal solution, indicated by the lowest BIC, is marked by a white square. The dark teal arrow shows the algorithm's initial stage to estimate $\widehat{\mathbf{B}}^{(1)}$, starting with the estimates $\widehat{\mathbf{B}} = \mathbf{0}$ and $\widehat{\mathbf{\Theta}}_{\varepsilon\varepsilon} = \mathbf{I}$ from the most restrictive setting ($\lambda_{\mathbf{B}}^{\text{max}}, \lambda_{\mathbf{\Theta}}^{\text{max}}$). In this stage, $\lambda_{\mathbf{\Theta}}$ is fixed at $\lambda_{\mathbf{\Theta}}^{\text{max}}$, and $\lambda_{\mathbf{B}}$ is gradually reduced, using previous solutions as starting points. The orange lines represent multiple independent paths, each beginning with $\widehat{\mathbf{B}} = \widehat{\mathbf{B}}^{(1)}$ and $\widehat{\mathbf{\Theta}}_{\varepsilon\varepsilon} = \mathbf{I}$. These paths progressively reduce $\lambda_{\mathbf{\Theta}}$ while keeping $\lambda_{\mathbf{B}}$ fixed, utilizing warm starts and enabling parallel computation.
    (\textit{b}) 3D scatterplot of the BIC values.
    (\textit{c}) Heatmap of BIC values with local search enabled.
    (\textit{d}) 3D scatterplot of BIC values with local search enabled.}
    \label{fig:BICexample}
\end{figure}

\subsubsection{Local search}

To enhance computational efficiency, \texttt{missoNet} also includes a local search mechanism. After an initial broad search (in a coarse grid or randomly) to identify promising parameter values, the algorithm refines the search within that region for greater precision. This strategy strikes a balance between efficiency and accuracy, allowing for more targeted exploration. Users can enable local search by setting \texttt{`fast = TRUE`}. Figure \ref{fig:BICexample} (c) illustrates this targeted search approach.

\subsubsection{Parallelization}

\texttt{missoNet} supports intelligent parallelization to speed up computations, without requiring manual setup. When tuning parameters through cross-validation, the $K$ tasks for $K$ folds are executed simultaneously using available clusters or CPU cores. Alternatively, when tuning with information criteria such as BIC, each regularization path is distributed across multiple cores. Users can easily activate this feature by setting \texttt{`parallel = TRUE`}. Figure \ref{fig:BICexample} (a) highlights the parallelizable paths within the tuning process.

\subsection{Additional simulation results from Section \ref{sim_theory_verify} \label{q10Sim}}

We present some of the results of the simulation settings performed in Section \ref{sim_theory_verify} to verify the theoretical bounds in this part of the Appendix. 

     \begin{figure}[!h]
	\centering
        \includegraphics[scale=0.26]{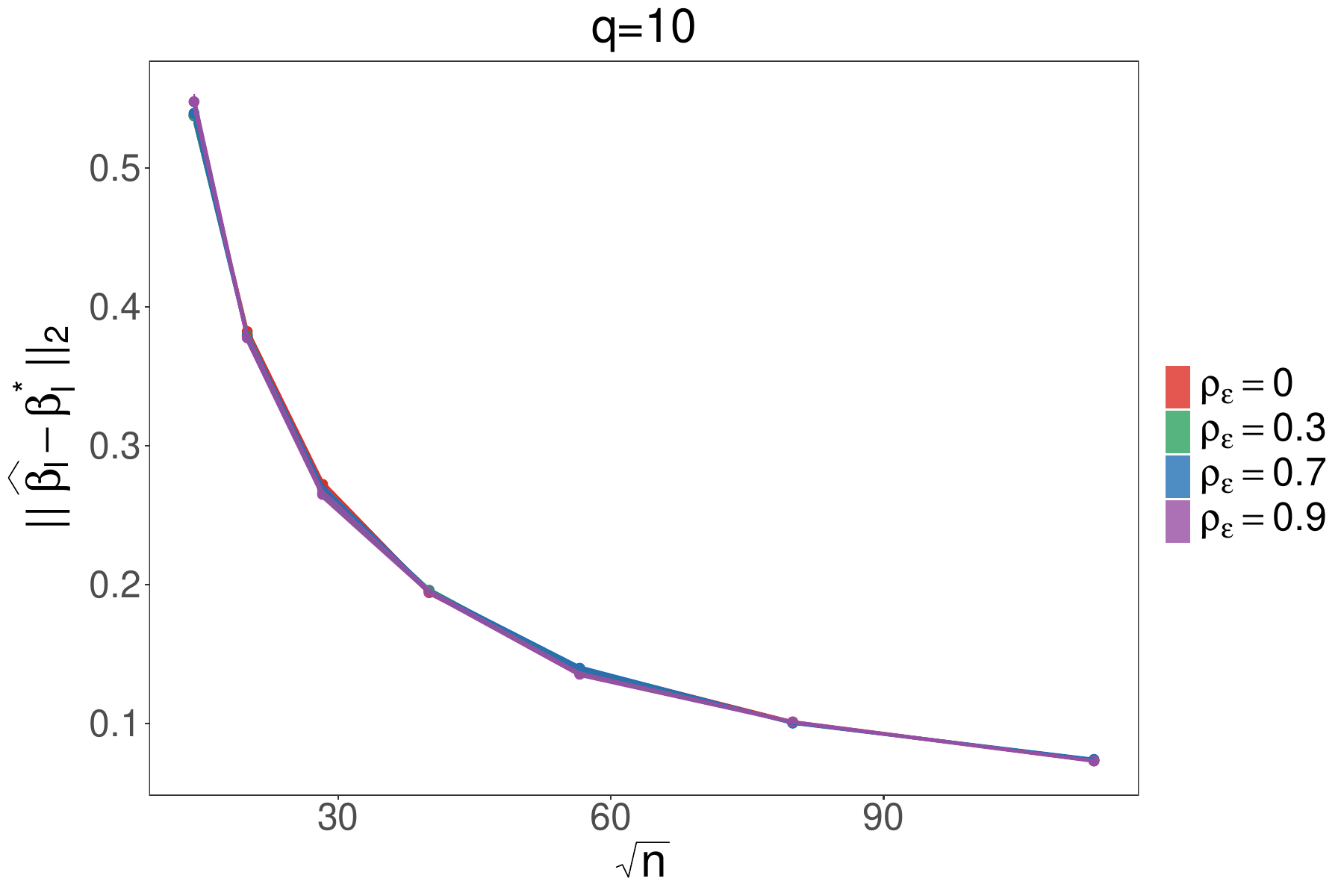}
        \includegraphics[scale=0.26]{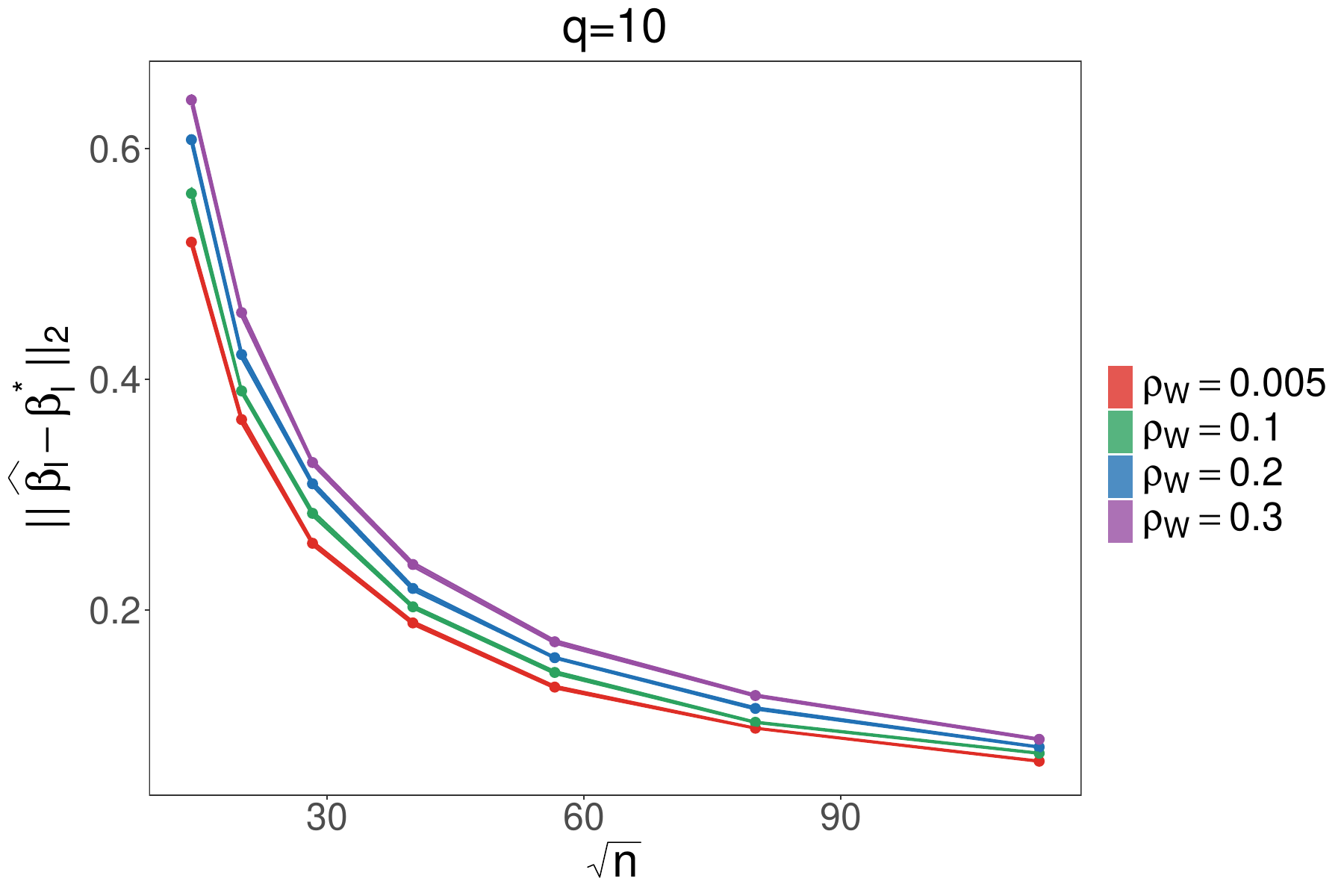}
	\caption{Scenario S1A. Plots of the average estimation error $\|\widehat{{\bbeta}_{l}}-{\bbeta}_{l}^{*}\|_{2}$  against $\sqrt{n}$, $n=200,400,800,1600,3200,6400,12800$ for the outcome dimension, $q=10$,  with the varying correlation between the errors, $\rho_{\varepsilon}$ (left panel) and the probability of being missing in the $j$th column for the
outcome $\rho_{W,j}$ (right panel).  Each point represents an average of 100 trials and the shaded regions indicate standard error for each method.} \label{fig:ERS1P1b}
\end{figure}    

\begin{figure}[!h]
	\centering
        \includegraphics[scale=0.26]{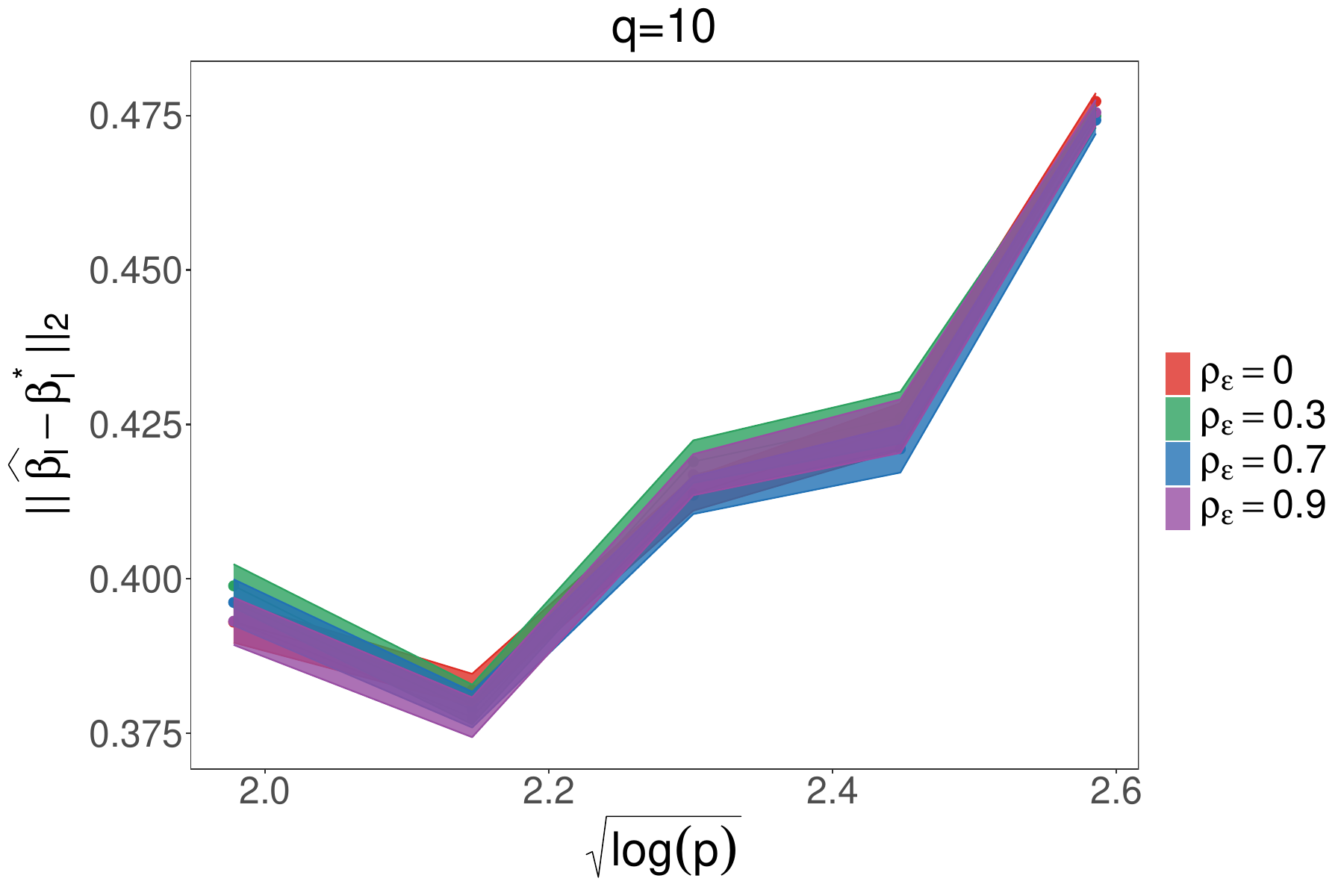}
        \includegraphics[scale=0.26]{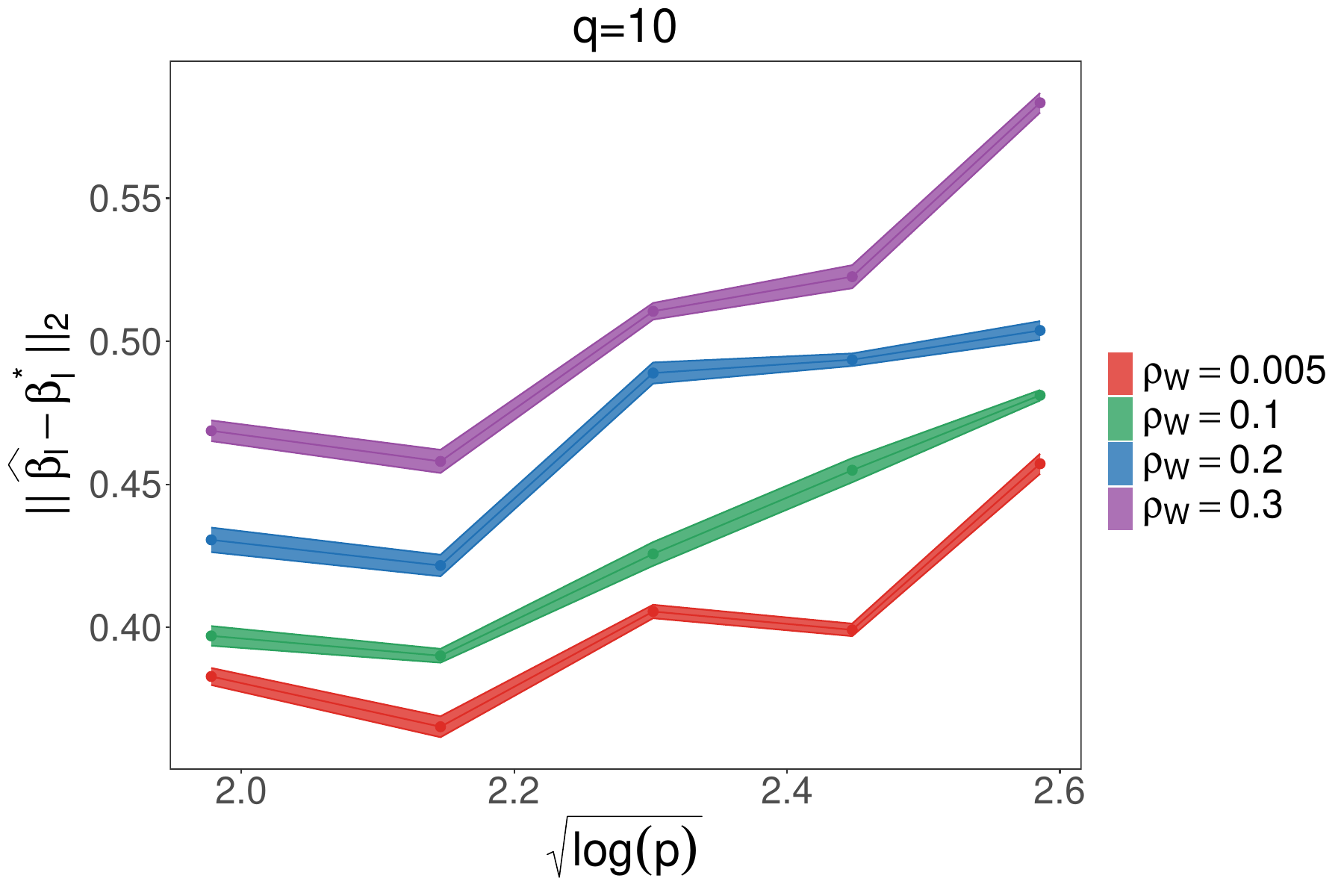}
	\caption{Scenario S1B. Plots of the average estimation error $\|\widehat{{\bbeta}_{l}}-{\bbeta}_{l}^{*}\|_{2}$ against $\sqrt{\log(p)}$, $p=50,100,200,400$ and $800$, for the outcome dimension, $q=10$,  with varying correlation between the errors, $\rho_{\varepsilon}$ (left panel) and the probability of being missing in each column of the outcomes, $\rho_{W,j}$ (right panel). Each point represents an average of 100 trials and the shaded regions indicate standard error for each method.} \label{fig:ERS1P2b}
\end{figure}   

     \begin{figure}[ht]
	\centering
        \includegraphics[scale=0.26]{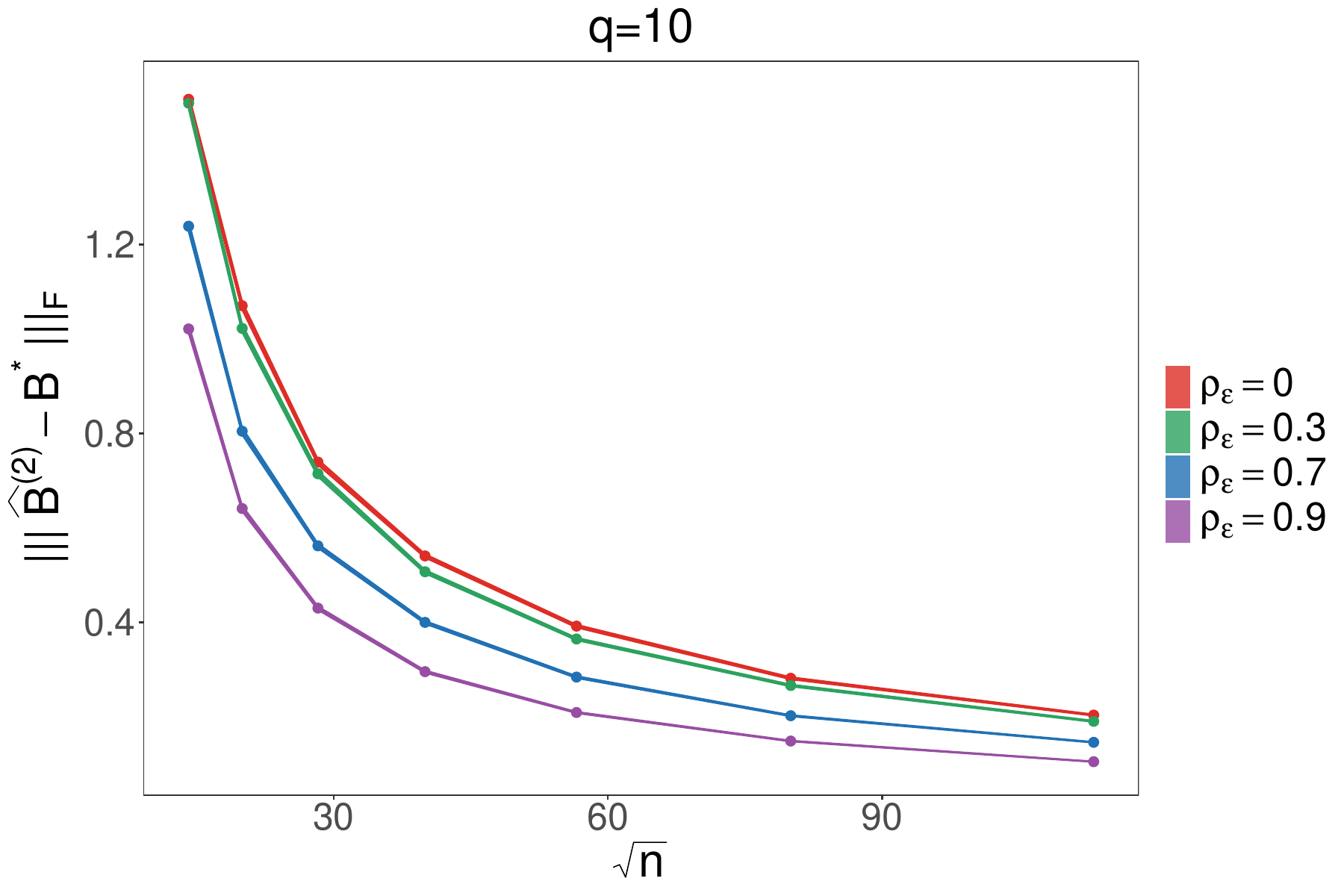}
        \includegraphics[scale=0.26]{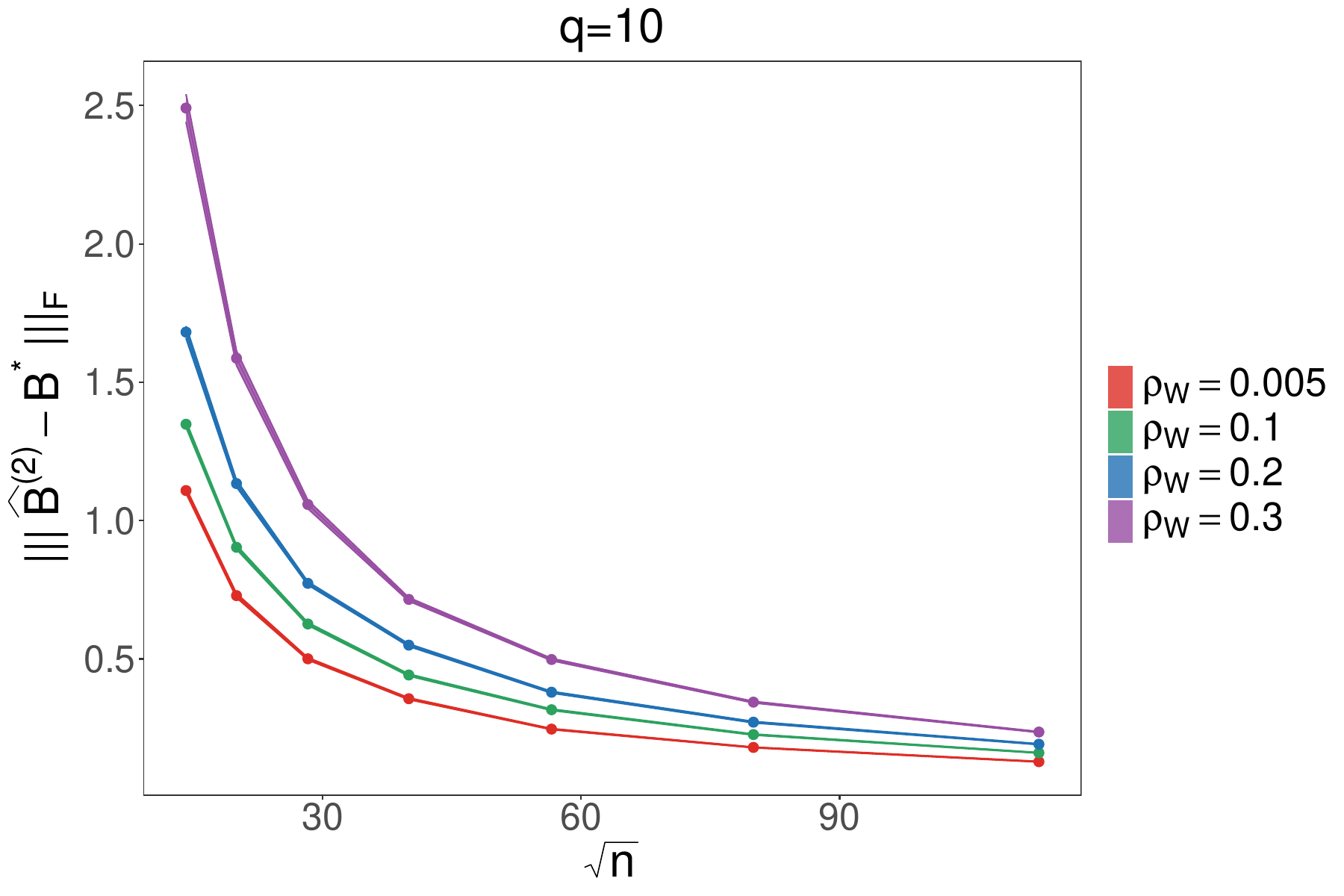}
	\caption{Scenario S3A. Plots of the average estimation error for $\bB^{*}$, $\VERT\widehat{\bB}^{(2)}-\bB^{*}\VERT_{F}$  against $\sqrt{n}$, $n=200,400,800,1600,3200,6400,12800$, for the outcome dimension, $q = 10$, with the varying correlation between the errors, $\rho_{\varepsilon}$ (left panel) and the probability of being missing in the $j$th column for the
outcome $\rho_{W,j}$ (right panel).  Each point represents an average of 100 trials and the shaded regions indicate standard error for each method.} \label{fig:ERS1P5b}
\end{figure} 

\begin{figure}[ht]
	\centering
        \includegraphics[scale=0.26]{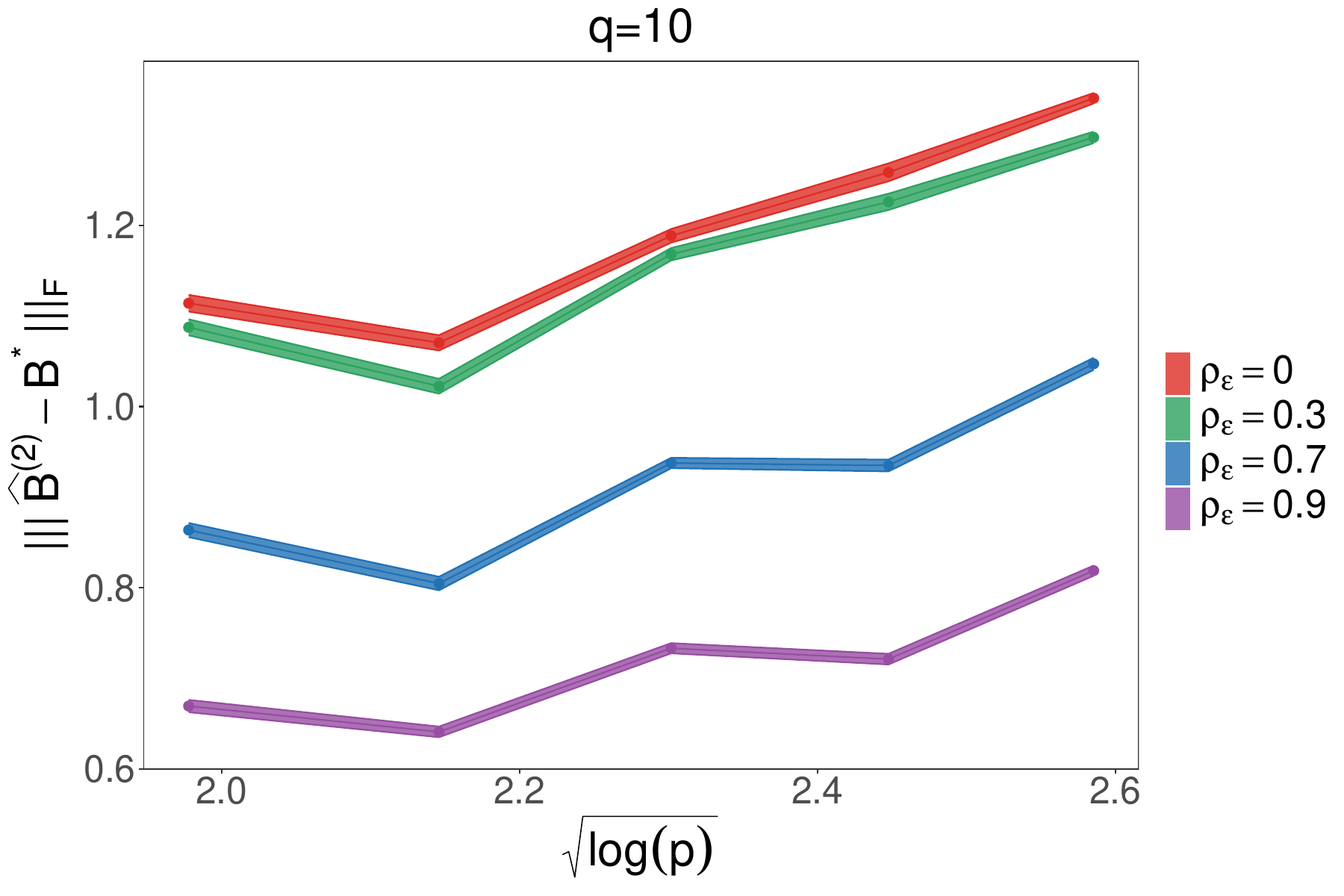}
        \includegraphics[scale=0.26]{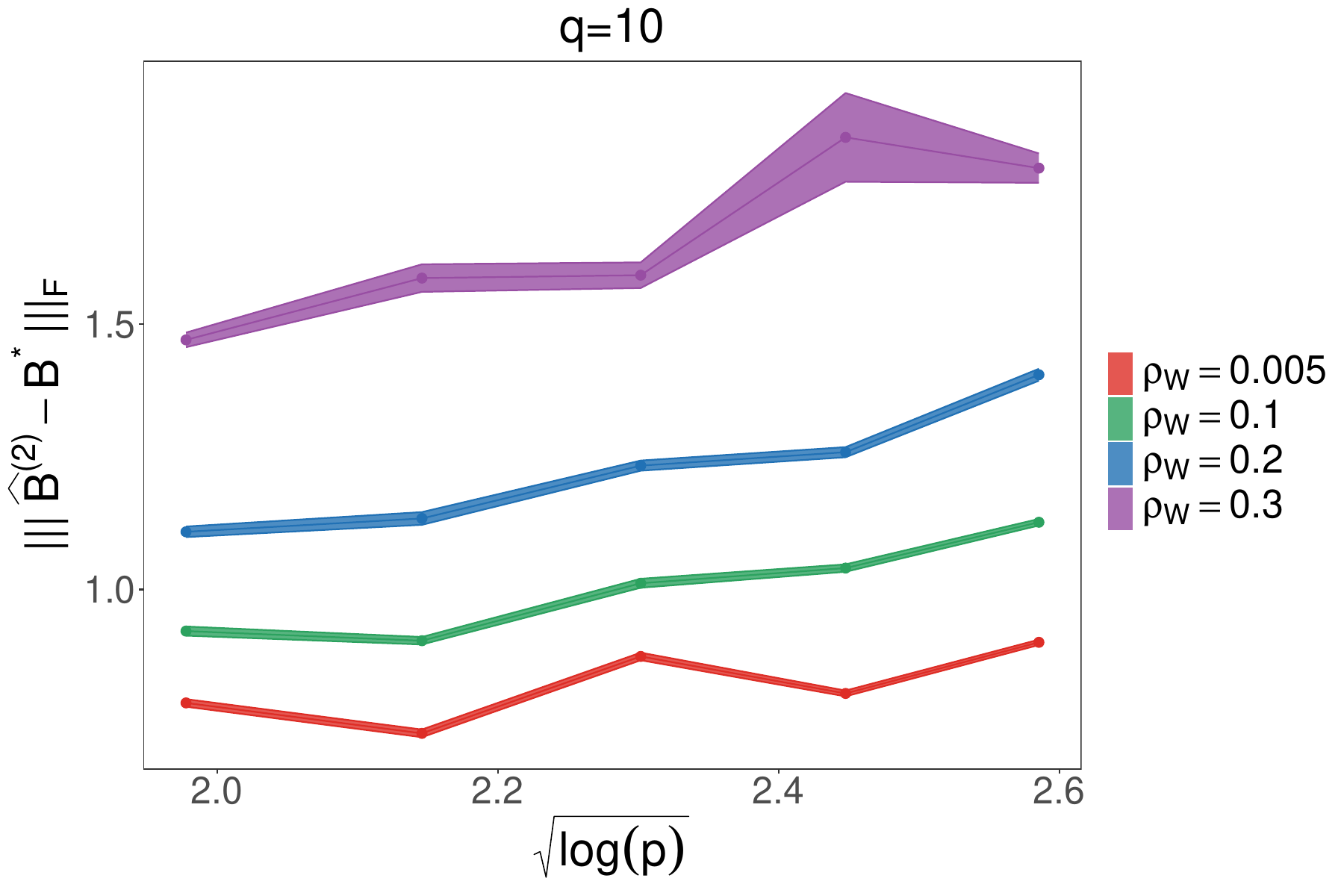}
	\caption{Scenario S3B. Plots of the average estimation error for $\bB^{*}$, $\VERT\widehat{\bB}^{(2)}-\bB^{*}\VERT_{F}$ against $\sqrt{\log(p)}$, $p=50,100,200,400$ and $800$, for the outcome dimension, $q = 10$, with varying correlation between the errors, $\rho_{\varepsilon}$ (left panel) and the probability of being missing in each column of the outcomes, $\rho_{W,j}$ (right panel). Each point represents an average of 100 trials and the shaded regions indicate standard error for each method.} \label{fig:ERS1P6b}
\end{figure}   
\begin{table}[ht]
    \centering
    \begin{minipage}{0.48\linewidth}
        \centering
        \renewcommand{\arraystretch}{1.0}
        \caption{Frobenius norm error for setups with varying $n$ and $\rho_{\varepsilon}$ with $p=100$, $q=10$, $\rho_{W,j} = 0.05$, and $s_{\max}=5$. Averages and standard errors in parentheses are based on 100 replications.}
        \begin{tabular}{cccc}
            \hline
            $n$ & $\rho_{\varepsilon}$ & $\VERT \widehat{\bB}^{(1)} - \bB^{*} \VERT_{F}$ & $\VERT \widehat{\bB}^{(2)} - \bB^{*} \VERT_{F}$ \\ \hline
            200   & 0   & 1.736 (0.015) & 1.508 (0.011) \\
                  & 0.3 & 1.736 (0.015) & 1.499 (0.012) \\
                  & 0.7 & 1.741 (0.016) & 1.239 (0.011) \\
                  & 0.9 & 1.767 (0.017) & 1.021 (0.009) \\ \hline
            400   & 0   & 1.230 (0.008) & 1.070 (0.008) \\
                  & 0.3 & 1.222 (0.010) & 1.022 (0.008) \\
                  & 0.7 & 1.219 (0.009) & 0.805 (0.007) \\
                  & 0.9 & 1.215 (0.010) & 0.641 (0.006) \\ \hline
            800   & 0   & 0.875 (0.007) & 0.739 (0.006) \\
                  & 0.3 & 0.854 (0.007) & 0.715 (0.006) \\
                  & 0.7 & 0.865 (0.008) & 0.562 (0.004) \\
                  & 0.9 & 0.852 (0.008) & 0.430 (0.004) \\ \hline
            1600  & 0   & 0.623 (0.004) & 0.541 (0.004) \\
                  & 0.3 & 0.629 (0.005) & 0.507 (0.004) \\
                  & 0.7 & 0.626 (0.004) & 0.400 (0.003) \\
                  & 0.9 & 0.625 (0.005) & 0.295 (0.003) \\ \hline
            3200  & 0   & 0.447 (0.004) & 0.392 (0.003) \\
                  & 0.3 & 0.441 (0.004) & 0.365 (0.003) \\
                  & 0.7 & 0.448 (0.005) & 0.284 (0.002) \\
                  & 0.9 & 0.434 (0.004) & 0.209 (0.001) \\ \hline
            6400  & 0   & 0.323 (0.002) & 0.281 (0.002) \\
                  & 0.3 & 0.322 (0.002) & 0.266 (0.002) \\
                  & 0.7 & 0.321 (0.002) & 0.202 (0.001) \\
                  & 0.9 & 0.323 (0.003) & 0.148 (0.001) \\ \hline
            12800 & 0   & 0.234 (0.002) & 0.203 (0.001) \\
                  & 0.3 & 0.235 (0.002) & 0.190 (0.001) \\
                  & 0.7 & 0.236 (0.002) & 0.146 (0.001) \\
                  & 0.9 & 0.233 (0.002) & 0.105 (0.001) \\ \hline
        \end{tabular}
        \label{tab:q10_rho_eps}
    \end{minipage}%
    \hfill
    \begin{minipage}{0.48\linewidth}
        \centering
        \renewcommand{\arraystretch}{1.0}
        \caption{Frobenius norm error for setups with varying $n$ and $\rho_{W,j}$ with $p=100$, $q=10$, $\rho_{\varepsilon} = 0.7$, and $s_{\max}=5$. Averages and standard errors in parentheses are based on 100 replications.}
        \begin{tabular}{cccc}
            \hline
            $n$ & $\rho_{W,j}$ & $\VERT \widehat{\bB}^{(1)} - \bB^{*}\VERT_{F}$ & $\VERT \widehat{\bB}^{(2)} - \bB^{*}\VERT_{F}$ \\ \hline
            200   & 0.005 & 1.673 (0.014) & 1.109 (0.009) \\
                  & 0.1   & 1.813 (0.018) & 1.348 (0.011) \\
                  & 0.2   & 1.967 (0.015) & 1.681 (0.022) \\
                  & 0.3   & 2.084 (0.016) & 2.491 (0.050) \\ \hline
            400   & 0.005 & 1.174 (0.011) & 0.729 (0.007) \\
                  & 0.1   & 1.258 (0.008) & 0.904 (0.007) \\
                  & 0.2   & 1.363 (0.012) & 1.134 (0.012) \\
                  & 0.3   & 1.484 (0.013) & 1.587 (0.026) \\ \hline
            800   & 0.005 & 0.829 (0.007) & 0.501 (0.004) \\
                  & 0.1   & 0.915 (0.008) & 0.627 (0.005) \\
                  & 0.2   & 0.999 (0.008) & 0.773 (0.006) \\
                  & 0.3   & 1.062 (0.007) & 1.058 (0.015) \\ \hline
            1600  & 0.005 & 0.606 (0.005) & 0.356 (0.003) \\
                  & 0.1   & 0.652 (0.005) & 0.442 (0.003) \\
                  & 0.2   & 0.706 (0.007) & 0.550 (0.004) \\
                  & 0.3   & 0.775 (0.007) & 0.715 (0.006) \\ \hline
            3200  & 0.005 & 0.427 (0.004) & 0.247 (0.002) \\
                  & 0.1   & 0.468 (0.005) & 0.317 (0.002) \\
                  & 0.2   & 0.510 (0.003) & 0.380 (0.003) \\
                  & 0.3   & 0.557 (0.005) & 0.498 (0.004) \\ \hline
            6400  & 0.005 & 0.312 (0.002) & 0.180 (0.001) \\
                  & 0.1   & 0.330 (0.002) & 0.227 (0.002) \\
                  & 0.2   & 0.369 (0.004) & 0.272 (0.002) \\
                  & 0.3   & 0.406 (0.003) & 0.344 (0.003) \\ \hline
            12800 & 0.005 & 0.220 (0.002) & 0.129 (0.001) \\
                  & 0.1   & 0.243 (0.002) & 0.161 (0.001) \\
                  & 0.2   & 0.261 (0.002) & 0.192 (0.002) \\
                  & 0.3   & 0.284 (0.003) & 0.236 (0.002) \\ \hline
        \end{tabular}
        \label{tab:q10_rhoWj}
    \end{minipage}
\end{table}

\begin{table}[ht]
    \centering
    \begin{minipage}{0.48\linewidth}
        \centering
        \renewcommand{\arraystretch}{1.0}
        \caption{Frobenius norm error for setups with varying $p$ and $\rho_{\varepsilon}$ with $n=400$, $q=10,$ $\rho_{W,j} = 0.05$ and $s_{\max}=5$. Averages and standard errors in parentheses are based on 100 replications.}
        \begin{tabular}{cccc}
            \hline
            $p$ & $\rho_{\varepsilon}$ & $\VERT \widehat{\bB}^{(1)} - \bB^{*}\VERT_{F}$ & $\VERT \widehat{\bB}^{(2)} - \bB^{*}\VERT_{F}$ \\ \hline
        50  & 0   & 1.295 (0.011) & 1.114 (0.009) \\
            & 0.3 & 1.315 (0.011) & 1.088 (0.009) \\
            & 0.7 & 1.307 (0.012) & 0.864 (0.008) \\
            & 0.9 & 1.298 (0.012) & 0.669 (0.007) \\ \hline
        100 & 0   & 1.230 (0.008) & 1.070 (0.008) \\
            & 0.3 & 1.222 (0.010) & 1.022 (0.008) \\
            & 0.7 & 1.219 (0.009) & 0.805 (0.007) \\
            & 0.9 & 1.215 (0.010) & 0.641 (0.006) \\ \hline
        200 & 0   & 1.333 (0.009) & 1.189 (0.007) \\
            & 0.3 & 1.350 (0.011) & 1.168 (0.007) \\
            & 0.7 & 1.333 (0.010) & 0.938 (0.006) \\
            & 0.9 & 1.344 (0.011) & 0.733 (0.006) \\ \hline
        400 & 0   & 1.362 (0.012) & 1.259 (0.010) \\
            & 0.3 & 1.367 (0.013) & 1.226 (0.009) \\
            & 0.7 & 1.351 (0.012) & 0.935 (0.006) \\
            & 0.9 & 1.362 (0.014) & 0.721 (0.006) \\ \hline
        800 & 0   & 1.530 (0.004) & 1.340 (0.006) \\
            & 0.3 & 1.523 (0.006) & 1.297 (0.006) \\
            & 0.7 & 1.521 (0.007) & 1.047 (0.007) \\
            & 0.9 & 1.526 (0.007) & 0.819 (0.006) \\ \hline
        \end{tabular}
        \label{tab:q10_rho_eps_1}
    \end{minipage}%
    \hfill
    \begin{minipage}{0.48\linewidth}
        \centering
        \renewcommand{\arraystretch}{1.0}
        \caption{Frobenius norm error for setups with varying $p$ and $\rho_{W,j}$ with $n=400$, $q=10$, $\rho_{\varepsilon} = 0.7$ and $s_{\max}=5$. Averages and standard errors in parentheses are based on 100 replications.}
        \begin{tabular}{cccc}
            \hline
            $p$ & $\rho_{W,j}$ & $\VERT \widehat{\bB}^{(1)} - \bB^{*} \VERT_{F}$ & $\VERT \widehat{\bB}^{(2)} - \bB^{*} \VERT_{F}$ \\ \hline
        50  & 0.005 & 1.263 (0.010) & 0.786 (0.007) \\
            & 0.1   & 1.310 (0.011) & 0.922 (0.009) \\
            & 0.2   & 1.418 (0.014) & 1.109 (0.010) \\
            & 0.3   & 1.543 (0.012) & 1.471 (0.014) \\ \hline
        100 & 0.005 & 1.174 (0.011) & 0.729 (0.007) \\
            & 0.1   & 1.258 (0.008) & 0.904 (0.007) \\
            & 0.2   & 1.363 (0.012) & 1.134 (0.012) \\
            & 0.3   & 1.484 (0.013) & 1.587 (0.026) \\ \hline
        200 & 0.005 & 1.306 (0.007) & 0.874 (0.007) \\
            & 0.1   & 1.373 (0.013) & 1.012 (0.008) \\
            & 0.2   & 1.580 (0.012) & 1.234 (0.010) \\
            & 0.3   & 1.650 (0.009) & 1.592 (0.024) \\ \hline
        400 & 0.005 & 1.279 (0.007) & 0.804 (0.005) \\
            & 0.1   & 1.459 (0.014) & 1.040 (0.007) \\
            & 0.2   & 1.586 (0.007) & 1.259 (0.010) \\
            & 0.3   & 1.683 (0.013) & 1.852 (0.084) \\ \hline
        800 & 0.005 & 1.466 (0.011) & 0.901 (0.005) \\
            & 0.1   & 1.545 (0.006) & 1.127 (0.006) \\
            & 0.2   & 1.621 (0.010) & 1.405 (0.011) \\
            & 0.3   & 1.877 (0.011) & 1.794 (0.028) \\ \hline
        \end{tabular}
        \label{tab:q10_rhoWj_1}
    \end{minipage}
\end{table}

\subsection{Details of the simulations conducted for method comparison in Section \ref{missVother}} \label{Sim:Comparison}

In this section, we present the structure of the precision matrices considered, the evaluation metrics and the tuning parameter selection process for the simulations conducted to compare missoNet with other existing methods described in Section \ref{missVother}. Additionally, we include tables summarizing the simulation results to illustrate the comparative performance.

\subsubsection{Structure of $\mathbf{\Theta}_{\varepsilon\varepsilon}^{*}$}
\label{precison_struc}

We considered two types of network structures for the true precision matrix $\mathbf{\Theta}_{\varepsilon\varepsilon}^{*}$.

\subsubsection*{Type 1: Inverse AR(1) matrix}
The first type treats $\mathbf{\Theta}_{\varepsilon\varepsilon}^{*}$ as the inverse of an autoregressive order one (AR(1)) covariance matrix, following the approach by \citet{rothman2010sparse} and \citet{yin2011sparse}. For $q$ responses, the AR(1) covariance matrix is given by:
\[
\bSigma_{\varepsilon\varepsilon}^{*} = \begin{bmatrix}
  1 & r & r^2 & \cdots & r^{q-1} \\
  r & 1 & r & \cdots & r^{q-2}\\
  r^2 & r & 1 & \cdots & r^{q-3}\\
  \vdots & \vdots & \vdots & \ddots & \vdots \\
  r^{q-1} & r^{q-2} & r^{q-3} & \cdots & 1 \\
\end{bmatrix},
\]
where $r$ is the autoregressive parameter between 0 and 1. The resulting precision matrix is sparse and tri-diagonal:
\[
(\bSigma_{\varepsilon\varepsilon}^{*})^{-1} = \frac{1}{1-r^2}\begin{bmatrix}
  1 & -r &  &  & \\
  -r & 1+r^2 & -r &  & \\
   & \ddots & \ddots & \ddots & \\
   &  & \ddots & 1+r^2 & -r \\
   &  &  & -r & 1 \\
\end{bmatrix},
\]
where unspecified elements are zero. This precision matrix represents a chain structure where each internal node is linked only to its two nearest neighbors. It is noteworthy that this structure supports well the regularization assumptions commonly used in multivariate approaches, allowing for both correlated errors and sparse conditional dependencies among the response variables.

\subsubsection*{Type 2: Composite structure}
The second type extends the structure of $\mathbf{\Theta}_{\varepsilon\varepsilon}^{*}$ by summing an inverse AR(1) matrix with a block-diagonal component, as discussed in \citet{fan2009network} and \citet{cai2011constrained}. The responses are divided into three groups of equal size ($q/3$ variables each). Each block represents one of three graph types:
\begin{itemize}
    \item \textbf{Independent}: A set of nodes without connecting edges, indicating complete independence among the variables.
    \item \textbf{Weak}: A complete graph with edge weights between 0.1 and 0.4, indicating weak interactions.
    \item \textbf{Strong}: A complete graph with edge weights between 0.5 and 1, indicating strong interactions.
\end{itemize}
To ensure positive definiteness, small positive values are added to the diagonal of $\mathbf{\Theta}_{\varepsilon\varepsilon}^{*}$. Compared to the linear scaling in the chain structure (Type 1), the composite structure (Type 2) has more nonzero parameters, scaling quadratically with $q$. Figure \ref{fig:ThetaStar_Struc} illustrates both structures.

\begin{figure}[htbp]
    \centering
    \includegraphics[width=0.93\textwidth,origin=c]{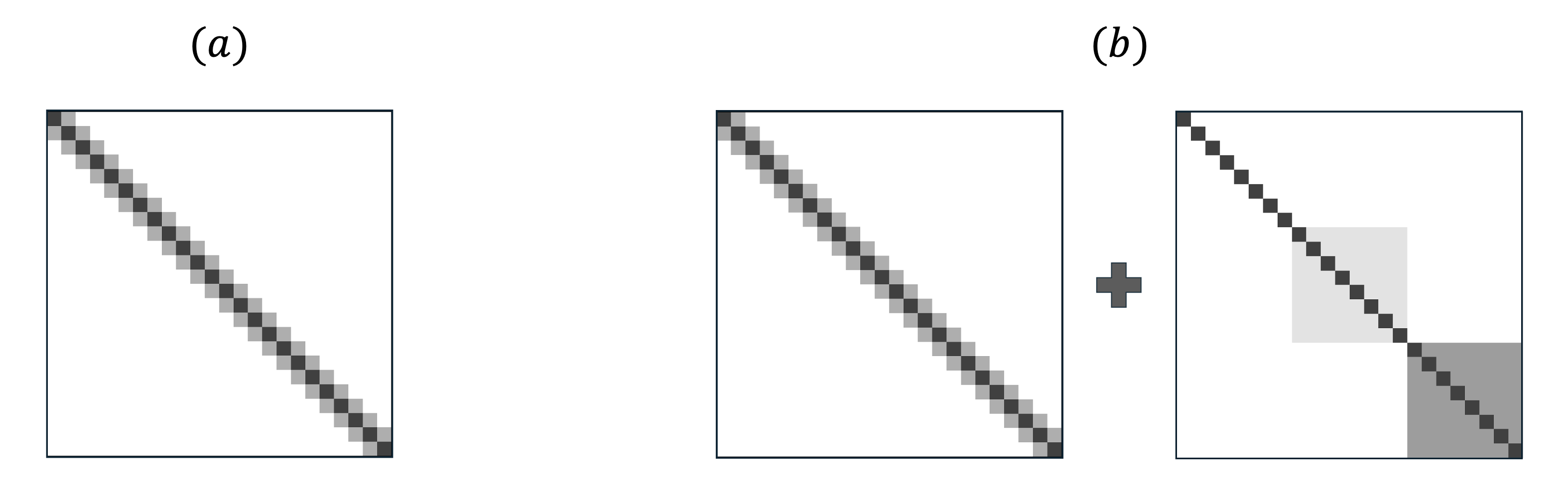}
    \caption{Structures of $\mathbf{\Theta}_{\varepsilon\varepsilon}^{*}$. Nonzero elements are shaded by relative strength. (a) Type 1: Inverse AR(1) structure. (b) Type 2: Composite structure.}
    \label{fig:ThetaStar_Struc}
\end{figure}

\subsubsection{Results from Section \ref{Comp_res} \label{Res_comp}}

\renewcommand{\arraystretch}{0.85}
\begin{xltabular}{\textwidth}{@{}>{\centering\arraybackslash}p{1.5cm} >{\arraybackslash}p{1.50cm} >{\centering\arraybackslash}p{1cm} >{\centering\arraybackslash}p{1cm} >{\centering\arraybackslash}p{1cm} >{\centering\arraybackslash}p{1cm} >{\centering\arraybackslash}p{1cm} >{\centering\arraybackslash}p{1cm} >{\centering\arraybackslash}p{1cm} >{\centering\arraybackslash}p{1cm} >{\centering\arraybackslash}p{0.9cm}@{}}
  \caption{Means and standard errors (in parentheses) of evaluation metrics for different methods under three levels of missing data, based on 100 simulation replications using Model 1 for data generation.}\\
  \toprule
  Model 1 & & \multicolumn{4}{c}{Estimation of $\mathbf{B}^{*}$} & \multicolumn{4}{c}{Estimation of $\mathbf{\Theta}_{\varepsilon\varepsilon}^{*}$} & \\
  \cmidrule(lr){3-6} \cmidrule(lr){7-10}
  Missing & Method & PE & TPR & TNR & MCC & KLL & TPR & TNR & MCC & Time \\
  \midrule
  \endfirsthead

  \multicolumn{11}{c}%
  {{Table \thetable\ Continued from Previous Page}} \\
  \toprule
  Model 1 & & \multicolumn{4}{c}{Estimation of $\mathbf{B}^*$} & \multicolumn{4}{c}{Estimation of $\mathbf{\Theta}_{\varepsilon\varepsilon}^*$} & \\
  \cmidrule(lr){3-6} \cmidrule(lr){7-10}
  Missing & Method & PE & TPR & TNR & MCC & KLL & TPR & TNR & MCC & Time \\
  \midrule
  \endhead

  \midrule \multicolumn{11}{r}{{Continued on next page}} \\
  \endfoot

  \endlastfoot

  1\% & \textit{Lasso} & 0.471 & 1 & 0.872 & 0.462 & & & & & 0.37 \\
  & [cv.min] & (0.007) & (0) & (0.003) & (0.004) & & & & & (0.00) \\
  \cline{2-6}
  & \textit{Lasso} & 1.693 & \color{Teal}0.995 & \color{Red}\textbf{0.990} & \color{Red}\textbf{0.900}\\
  & [cv.1se] & (0.020) & (0.001) & (0.000) & (0.004) \\
  \cline{2-11}
  & \textit{missoNet} & \color{Red}\textbf{0.187}\footnote{For each level of missingness, the best-performing method and tuning strategy for each metric (excluding the value 1 for simplicity) is highlighted in red. Results within $\pm$5\% of the best performance are shown in teal.} & 1 & 0.885 & 0.486 & 0.665 & 1 & \color{Teal}0.775 & \color{Teal}0.443 & 6.96 \\
  & [cv.min] & (0.004) & (0) & (0.003) & (0.006) & (0.038) & (0) & (0.014) & (0.014) & (0.04) \\
  \cline{2-10}
  & \textit{missoNet} & 0.680 & 1 & \color{Teal}0.988 & \color{Teal}0.870 & 0.732 & 1 & 0.758 & 0.416 &  \\
  & [cv.1se] & (0.020) & (0) & (0.001) & (0.005) & (0.036) & (0) & (0.016) & (0.015) \\
  \cline{2-11}
  & \textit{missoNet} & 0.351 & 1 & \color{Teal}0.971 & 0.757 & \color{Red}\textbf{0.533} & 1 & \color{Teal}0.804 & \color{Teal}0.465 & \color{Red}\textbf{1.42}\footnote{Computation time is reported in CPU seconds, with all calculations performed on a 3.5GHz Apple M2 Pro chip. Parallelization (5 cores) and local search are enabled by default for \texttt{missoNet}.} \\
  & [BIC] & (0.009) & (0) & (0.001) & (0.005) & (0.009) & (0) & (0.002) & (0.003) & (0.01) \\
  \cline{2-11}
  & \textit{MRCE} & 0.232 & 1 & 0.870 & 0.458 & 2.155 & 1 & \color{Red}\textbf{0.808} & \color{Red}\textbf{0.468} & 89.62 \\
  & [cv.min] & (0.004) & (0) & (0.003) & (0.005) & (0.179) & (0) & (0.017) & (0.016) & (0.27) \\
  \cline{2-11}
  & \textit{cglasso} & 0.321 & 1 & \color{Teal}0.955 & 0.677 & 0.884 & 1 & 0.711 & 0.376 & 18.80 \\
  & [BIC] & (0.008) & (0) & (0.001) & (0.005) & (0.014) & (0) & (0.004) & (0.003) & (0.05) \\
  \midrule

  10\% & \textit{Lasso} & 0.668 & 1 & 0.874 & 0.466 & & & & & 0.37 \\
  & [cv.min] & (0.008) & (0) & (0.003) & (0.004) & & & & & (0.00) \\
  \cline{2-6}
  & \textit{Lasso} & 2.215 & \color{Teal}0.989 & \color{Red}\textbf{0.991} & \color{Red}\textbf{0.900} & \\
  & [cv.1se] & (0.023) & (0.002) & (0.000) & (0.004) \\
  \cline{2-11}
  & \textit{missoNet} & \color{Red}\textbf{0.245} & 1 & 0.888 & 0.490 & 0.643 & 1 & \color{Red}\textbf{0.833} & \color{Red}\textbf{0.499} & 7.73 \\
  & [cv.min] & (0.004) & (0) & (0.003) & (0.005) & (0.032) & (0) & (0.009) & (0.013) & (0.03) \\
  \cline{2-10}
  & \textit{missoNet} & 0.806 & 1 & \color{Teal}0.987 & \color{Teal}0.870 & \color{Red}\textbf{0.750} & 1 & \color{Teal}0.826 & \color{Teal}0.489 &  \\
  & [cv.1se] & (0.023) & (0) & (0.001) & (0.005) & (0.035) & (0) & (0.009) & (0.013) \\
  \cline{2-11}
  & \textit{missoNet} & 0.406 & 1 & \color{Teal}0.962 & 0.708 & \color{Teal}0.773 & 1 & \color{Teal}0.798 & 0.457 & \color{Red}\textbf{1.55} \\
  & [BIC] & (0.010) & (0) & (0.001) & (0.005) & (0.011) & (0) & (0.003) & (0.003) & (0.01) \\
  \cline{2-11}
  & \textit{MRCE} & 0.452 & 1 & 0.873 & 0.466 & 2.964 & 1 & \color{Teal}0.799 & 0.457 & 80.16 \\
  & [cv.min] & (0.005) & (0) & (0.003) & (0.005) & (0.169) & (0) & (0.015) & (0.014) & (0.22) \\
  \cline{2-11}
  & \textit{cglasso} & 0.338 & 1 & \color{Teal}0.953 & 0.672 & 0.877 & 1 & 0.667 & 0.344 & 24.65 \\
  & [BIC] & (0.008) & (0) & (0.001) & (0.005) & (0.010) & (0) & (0.004) & (0.003) & (0.06) \\
  \midrule

  20\% & \textit{Lasso} & 0.984 & 1 & 0.872 & 0.462 & & & & & 0.37 \\
  & [cv.min] & (0.010) & (0.000) & (0.003) & (0.004) & & & & & (0.00) \\
  \cline{2-6}
  & \textit{Lasso} & 2.918 & \color{Teal}0.978 & \color{Red}\textbf{0.992} & \color{Red}\textbf{0.896} \\
  & [cv.1se] & (0.028) & (0.003) & (0.000) & (0.004) \\
  \cline{2-11}
  & \textit{missoNet} & \color{Red}\textbf{0.318} & 1 & 0.888 & 0.490 & \color{Red}\textbf{0.752} & 1 & \color{Red}\textbf{0.878} & \color{Red}\textbf{0.572} & 8.64 \\
  & [cv.min] & (0.006) & (0) & (0.003) & (0.005) & (0.027) & (0) & (0.005) & (0.008) & (0.03) \\
  \cline{2-10}
  & \textit{missoNet} & 0.888 & 1 & \color{Teal}0.984 & \color{Teal}0.852 & 0.854 & 1 & \color{Teal}0.871 & \color{Teal}0.559 &  \\
  & [cv.1se] & (0.025) & (0) & (0.001) & (0.006) & (0.033) & (0) & (0.005) & (0.009) \\
  \cline{2-11}
  & \textit{missoNet} & 0.522 & 1 & \color{Teal}0.958 & 0.687 & 1.204 & 1 & 0.790 & 0.453 & \color{Red}\textbf{1.70} \\
  & [BIC] & (0.015) & (0) & (0.001) & (0.005) & (0.021) & (0) & (0.003) & (0.003) & (0.01) \\
  \cline{2-11}
  & \textit{MRCE} & 0.812 & 1 & 0.874 & 0.466 & 4.194 & 1 & 0.812 & 0.471 & 73.30 \\
  & [cv.min] & (0.009) & (0) & (0.003) & (0.005) & (0.167) & (0) & (0.012) & (0.012) & (0.18) \\
  \cline{2-11}
  & \textit{cglasso} & 0.359 & 1 & \color{Teal}0.949 & 0.654 & 0.997 & 1 & 0.616 & 0.311 & 30.44 \\
  & [BIC] & (0.008) & (0) & (0.001) & (0.005) & (0.011) & (0) & (0.004) & (0.003) & (0.07) \\
  \bottomrule
\label{ch1:tab_model1}
\end{xltabular}

\renewcommand{\arraystretch}{0.85}
\begin{xltabular}{\textwidth}{@{}>{\centering\arraybackslash}p{1.5cm} >{\arraybackslash}p{1.50cm} >{\centering\arraybackslash}p{1cm} >{\centering\arraybackslash}p{1cm} >{\centering\arraybackslash}p{1cm} >{\centering\arraybackslash}p{1cm} >{\centering\arraybackslash}p{1cm} >{\centering\arraybackslash}p{1cm} >{\centering\arraybackslash}p{1cm} >{\centering\arraybackslash}p{1cm} >{\centering\arraybackslash}p{0.9cm}@{}}
  \caption{Means and standard errors (in parentheses) of evaluation metrics for different methods under three levels of missing data, based on 100 simulation replications using Model 2 for data generation.}\\
  \toprule
  Model 2 & & \multicolumn{4}{c}{Estimation of $\mathbf{B}^{*}$} & \multicolumn{4}{c}{Estimation of $\mathbf{\Theta}_{\varepsilon\varepsilon}^{*}$} & \\
  \cmidrule(lr){3-6} \cmidrule(lr){7-10}
  Missing & Method & PE & TPR & TNR & MCC & KLL & TPR & TNR & MCC & Time \\
  \midrule
  \endfirsthead

  \multicolumn{11}{c}%
  {{Table \thetable\ Continued from Previous Page}} \\
  \toprule
  Model 2 & & \multicolumn{4}{c}{Estimation of $\mathbf{B}^*$} & \multicolumn{4}{c}{Estimation of $\mathbf{\Theta}_{\varepsilon\varepsilon}^*$} & \\
  \cmidrule(lr){3-6} \cmidrule(lr){7-10}
  Missing & Method & PE & TPR & TNR & MCC & KLL & TPR & TNR & MCC & Time \\
  \midrule
  \endhead

  \midrule \multicolumn{11}{r}{{Continued on next page}} \\
  \endfoot

  \endlastfoot

  1\% & \textit{Lasso} & 1.377 & \color{Teal}0.998 & 0.948 & 0.391 & & & & & 0.87 \\
  & [cv.min] & (0.024) & (0.001) & (0.001) & (0.004) & & & & & (0.00) \\
  \cline{2-6}
  & \textit{Lasso} & 3.436 & 0.910 & \color{Teal}0.997 & 0.827\\
  & [cv.1se] & (0.040) & (0.006) & (0.000) & (0.007) \\
  \cline{2-11}
  & \textit{missoNet} & \color{Red}\textbf{0.463} & 1 & \color{Teal}0.965 & 0.468 & 4.572 & 1 & 0.822 & 0.365 & 76.71 \\
  & [cv.min] & (0.007) & (0) & (0.001) & (0.006) & (0.208) & (0) & (0.009) & (0.010) & (0.75) \\
  \cline{2-10}
  & \textit{missoNet} & 1.409 & \color{Teal}0.988 & \color{Red}\textbf{0.998} & \color{Red}\textbf{0.925} & 4.835 & 1 & 0.807 & 0.350 &  \\
  & [cv.1se] & (0.038) & (0.002) & (0.000) & (0.004) & (0.199) & (0) & (0.010) & (0.010) \\
  \cline{2-11}
  & \textit{missoNet} & 0.832 & 1 & \color{Teal}0.995 & 0.822 & \color{Red}\textbf{2.248} & 1 & \color{Red}\textbf{0.920} & \color{Red}\textbf{0.526} & \color{Red}\textbf{10.34} \\
  & [BIC] & (0.015) & (0) & (0.000) & (0.005) & (0.035) & (0) & (0.001) & (0.003) & (0.06) \\
  \cline{2-11}
  & \textit{MRCE} & 0.571 & \color{Teal}0.999 & \color{Teal}0.961 & 0.443 & 4.886 & 1 & \color{Teal}0.897 & 0.473 & 1616.33 \\
  & [cv.min] & (0.009) & (0.001) & (0.001) & (0.005) & (0.206) & (0) & (0.006) & (0.011) & (5.65) \\
  \cline{2-11}
  & \textit{cglasso} & 0.751 & \color{Teal}0.999 & \color{Teal}0.990 & 0.699 & 5.221 & 1 & 0.825 & 0.368 & 59.84 \\
  & [BIC] & (0.016) & (0.001) & (0.000) & (0.006) & (0.056) & (0) & (0.002) & (0.002) & (0.18) \\
  \midrule

  10\% & \textit{Lasso} & 1.678 & \color{Teal}0.995 & 0.947 & 0.388 & & & & & 0.86 \\
  & [cv.min] & (0.025) & (0.001) & (0.001) & (0.004) & & & & & (0.00) \\
  \cline{2-6}
  & \textit{Lasso} & 4.143 & 0.874 & \color{Teal}0.997 & 0.805 \\
  & [cv.1se] & (0.045) & (0.007) & (0.000) & (0.007) \\
  \cline{2-11}
  & \textit{missoNet} & \color{Red}\textbf{0.615} & 1 & \color{Teal}0.965 & 0.466 & 4.823 & 1 & \color{Teal}0.869 & 0.426 & 173.00 \\
  & [cv.min] & (0.009) & (0) & (0.001) & (0.005) & (0.210) & (0) & (0.006) & (0.009) & (0.88) \\
  \cline{2-10}
  & \textit{missoNet} & 1.643 & \color{Teal}0.990 & \color{Red}\textbf{0.998} & \color{Red}\textbf{0.904} & 5.237 & 1 & 0.858 & 0.409 &  \\
  & [cv.1se] & (0.045) & (0.002) & (0.000) & (0.005) & (0.220) & (0) & (0.007) & (0.009) \\
  \cline{2-11}
  & \textit{missoNet} & 1.052 & 1 & \color{Teal}0.994 & 0.792 & \color{Red}\textbf{3.205} & 1 & \color{Red}\textbf{0.917} & \color{Red}\textbf{0.519} & \color{Red}\textbf{21.65} \\
  & [BIC] & (0.022) & (0) & (0.000) & (0.006) & (0.049) & (0) & (0.001) & (0.003) & (0.51) \\
  \cline{2-11}
  & \textit{MRCE} & 0.984 & \color{Teal}0.998 & \color{Teal}0.960 & 0.442 & 6.983 & 1 & \color{Teal}0.899 & 0.478 & 1460.76 \\
  & [cv.min] & (0.014) & (0.001) & (0.001) & (0.004) & (0.183) & (0) & (0.006) & (0.009) & (5.14) \\
  \cline{2-11}
  & \textit{cglasso} & 0.820 & \color{Teal}0.998 & \color{Teal}0.989 & 0.689 & 5.328 & 1 & 0.805 & 0.348 & 66.26 \\
  & [BIC] & (0.015) & (0.001) & (0.000) & (0.007) & (0.051) & (0) & (0.002) & (0.002) & (0.17) \\
  \midrule

  20\% & \textit{Lasso} & 2.084 & \color{Teal}0.995 & 0.948 & 0.391 & & & & & 0.86 \\
  & [cv.min] & (0.022) & (0.001) & (0.001) & (0.004) & & & & & (0.00) \\
  \cline{2-6}
  & \textit{Lasso} & 4.947 & 0.818 & \color{Teal}0.997 & 0.782 \\
  & [cv.1se] & (0.045) & (0.008) & (0.000) & (0.007) \\
  \cline{2-11}
  & \textit{missoNet} & \color{Red}\textbf{0.798} & 1 & \color{Teal}0.964 & 0.457 & 4.748 & 1 & \color{Teal}0.910 & \color{Teal}0.503 & 213.50 \\
  & [cv.min] & (0.012) & (0) & (0.001) & (0.005) & (0.153) & (0) & (0.003) & (0.007) & (1.58) \\
  \cline{2-10}
  & \textit{missoNet} & 1.979 & \color{Teal}0.986 & \color{Red}\textbf{0.998} & \color{Red}\textbf{0.899} & 5.273 & 1 & \color{Teal}0.902 & 0.485 &  \\
  & [cv.1se] & (0.051) & (0.002) & (0.000) & (0.006) & (0.173) & (0) & (0.004) & (0.007) \\
  \cline{2-11}
  & \textit{missoNet} & 1.246 & 1 & \color{Teal}0.992 & 0.747 & \color{Red}\textbf{4.347} & 1 & \color{Red}\textbf{0.918} & \color{Red}\textbf{0.522} & \color{Red}\textbf{30.78} \\
  & [BIC] & (0.023) & (0) & (0.000) & (0.006) & (0.082) & (0) & (0.001) & (0.003) & (0.26) \\
  \cline{2-11}
  & \textit{MRCE} & 1.549 & \color{Teal}0.998 & \color{Teal}0.960 & 0.439 & 9.800 & 1 & \color{Teal}0.907 & \color{Teal}0.494 & 1336.81 \\
  & [cv.min] & (0.018) & (0.001) & (0.001) & (0.005) & (0.197) & (0) & (0.004) & (0.008) & (4.89) \\
  \cline{2-11}
  & \textit{cglasso} & 0.893 & \color{Teal}0.997 & \color{Teal}0.987 & 0.662 & 5.696 & 1 & 0.777 & 0.323 & 73.19 \\
  & [BIC] & (0.016) & (0.001) & (0.000) & (0.006) & (0.047) & (0) & (0.002) & (0.002) & (0.18) \\
  \bottomrule
\label{ch1:tab_model2}
\end{xltabular}

\renewcommand{\arraystretch}{0.85}
\begin{xltabular}{\textwidth}{@{}>{\centering\arraybackslash}p{1.5cm} >{\arraybackslash}p{1.50cm} >{\centering\arraybackslash}p{1cm} >{\centering\arraybackslash}p{1cm} >{\centering\arraybackslash}p{1cm} >{\centering\arraybackslash}p{1cm} >{\centering\arraybackslash}p{1cm} >{\centering\arraybackslash}p{1cm} >{\centering\arraybackslash}p{1cm} >{\centering\arraybackslash}p{1cm} >{\centering\arraybackslash}p{0.9cm}@{}}
  \caption{Means and standard errors (in parentheses) of evaluation metrics for different methods under three levels of missing data, based on 100 simulation replications using Model 3 for data generation.}\\
  \toprule
  Model 3 & & \multicolumn{4}{c}{Estimation of $\mathbf{B}^{*}$} & \multicolumn{4}{c}{Estimation of $\mathbf{\Theta}_{\varepsilon\varepsilon}^{*}$} & \\
  \cmidrule(lr){3-6} \cmidrule(lr){7-10}
  Missing & Method & PE & TPR & TNR & MCC & KLL & TPR & TNR & MCC & Time \\
  \midrule
  \endfirsthead

  \multicolumn{11}{c}%
  {{Table \thetable\ Continued from Previous Page}} \\
  \toprule
  Model 3 & & \multicolumn{4}{c}{Estimation of $\mathbf{B}^*$} & \multicolumn{4}{c}{Estimation of $\mathbf{\Theta}_{\varepsilon\varepsilon}^*$} & \\
  \cmidrule(lr){3-6} \cmidrule(lr){7-10}
  Missing & Method & PE & TPR & TNR & MCC & KLL & TPR & TNR & MCC & Time \\
  \midrule
  \endhead

  \midrule \multicolumn{11}{r}{{Continued on next page}} \\
  \endfoot

  \endlastfoot

  1\% & \textit{Lasso} & 0.971 & \color{Teal}0.997 & 0.870 & 0.460 & & & & & 0.35 \\
  & [cv.min] & (0.015) & (0.001) & (0.002) & (0.004) & & & & & (0.00) \\
  \cline{2-6}
  & \textit{Lasso} & 2.724 & 0.932 & \color{Red}\textbf{0.987} & \color{Red}\textbf{0.838} \\
  & [cv.1se] & (0.036) & (0.006) & (0.000) & (0.004) \\
  \cline{2-11}
  & \textit{missoNet} & \color{Red}\textbf{0.678} & 1 & 0.893 & 0.499 & 1.475 & 1 & 0.825 & 0.489 & 6.79 \\
  & [cv.min] & (0.011) & (0) & (0.002) & (0.005) & (0.064) & (0) & (0.013) & (0.018) & (0.03) \\
  \cline{2-10}
  & \textit{missoNet} & 1.289 & \color{Teal}0.971 & \color{Teal}0.974 & 0.762 & 1.498 & 1 & 0.809 & 0.468 &  \\
  & [cv.1se] & (0.024) & (0.004) & (0.001) & (0.006) & (0.059) & (0) & (0.013) & (0.017) \\
  \cline{2-11}
  & \textit{missoNet} & 1.225 & \color{Teal}0.969 & \color{Teal}0.976 & 0.769 & \color{Red}\textbf{0.659} & 1 & \color{Red}\textbf{0.976} & \color{Red}\textbf{0.852} & \color{Red}\textbf{0.99} \\
  & [BIC] & (0.023) & (0.004) & (0.001) & (0.005) & (0.015) & (0) & (0.001) & (0.005) & (0.01) \\
  \cline{2-11}
  & \textit{MRCE} & 0.796 & \color{Teal}0.997 & 0.877 & 0.470 & 1.203 & \color{Teal}0.987 & 0.853 & 0.527 & 71.00 \\
  & [cv.min] & (0.012) & (0.001) & (0.003) & (0.005) & (0.063) & (0.010) & (0.014) & (0.020) & (0.20) \\
  \cline{2-11}
  & \textit{cglasso} & 1.105 & \color{Teal}0.980 & \color{Teal}0.960 & 0.687 & 1.659 & 1 & 0.894 & 0.600 & 9.67 \\
  & [BIC] & (0.020) & (0.003) & (0.001) & (0.005) & (0.025) & (0) & (0.003) & (0.007) & (0.02) \\
  \midrule

  10\% & \textit{Lasso} & 1.206 & \color{Teal}0.996 & 0.872 & 0.462 & & & & & 0.35 \\
  & [cv.min] & (0.015) & (0.001) & (0.002) & (0.004) & & & & & (0.00) \\
  \cline{2-6}
  & \textit{Lasso} & 3.343 & 0.900 & \color{Red}\textbf{0.987} & \color{Red}\textbf{0.805} & \\
  & [cv.1se] & (0.042) & (0.006) & (0.000) & (0.005) \\
  \cline{2-11}
  & \textit{missoNet} & \color{Red}\textbf{0.777} & 1 & 0.894 & 0.504 & 1.501 & 1 & 0.879 & 0.572 & 7.41 \\
  & [cv.min] & (0.013) & (0) & (0.002) & (0.004) & (0.057) & (0) & (0.008) & (0.017) & (0.03) \\
  \cline{2-10}
  & \textit{missoNet} & 1.421 & \color{Teal}0.964 & \color{Teal}0.972 & 0.748 & 1.514 & 1 & 0.862 & 0.542 &  \\
  & [cv.1se] & (0.027) & (0.004) & (0.001) & (0.007) & (0.051) & (0) & (0.009) & (0.016) \\
  \cline{2-11}
  & \textit{missoNet} & 1.351 & \color{Teal}0.967 & \color{Teal}0.974 & 0.755 & \color{Red}\textbf{0.833} & 1 & \color{Red}\textbf{0.971} & \color{Red}\textbf{0.828} & \color{Red}\textbf{1.07} \\
  & [BIC] & (0.028) & (0.004) & (0.001) & (0.005) & (0.019) & (0) & (0.002) & (0.007) & (0.01) \\
  \cline{2-11}
  & \textit{MRCE} & 1.088 & \color{Teal}0.995 & 0.880 & 0.476 & 1.656 & \color{Teal}0.970 & 0.857 & 0.515 & 68.39 \\
  & [cv.min] & (0.013) & (0.001) & (0.002) & (0.004) & (0.081) & (0.015) & (0.014) & (0.019) & (0.18) \\
  \cline{2-11}
  & \textit{cglasso} & 1.158 & \color{Teal}0.980 & \color{Teal}0.958 & 0.671 & 1.682 & 1 & 0.869 & 0.555 & 12.72 \\
  & [BIC] & (0.021) & (0.003) & (0.001) & (0.005) & (0.028) & (0) & (0.004) & (0.007) & (0.03) \\
  \midrule

  20\% & \textit{Lasso} & 1.601 & \color{Teal}0.994 & 0.874 & 0.466 & & & & & 0.35 \\
  & [cv.min] & (0.018) & (0.002) & (0.002) & (0.004) & & & & & (0.00) \\
  \cline{2-6}
  & \textit{Lasso} & 4.136 & 0.859 & \color{Red}\textbf{0.988} & \color{Red}\textbf{0.798} \\
  & [cv.1se] & (0.042) & (0.008) & (0.000) & (0.006) \\
  \cline{2-11}
  & \textit{missoNet} & \color{Red}\textbf{0.941} & 1 & 0.898 & 0.511 & 1.778 & 1 & 0.896 & 0.605 & 8.37 \\
  & [cv.min] & (0.015) & (0) & (0.002) & (0.005) & (0.069) & (0) & (0.009) & (0.017) & (0.04) \\
  \cline{2-10}
  & \textit{missoNet} & 1.567 & \color{Teal}0.958 & \color{Teal}0.969 & 0.728 & 1.797 & 1 & 0.882 & 0.576 &  \\
  & [cv.1se] & (0.030) & (0.004) & (0.001) & (0.007) & (0.068) & (0) & (0.009) & (0.017) \\
  \cline{2-11}
  & \textit{missoNet} & 1.503 & \color{Teal}0.961 & \color{Teal}0.969 & 0.728 & \color{Red}\textbf{1.054} & \color{Teal}0.982 & \color{Red}\textbf{0.971} & \color{Red}\textbf{0.808} & \color{Red}\textbf{1.16} \\
  & [BIC] & (0.032) & (0.004) & (0.001) & (0.005) & (0.021) & (0.003) & (0.002) & (0.006) & (0.01) \\
  \cline{2-11}
  & \textit{MRCE} & 1.547 & \color{Teal}0.993 & 0.879 & 0.476 & 2.465 & 0.926 & 0.859 & 0.500 & 66.40 \\
  & [cv.min] & (0.014) & (0.001) & (0.003) & (0.004) & (0.106) & (0.022) & (0.012) & (0.019) & (0.18) \\
  \cline{2-11}
  & \textit{cglasso} & 1.239 & \color{Teal}0.977 & \color{Teal}0.951 & 0.647 & 1.781 & \color{Teal}0.999 & 0.841 & 0.509 & 15.74 \\
  & [BIC] & (0.024) & (0.003) & (0.002) & (0.006) & (0.028) & (0.001) & (0.004) & (0.007) & (0.03) \\
  \bottomrule
\label{ch1:tab_model3}
\end{xltabular}

\renewcommand{\arraystretch}{0.85}
\begin{xltabular}{\textwidth}{@{}>{\centering\arraybackslash}p{1.5cm} >{\arraybackslash}p{1.50cm} >{\centering\arraybackslash}p{1cm} >{\centering\arraybackslash}p{1cm} >{\centering\arraybackslash}p{1cm} >{\centering\arraybackslash}p{1cm} >{\centering\arraybackslash}p{1cm} >{\centering\arraybackslash}p{1cm} >{\centering\arraybackslash}p{1cm} >{\centering\arraybackslash}p{1cm} >{\centering\arraybackslash}p{0.9cm}@{}}
  \caption{Means and standard errors (in parentheses) of evaluation metrics for different methods under three levels of missing data, based on 100 simulation replications using Model 4 for data generation.}\\
  \toprule
  Model 4 & & \multicolumn{4}{c}{Estimation of $\mathbf{B}^{*}$} & \multicolumn{4}{c}{Estimation of $\mathbf{\Theta}_{\varepsilon\varepsilon}^{*}$} & \\
  \cmidrule(lr){3-6} \cmidrule(lr){7-10}
  Missing & Method & PE & TPR & TNR & MCC & KLL & TPR & TNR & MCC & Time \\
  \midrule
  \endfirsthead

  \multicolumn{11}{c}%
  {{Table \thetable\ Continued from Previous Page}} \\
  \toprule
  Model 4 & & \multicolumn{4}{c}{Estimation of $\mathbf{B}^*$} & \multicolumn{4}{c}{Estimation of $\mathbf{\Theta}_{\varepsilon\varepsilon}^*$} & \\
  \cmidrule(lr){3-6} \cmidrule(lr){7-10}
  Missing & Method & PE & TPR & TNR & MCC & KLL & TPR & TNR & MCC & Time \\
  \midrule
  \endhead

  \midrule \multicolumn{11}{r}{{Continued on next page}} \\
  \endfoot

  \endlastfoot

  1\% & \textit{Lasso} & 0.950 & \color{Teal}0.996 & 0.875 & 0.468 & & & & & 0.35 \\
  & [cv.min] & (0.018) & (0.001) & (0.003) & (0.005) & & & & & (0.00) \\
  \cline{2-6}
  & \textit{Lasso} & 2.710 & 0.939 & \color{Red}\textbf{0.987} & \color{Red}\textbf{0.828} \\
  & [cv.1se] & (0.036) & (0.004) & (0.001) & (0.005) \\
  \cline{2-11}
  & \textit{missoNet} & \color{Red}\textbf{0.468} & 1 & 0.884 & 0.484 & 1.331 & \color{Teal}0.960 & \color{Teal}0.904 & \color{Teal}0.799 & 8.16 \\
  & [cv.min] & (0.008) & (0) & (0.002) & (0.005) & (0.046) & (0.002) & (0.009) & (0.015) & (0.09) \\
  \cline{2-10}
  & \textit{missoNet} & 1.377 & \color{Teal}0.968 & \color{Teal}0.984 & \color{Teal}0.826 & 1.428 & \color{Teal}0.952 & \color{Teal}0.887 & 0.765 \\
  & [cv.1se] & (0.039) & (0.004) & (0.001) & (0.007) & (0.048) & (0.002) & (0.010) & (0.016) \\
  \cline{2-11}
  & \textit{missoNet} & 0.909 & \color{Teal}0.985 & \color{Teal}0.973 & 0.755 & \color{Red}\textbf{1.183} & \color{Red}\textbf{0.961} & \color{Red}\textbf{0.913} & \color{Red}\textbf{0.813} & \color{Red}\textbf{1.24} \\
  & [BIC] & (0.022) & (0.002) & (0.001) & (0.006) & (0.018) & (0.002) & (0.003) & (0.005) & (0.01) \\
  \cline{2-11}
  & \textit{MRCE} & 0.561 & \color{Teal}0.999 & 0.869 & 0.458 & 2.090 & \color{Teal}0.945 & \color{Teal}0.905 & \color{Teal}0.785 & 126.37 \\
  & [cv.min] & (0.009) & (0.001) & (0.003) & (0.005) & (0.095) & (0.002) & (0.014) & (0.017) & (0.69) \\
  \cline{2-11}
  & \textit{cglasso} & 0.727 & \color{Teal}0.995 & \color{Teal}0.947 & 0.649 & 1.564 & \color{Teal}0.951 & 0.794 & 0.647 & 14.89 \\
  & [BIC] & (0.013) & (0.001) & (0.001) & (0.005) & (0.020) & (0.002) & (0.004) & (0.005) & (0.08) \\
  \midrule

  10\% & \textit{Lasso} & 1.206 & \color{Teal}0.995 & 0.877 & 0.472 & & & & & 0.35 \\
  & [cv.min] & (0.018) & (0.002) & (0.003) & (0.005) & & & & & (0.00) \\
  \cline{2-6}
  & \textit{Lasso} & 3.351 & 0.904 & \color{Red}\textbf{0.987} & \color{Red}\textbf{0.821} & \\
  & [cv.1se] & (0.040) & (0.005) & (0.001) & (0.005) \\
  \cline{2-11}
  & \textit{missoNet} & \color{Red}\textbf{0.583} & 1 & 0.885 & 0.486 & 1.725 & \color{Teal}0.936 & \color{Red}\textbf{0.932} & \color{Red}\textbf{0.828} & 9.26 \\
  & [cv.min] & (0.010) & (0) & (0.002) & (0.004) & (0.050) & (0.003) & (0.007) & (0.014) & (0.08) \\
  \cline{2-10}
  & \textit{missoNet} & 1.538 & \color{Teal}0.962 & \color{Teal}0.980 & \color{Teal}0.800 & 1.690 & \color{Teal}0.927 & \color{Teal}0.920 & \color{Teal}0.801 \\
  & [cv.1se] & (0.042) & (0.004) & (0.001) & (0.007) & (0.043) & (0.003) & (0.007) & (0.014) \\
  \cline{2-11}
  & \textit{missoNet} & 1.042 & \color{Teal}0.985 & \color{Teal}0.965 & 0.710 & 1.678 & \color{Teal}0.928 & \color{Teal}0.899 & 0.768 & \color{Red}\textbf{1.37} \\
  & [BIC] & (0.024) & (0.002) & (0.001) & (0.006) & (0.025) & (0.003) & (0.003) & (0.006) & (0.01) \\
  \cline{2-11}
  & \textit{MRCE} & 0.856 & \color{Teal}0.997 & 0.872 & 0.464 & 2.685 & \color{Teal}0.929 & \color{Teal}0.890 & 0.755 & 108.28 \\
  & [cv.min] & (0.011) & (0.001) & (0.003) & (0.004) & (0.116) & (0.003) & (0.014) & (0.017) & (0.54) \\
  \cline{2-11}
  & \textit{cglasso} & 0.751 & \color{Teal}0.996 & 0.940 & 0.620 &  \color{Red}\textbf{1.509} & \color{Red}\textbf{0.942} & 0.758 & 0.601 & 18.27 \\
  & [BIC] & (0.014) & (0.001) & (0.002) & (0.006) & (0.019) & (0.002) & (0.005) & (0.005) & (0.08) \\
  \midrule

  20\% & \textit{Lasso} & 1.578 & \color{Teal}0.993 & 0.875 & 0.468 & & & & & 0.35 \\
  & [cv.min] & (0.018) & (0.002) & (0.003) & (0.005) & & & & & (0.00) \\
  \cline{2-6}
  & \textit{Lasso} & 4.104 & 0.865 & \color{Red}\textbf{0.988} & \color{Red}\textbf{0.798} \\
  & [cv.1se] & (0.040) & (0.007) & (0.001) & (0.006) \\
  \cline{2-11}
  & \textit{missoNet} & \color{Red}\textbf{0.723} & 1 & 0.892 & 0.499 & 2.558 & \color{Teal}0.903 & \color{Red}\textbf{0.950} & \color{Red}\textbf{0.831} & 33.37 \\
  & [cv.min] & (0.012) & (0) & (0.002) & (0.004) & (0.068) & (0.003) & (0.005) & (0.010) & (0.72) \\
  \cline{2-10}
  & \textit{missoNet} & 1.708 & \color{Teal}0.958 & \color{Teal}0.978 & \color{Teal}0.768 & 2.378 & \color{Teal}0.893 & \color{Teal}0.939 & \color{Teal}0.807 \\
  & [cv.1se] & (0.044) & (0.004) & (0.001) & (0.006) & (0.060) & (0.003) & (0.006) & (0.011) \\
  \cline{2-11}
  & \textit{missoNet} & 1.187 & \color{Teal}0.974 & \color{Teal}0.957 & 0.671 & 2.769 & \color{Teal}0.884 & 0.882 & 0.708 & \color{Red}\textbf{1.57} \\
  & [BIC] & (0.026) & (0.003) & (0.001) & (0.005) & (0.066) & (0.003) & (0.004) & (0.007) & (0.02) \\
  \cline{2-11}
  & \textit{MRCE} & 1.300 & \color{Teal}0.995 & 0.874 & 0.468 & 3.517 & \color{Teal}0.912 & 0.898 & 0.753 & 95.52 \\
  & [cv.min] & (0.014) & (0.001) & (0.003) & (0.004) & (0.107) & (0.003) & (0.013) & (0.015) & (0.43) \\
  \cline{2-11}
  & \textit{cglasso} & 0.831 & \color{Teal}0.994 & 0.939 & 0.620 & \color{Red}\textbf{1.510} & \color{Red}\textbf{0.934} & 0.710 & 0.545 & 21.55 \\
  & [BIC] & (0.014) & (0.001) & (0.002) & (0.005) & (0.016) & (0.002) & (0.005) & (0.006) & (0.08) \\
  \bottomrule
\label{ch1:tab_model4}
\end{xltabular}

\subsubsection{Metrics of evaluation \label{ME_missVother}}

We evaluated the estimated coefficient and precision matrices using several metrics, including Prediction Error (PE), Kullback-Liebler Loss (KLL), True Positive Rate (TPR), True Negative Rate (TNR), and Matthews Correlation Coefficient (MCC):
\[
    \text{PE}(\widehat{\mathbf{B}}, \mathbf{B}^{*}) = \text{Tr}[(\widehat{\mathbf{B}} - \mathbf{B}^{*})^\top \mathbf{\Sigma}^{*}_{xx} (\widehat{\mathbf{B}} - \mathbf{B}^{*})],
\]
\[
    \text{KLL}(\widehat{\mathbf{\Theta}}_{\varepsilon\varepsilon}, \mathbf{\Theta}_{\varepsilon\varepsilon}^{*}) = \text{Tr}[(\mathbf{\Theta}_{\varepsilon\varepsilon}^{*})^{-1}\widehat{\mathbf{\Theta}}_{\varepsilon\varepsilon}] - \text{logdet}[(\mathbf{\Theta}_{\varepsilon\varepsilon}^{*})^{-1}\widehat{\mathbf{\Theta}}_{\varepsilon\varepsilon}] - q,
\]
where $\mathbf{\Sigma}^*_{xx}$ represents the population covariance matrix of predictors. Additional metrics include:
\begin{align*}
\text{TPR} &= \frac{\text{TP}}{\text{TP}+\text{FN}},\indent\indent
\text{TNR} = \frac{\text{TN}}{\text{TN}+\text{FP}},\\
\text{MCC} &= \frac{\text{TP}\times \text{TN} - \text{FP}\times \text{FN}}{\sqrt{(\text{TP}+\text{FP})(\text{TP}+\text{FN})(\text{TN}+\text{FP})(\text{TN}+\text{FN})}},\nonumber
\end{align*}
where
\begin{align*}
    \text{TP} & = |\{(j,k):\mathbf{A}_{jk}\neq 0\ \text{and}\  \mathbf{A}^*_{jk}\neq 0\}|,\ \ 
    \text{TN} = |\{(j,k):\mathbf{A}_{jk}=0\ \text{and}\  \mathbf{A}^*_{jk}=0\}|,\\
    \text{FP} &= |\{(j,k):\mathbf{A}_{jk}\neq 0\ \text{and}\  \mathbf{A}^*_{jk}=0\}|,\ \ 
    \text{FN} = |\{(j,k):\mathbf{A}_{jk}=0\ \text{and}\  \mathbf{A}^*_{jk}\neq 0\}|,
\end{align*}
represent counts of correctly or incorrectly identified matrix elements for any matrix $\mathbf{A}$ and the underlying true matrix $\mathbf{A}^{*}$.

\subsubsection{Tuning parameter selection \label{Tuning_missVother}}

All methods used the same regularization paths and equivalent iteration tolerances to ensure fair comparisons. The tuning parameters for \textit{MRCE} and \textit{Lasso} were chosen via five-fold cross-validation with minimum error (cv.min), and the missing values are imputed by column means before training. For \textit{cglasso}, BIC was used as recommended in its documentation. For \texttt{missoNet}, both cross-validation and BIC were employed since both options are compatible with our approach, providing flexibility in selecting models based on predictive accuracy or parsimony. Additionally, cross-validation one-standard-error estimates (cv.1se) were reported for \texttt{missoNet} and \textit{Lasso}, which prioritize model simplicity while maintaining comparable accuracy.

\section{Supplementary Materials \label{sup_tab_S1}}

\noindent
\textbf{Supplementary Table S1}: Region-specific results obtained using cross-validation using the one-standard-error rule, following the aggregation of nearby ACPA-associated differentially methylated CpGs into candidate regions.\\
\textit{Sheet 1}: Summary of results across all regions.\\
\textit{Region 1--16}: Detailed cross-validation results for each of the 16 individual regions.\\

\noindent
The full file is provided in Excel format as \texttt{Table\_S1.xlsx}.

\bibliography{Paper}
\bibliographystyle{biom}

\end{document}